\documentclass[a4paper,11pt]{article}
\pdfoutput=1

\usepackage[height=23.6truecm,width=16truecm,top=3.05truecm]{geometry}

\usepackage[skip=6pt,indent]{parskip}

\usepackage[utf8]{inputenc}
\usepackage{graphicx}
\graphicspath{{./fig/}}
\usepackage{amsmath}
\usepackage{amssymb}
\usepackage{hyperref}
\hypersetup{colorlinks,citecolor=TokiwaIro,linkcolor=RuriIro,urlcolor=RuriIro,linktocpage}
\usepackage{cite}
\usepackage{braket}
\usepackage{bm}
\usepackage{bbm}
\usepackage{color}
\usepackage[svgnames,table]{xcolor}
\usepackage{float}
\usepackage{hhline}
\usepackage{multirow}
\usepackage{indentfirst}
\usepackage{latexsym}
\usepackage{mathtools}
\usepackage{cancel}
\usepackage{slashed}
\usepackage{colortbl}
\usepackage[mathscr]{euscript}
\usepackage{enumitem}
\usepackage{subcaption}
\usepackage{comment}
\usepackage{amsthm}

\usepackage{booktabs}
\usepackage{tabularx}
\newcolumntype{Y}{>{\raggedright\arraybackslash}X}

\usepackage{tikz}
\usetikzlibrary{positioning,intersections,calc,arrows.meta}
\usepackage{quantikz}

\numberwithin{equation}{section}

\allowdisplaybreaks

\definecolor{RuriIro}{HTML}{1E50A2}
\definecolor{TokiwaIro}{HTML}{007B43}

\definecolor{dred}{rgb}{0.7,0.15,0.09}
\definecolor{dblue}{rgb}{0,0.12,0.64}
\definecolor{dgreen}{rgb}{0.2,0.51,0.19}
\definecolor{pegn}{rgb}{0.33,0.51,0.14}

\definecolor{kblue}{rgb}{0,0.48,0.73}
\definecolor{kred}{rgb}{0.73,0.25,0}
\definecolor{kgreen}{rgb}{0.48,0.73,0}

\definecolor{rgreen}{HTML}{7BAA17}
\definecolor{rred}{HTML}{AB1732}
\definecolor{rblue}{HTML}{007BBB}

\theoremstyle{plain}
\newtheorem{thm}{Theorem}[section]

\theoremstyle{remark}

\newcommand{\nn}{\nonumber}

\newcommand{\mc}{\mathcal}
\newcommand{\mr}{\mathrm}
\newcommand{\ms}{\mathsf}

\newcommand{\mbb}{\mathbb}

\newcommand{\mf}{\mathfrak}

\newcommand{\del}{\partial}

\newcommand{\dd}{\mathrm{d}}
\newcommand{\ee}{\mathrm{e}}
\newcommand{\iu}{\mathrm{i}}

\newcommand{\TT}{\mathrm{T}}

\newcommand{\tr}{\mathop{\mathrm{tr}}\nolimits}

\newcommand{\ketbra}[2]{\ket{#1}\hspace{-4pt}\bra{#2}}

\newcommand{\QCHD}{\mathrm{QCHD}}

\newcommand{\eff}{\mathrm{eff}}

\newcommand{\geo}{\mathrm{geo}}
\newcommand{\pen}{\mathrm{pen}}

\newcommand{\bx}{\bm{x}}

\newcommand{\hp}{\hat{p}}

\begin{document}

\begin{titlepage}

\begin{flushright}
\end{flushright}

\vspace{1cm}

\begin{center}

{\LARGE \bfseries
Quantum Constrained Hamiltonian Descent
}

\vspace{1cm}

\renewcommand{\thefootnote}{\fnsymbol{footnote}}
{%
\hypersetup{linkcolor=black}
Yoshihiko Abe$^{1,2,3}$\footnote[1]{yabe3@keio.jp}
and
Ryo Nagai$^{3,4}$\footnote[2]{ryo.nagai.jd@hitachi.com}
}%
\vspace{8mm}

{\itshape%
$^1${Graduate School of Science and Technology, Keio University, Yokohama, Kanagawa 223-8522, Japan}\\
$^2${Keio University Sustainable Quantum Artificial Intelligence Center (KSQAIC), Keio University, Tokyo 108-8345, Japan}\\
$^3${Quantum Computing Center, Keio University, 3-14-1 Hiyoshi, Kohoku-ku, Yokohama, Kanagawa, 223-8522, Japan}\\
$^4${Research and Development Group,
Hitachi, Ltd., Kokubunji, Tokyo 185-8601, Japan
}
}%

\vspace{8mm}

\end{center}

\abstract{
We propose Quantum Constrained Hamiltonian Descent (QCHD), a continuous-time quantum framework for optimization under smooth equality constraints. 
The unconstrained counterpart, Quantum Hamiltonian Descent (QHD), minimizes a cost function by simulating the Schr\"odinger dynamics of a particle in the potential given by the cost function.
Applying Dirac's formalism to the classical dynamics with constraints, we construct an intrinsic quantum Hamiltonian with a Laplace--Beltrami operator on the constraint manifold in a kinetic term.
Its evolution preserves the constraints exactly without penalty terms and is expected to solve quantum-dynamically the constrained optimization problem.
We analyze its relation to Quantum Hamiltonian Descent supplemented by a strong penalty term. 
For a squared-distance penalty, the effective tangential dynamics acquires an extrinsic-curvature potential and, in nontrivial normal-rotation sectors, a coupling to the normal-bundle connection. 
These embedding-dependent terms remain in the strong-penalty limit, generally yielding quantum dynamics different from QCHD. 
Numerical comparisons for circular, elliptical, and two simultaneous constraints show faster optimization and improved concentration near global minima for the initial states and finite penalty strengths considered. 
}%

\end{titlepage}

\renewcommand{\thefootnote}{\arabic{footnote}}
\setcounter{footnote}{0}
\setcounter{page}{1}

\tableofcontents

\section{Introduction}
\label{sec:introduction}

Optimization problems arise in a wide range of fields, including engineering, machine learning, and finance.
A basic method for continuous optimization, the minimization of a real-valued cost function $V(\bx)$ of continuous parameters $\bx$,
is gradient descent, which repeatedly updates the parameters along the gradient of the cost function.
When the cost function has many valleys, however, gradient descent can be trapped in one of them (a local minimum) and never reaches the global minimum.
This trapping is a major bottleneck of continuous optimization.

A promising way around this bottleneck is an optimizer that exploits quantum effects.
A representative example is Quantum Hamiltonian Descent (QHD) proposed by Leng \textit{et al.}~\cite{2023-LengEtAl}.
QHD treats the parameters $\bx$ as the position of a particle and the cost function $V(\bx)$ as its potential energy, and searches for the minimum by simulating the quantum dynamics of the particle.
Quantum tunneling can facilitate transitions between local basins and may improve global exploration, although it does not guarantee convergence to a global minimum for a general nonconvex cost function.
Another advantage of QHD is that it is a coherent optimizer, which requires no intermediate measurements.
QHD has therefore attracted considerable attention in recent years~\cite{2023-LengEtAla,2024-KushnirEtAl,2025-LengEtAl,2025-LengShi,2025-ChakrabartiEtAl,2025-PengEtAl,2026-CatliEtAl}.
Moreover, the quantum dynamics of QHD is a time-dependent Hamiltonian simulation and can in principle be implemented on a quantum computer~\cite{2026-ChenEtAla,2026-ZhaoEtAl}, which makes QHD a promising application of quantum computing.

In practice, however, one often needs a constrained search rather than a plain minimization of the cost function:
the parameters $\bx$ have to satisfy equality constraints $f^A(\bx)=0$,
and one looks for the minimum of $V(\bx)$ among the $\bx$ that satisfy them.
Typical constraints are normalization conditions, such as a unit norm of the parameters.
In portfolio optimization in finance, for example, the investment weights are normalized to sum to one, and the risk of the portfolio is minimized under this constraint.

The question in this paper is whether QHD can solve constrained optimization problems.
A natural approach is to enforce the constraints using penalty terms.\footnote{QHD has been combined with the augmented Lagrangian method, which solves a sequence of unconstrained subproblems with QHD while updating the multipliers~\cite{2025-LiEtAla,2026-WuEtAl}.}
The simplest version of the penalty approach adds a term proportional to the constraint violation, such as $\kappa\sum_A(f^A)^2$, to $V$, thereby converting the problem approximately into an unconstrained one, and applies the ordinary QHD.
We call this approach QHD+Penalty and use it as the reference for comparison in this paper.
The penalty, however, is not unique.
Both the strength $\kappa$ and the functions $f^A$ representing a given constraint can be chosen in many ways, and they have to be tuned for each problem.
One would therefore like to have a coherent algorithm that requires no such tuning.

\begin{figure}[tbp]
  \centering
  \includegraphics[width=0.92\textwidth]{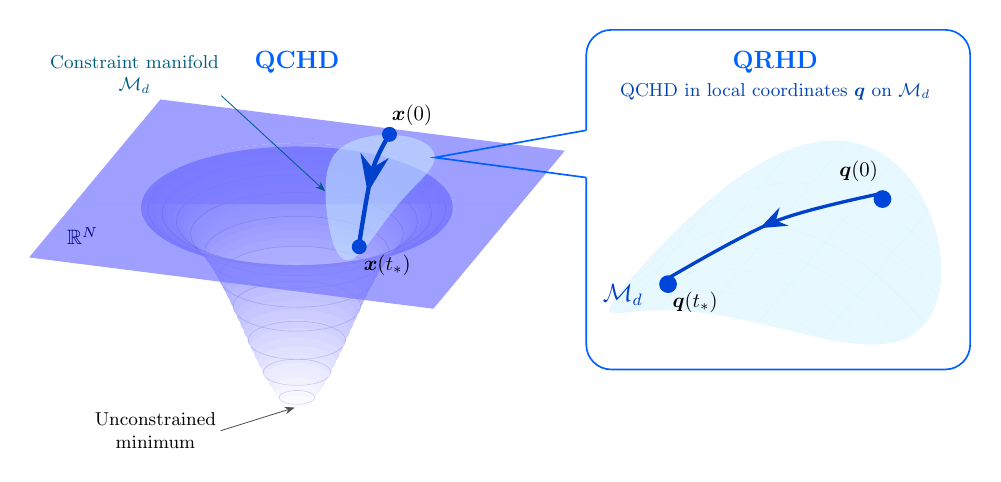}
  \caption{Schematic picture of QCHD.
    Left (QCHD): the wave function evolves toward the minimum of the potential $V$ within the constraint manifold $\mc{M}_d$ (light blue).
    Right (QRHD): the same motion described in local coordinates $\bm{q}$ on $\mc{M}_d$ corresponds to QRHD~\cite{2026-AbeNagai}.}
  \label{fig:concept}
\end{figure}

In this paper, we extend QHD to optimization under smooth equality constraints and propose Quantum Constrained Hamiltonian Descent (QCHD), a coherent quantum framework that requires no penalty term. 
The idea is that starting from a classical formulation with constraints, we apply Dirac's prescription~\cite{1950-Dirac-CJM,1958-Dirac-PRSLA,1964-Dirac} and adopt an intrinsic quantum Hamiltonian operator with a Laplace--Beltrami operator in a kinetic term determined by the induced metric on the constraint manifold. 
The resulting dynamics preserves the constraints exactly throughout the evolution. 
In local coordinates, this Hamiltonian coincides with that of Quantum Riemannian Hamiltonian Descent (QRHD)~\cite{2026-AbeNagai}, previously proposed by the authors for quantum optimization on Riemannian manifolds. 
Through constraint functions in the ambient space, QCHD is thus related to the intrinsic dynamics of QRHD on the constraint manifold directly which is the geometric framework of optimization on manifolds~\cite{1972-Luenberger-MS,1982-Gabay-JOTA,1998-EdelmanEtAl,2016-ZhangSra,2020-AlimisisEtAl,2022-DuruisseauxLeok,2023-FrancaEtAl}. 
Fig.~\ref{fig:concept} schematically illustrates the QCHD framework and its relation to QRHD.

We further analyze the relation between QCHD and QHD supplemented by a penalty term (QHD+Penalty). 
For a squared-distance penalty, we derive the leading effective tangential Hamiltonian in the strong penalty regime and identify an extrinsic-curvature potential and, in nontrivial normal-rotation sectors, a coupling to the normal-bundle connection. 
These embedding-dependent contributions remain in the strong-penalty limit, showing that intrinsic quantization (QCHD) and quantized system with a large penalty contribution (QHD+Penalty) generally lead to different quantum dynamics. 
More general penalty profiles can additionally produce position-dependent normal energy levels. 
In addition, numerical examples show faster optimization and improved concentration near constrained global minima for the initial states and finite penalty strengths.

This paper is organized as follows.
Sec.~\ref{sec:constrained} reviews the classical mechanics of constrained systems and its quantization, and derives the Hamiltonian operator of quantum mechanics on the constraint manifold.
In Sec.~\ref{sec:qchd}, we define QCHD based on the quantum mechanics developed in Sec.~\ref{sec:constrained}.
We discuss the relation between QCHD and QHD+Penalty, and also derive the effective Hamiltonian near the constraint manifold from the latter in Sec.~\ref{sec:penalty}.
Sec.~\ref{sec:numerics} compares QCHD with QHD+Penalty numerically in some examples.
Sec.~\ref{sec:conclusions} is devoted to our conclusions.

\paragraph{Notation.}

Throughout the paper we work in units with $\hbar=1$ and, unless stated otherwise, use the Einstein summation convention, in which repeated indices are summed over and the symbol $\sum$ is dropped:
\begin{align}
    \sum_{i} A_i B^i \to A_i B^i\,,
    \qquad
    \sum_{i,j} g_{ij} A^i B^j \to g_{ij} A^i B^j\,.
\end{align}
Boldface symbols denote the collection of all components of a vector or of coordinates, such as $\bm{x}=\{x^i\}_{i=1}^N$, $\bm{q}=\{q^\alpha\}_{\alpha = 1}^{d=N-M}$, and $\bm{\xi}=\{\xi^a\}_{a=1}^M$.
Our notation for coordinates, indices, and metrics are summarized in Tab.~\ref{tab:coordinates-tab} of App.~\ref{app:geometry} and the schematic picture is found in Fig.~\ref{fig:indices} of Sec.~\ref{sec:constrained}.

\paragraph{Related works.}

Related ideas have appeared recently.
Gribling \textit{et al.}~\cite{2025-GriblingEtAl} proposed a quantum interior-point method for convex optimization,
a quantum annealing algorithm that follows the ground state of a Schr\"odinger operator along the central path.
The Schr\"odinger operator is one in which the Laplacian is replaced by the Laplace--Beltrami operator of the Hessian metric of a self-concordant barrier,
and they proved lower bounds on the spectral gaps of these operators, on which the construction relies.
Their algorithm therefore differs from QHD, QRHD, and QCHD, which simulate the real-time dynamics generated by a time-dependent Hamiltonian.

Herman \textit{et al.}~\cite{2026-HermanEtAl} studied the quantum simulation of the Schr\"odinger equation
whose kinetic term is the Laplace--Beltrami operator of a Riemannian manifold.
They gave explicit algorithms with rigorous error and resource bounds for homogeneous manifolds such as spheres and simplices (a global spectral method),
and a coherent-chart method that is in principle applicable to any compact manifold under conditions on its coordinate charts.
They also proved convergence guarantees for geodesically convex cost functions on geodesically convex domains of these manifolds,
and noted that equality constraints can be imposed by replacing the Laplacian with the Laplace--Beltrami operator of the feasible manifold.
The present work differs from these works in its starting point and in the question asked.
They take the manifold, with its metric or coordinate charts, as given and ask how to simulate the Schr\"odinger equation on it efficiently.
We start from the constraint functions $f^A(\bx)=0$ in the ambient space $\mathbb{R}^N$ and construct an intrinsic quantum dynamics from explicit constraint functions in the ambient space.
The quantization of the constrained classical dynamics yields the Laplace--Beltrami operator in ambient coordinates, without introducing coordinate charts on the manifold.

\section{Classical and quantum mechanics of constrained systems}
\label{sec:constrained}

This section reviews the classical and quantum mechanics of a constrained system, which is a basis of QCHD.
Sec.~\ref{sec:classical} formulates the canonical formulation of the classical mechanics with constraints using Lagrange multipliers.
In Sec.~\ref{sec:quantization}, we quantize the constrained system and define the Hamiltonian operator on the submanifold characterized by the constraints.

\subsection{Classical mechanics of constrained systems}
\label{sec:classical}

We consider a point particle of mass $m$ in the $N$-dimensional Euclidean space $\mbb{R}^N$ moving in a potential $V(\bm{x})$, where $\bm{x} = \{x^i\}_{i=1}^N \in \mbb{R}^N$ is the position of the particle.
In the application to optimization in Sec.~\ref{sec:qchd}, $\bm{x}$ is regarded as the parameters to be optimized and $V(\bm{x})$ as the cost function.
Suppose that the motion of this particle is subject to $M (< N)$ independent constraints
\begin{align}
    f^A(\bm{x}) = 0\,,
    \qquad
    A= 1,\ldots, M\,.
    \label{eq:constraints}
\end{align}
For simplicity, we assume that the constraints are holonomic that is independent of the particle velocity $\dot{\bm{x}}$.\footnote{
The case with velocity-dependent constraints are discussed in Ref.~\cite{1996-NojiriEtAl}.}
The constraints \eqref{eq:constraints} define a $d =(N-M)$-dimensional submanifold of $\mbb{R}^N$ as 
\begin{align}
    \mc{M}_d = \{ \bm{x} \in \mbb{R}^N | f^A(\bm{x}) =0 \}\,,
\end{align}
and we call this $\mc{M}_d$ the constraint manifold.
We assume that $f^A(\bm{x})$ are smooth and that $\mathop{\mathrm{rank}}(\del_i f^A) =M$ in a neighborhood of the constraint manifold.
According to Dirac's prescription~\cite{1950-Dirac-CJM,1958-Dirac-PRSLA,1964-Dirac}, we consider the classical mechanics of this constrained system.

Let us begin with the Lagrangian formalism of this system.
With the subsequent quantization process in mind, we consider the Lagrangian following Refs.~\cite{1990-HommaEtAl-ZPC-PaF,1990-HommaEtAl-PRD,1996-NojiriEtAl}
\begin{align}
    L = \frac{m}{2} \delta_{ij} \dot{x}^i \dot{x}^j - V(\bm{x}) + \lambda_A \dot{f}^A(\bm{x})\,,
    \label{eq:classical-L}
\end{align}
where $\lambda_A$ is a Lagrange multiplier including the constraint conditions into the Lagrangian.\footnote{
Since the constraints are functions of $\bm{x}$, the Lagrangian \eqref{eq:classical-L} can also be written as
$L = \frac{m}{2} \delta_{ij} \dot{x}^i \dot{x}^j - V(\bm{x}) + (\lambda_A \del_i f^A(\bm{x})) \dot{x}^i$,
in which the last term has the same form as the coupling to a gauge potential.
The role of this term is, however, different from that of a gauge interaction.
Rather than acting as a magnetic field, $\lambda_A \del_i f^A$ restricts the particle dynamics to the tangent directions of the constraint manifold.
}
The constraint $\dot{f}^A(\bm{x}) = 0$ is imposed instead of \eqref{eq:constraints}, which turns out to simplify the canonical structure of the system~\cite{1990-HommaEtAl-ZPC-PaF,1990-HommaEtAl-PRD,1996-NojiriEtAl}.
We note that the dynamics described by the Lagrangian \eqref{eq:classical-L} is equivalent to that with constraints \eqref{eq:constraints} when the initial state satisfies $f^A=0$.
This point reappears in Sec.~\ref{sec:qchd} as a condition on the initial state.

\paragraph{Canonical formalism.}

In order to move onto the canonical formalism, let us introduce the canonical conjugate momentum
\begin{align}
    p_i \coloneqq \frac{\del L}{\del \dot{x}^i} = m \delta_{ij} \dot{x}^j + \lambda_A \del_i f^A\,,
    \label{eq:classical-p}
\end{align}
where $\del_i = \del / \del x^i$.
Since the Lagrangian \eqref{eq:classical-L} does not contain $\dot{\lambda}_A$, the conjugate momentum of the Lagrange multiplier identically vanishes 
\begin{align}
    \pi^A_\lambda \coloneqq \frac{\del L}{\del \dot{\lambda}_A} = 0\,,
    \label{eq:primary-constraint}
\end{align}
and the velocities $\dot{\lambda}_A$ cannot be expressed in terms of the conjugate momenta $\pi^A_\lambda$.
This kind of system is called a singular system or constrained system in Dirac's terminology~\cite{1950-Dirac-CJM,1958-Dirac-PRSLA,1964-Dirac}.
Once the constraint like \eqref{eq:primary-constraint} appear, they require additional constraint conditions so that the constraints are preserved in the time evolution.
We write these constraint condition in the phase space as $\phi_{\ms{A}}$, $\ms{A} = 1,\ldots ,\ms{M} < 2M$, and introduce Dirac's weak equality by $\approx$ that is an equality that holds only when the constraints are satisfied $\phi_{\ms{A}} \approx 0$.

Now we consider the canonical structure of the constrained system.
When we have constraints, the Poisson bracket $\{\bullet,\bullet\}_P$ does not preserve the constraints $\phi_{\ms{A}} \approx 0$ and does not give the canonical structure compatible with the constraints.
Namely, $\{A, \phi_\ms{A}\}_P\neq 0$ for a general function $A$ of the canonical variables.
To overcome this difficulty,
Dirac proposed the so-called Dirac bracket $\{\bullet,\bullet\}_D$~\cite{1950-Dirac-CJM,1958-Dirac-PRSLA,1964-Dirac} by 
\begin{align}
    \{A, B\}_D \coloneqq \{A, B\}_P - \{A, \phi_{\ms{A}} \}_P (C^{-1})^{\ms{AB}} \{\phi_\ms{B}, B\}_P,
    \qquad 
    C_{\ms{AB}} \coloneqq \{\phi_\ms{A}, \phi_{\ms{B}}\}_P\,.
    \label{eq:dirac-bracket-def}
\end{align}
The matrix $C_{\ms{AB}}$ is assumed to be invertible here.\footnote{
The constraints satisfying this invertible condition are called second-class.
}
For this Dirac bracket, for any $A$, it is found that 
\begin{align}
    \{A, \phi_\ms{A}\}_D = \{A,\phi_\ms{A}\}_P - \{A,\phi_{\ms{C}}\}_P (C^{-1})^{\ms{CB}}\{\phi_\ms{B},\phi_{\ms{A}}\}_P 
    =\{A,\phi_\ms{A}\}_P - \{A, \phi_\ms{A}\}_P = 0\,,
\end{align}
and the constraints $\phi_\ms{A}=0$ can be imposed consistently either before or after evaluating the brackets.
In other words, the Dirac bracket defines the Poisson structure on the constraint manifold compatible with the second-class constraints.

For the system~\eqref{eq:classical-L}, we find the two types of the constraints appear
\begin{align}
    \phi^A_1 = \pi^A_\lambda \approx 0\,,
    \qquad 
    \phi^A_2 = \del^i f^A (p_i - \lambda_B \del_i f^B) \approx 0\,,
    \label{eq:Dirac-constraints}
\end{align}
and the Dirac brackets of the dynamical variables $(x^i,p_j)$ on the phase space satisfy
\begin{align}
    \{x^i,p_j\}_D = \delta^i{}_j\,,
    \qquad
    \{x^i,x^j\}_D = \{p_i,p_j\}_D =0\,.
    \label{eq:Dirac-brakect-xp}
\end{align}
For detailed calculations, see App.~\ref{app:hamiltonian}.
Eq.~\eqref{eq:Dirac-brakect-xp} implies that the Dirac brackets of $(x^i, p_j)$ are simple and canonical, which is convenient for the quantization discussed in the next subsection.
It is noted that this simple canonical structure is the reason that we impose the constraints $\dot{f}^A(\bm{x})=0$.\footnote{
If we consider the Lagrangian 
$L' = \frac{m}{2} \delta_{ij} \dot{x}^i \dot{x}^j - V(\bm{x}) + \lambda'_A f^A$ instead of \eqref{eq:classical-L}, the Dirac brackets of $(x^i,p_j)$ are not \eqref{eq:Dirac-brakect-xp} and $\{p_i,p_j\}_D \neq 0$~\cite{1990-HommaEtAl-PRD}.
}

\paragraph{Hamiltonian.}

We then discuss the Hamiltonian of the constrained system.
After some calculation shown in App.~\ref{app:hamiltonian}, we find that the Hamiltonian is written as 
\begin{align}
    H \approx \frac{1}{2m} M^{ij} p_i p_j + V(\bm{x})\big|_{\mc{M}_d}\,,
    \label{eq:H-classical}
\end{align}
where $V(\bm{x})\bigl|_{\mc{M}_d}$ is the potential restricted to the constraint manifold.
$M^{ij}$ is the projector onto the tangent space of $\mc{M}_d$, which is the parameter space compatible with the constraints.
Concretely, the projector is defined as follows. 
We first introduce the normal vector of $\mc{M}_d$
\begin{align}
    n^A{}_i \coloneqq \del_i f^A\,,
    \label{eq:normal}
\end{align}
and
the $M \times M$ Gram matrix as
\begin{align}
    \mr{Gram}^{AB} = n^A{}_i n^B{}_j \delta^{ij} = n^A{}_i n^{Bi}\,.
\end{align}
Following the general relativity and gravitational theory convention, we write the inverse of this Gram matrix 
\begin{align}
    G_{AB} = (\mr{Gram}^{-1})_{AB}\,,
    \qquad 
    G^{AB} = \mr{Gram}^{AB}\,,
    \label{eq:Gram-matrix}
\end{align}
like a metric tensor.
Due to the independence of the constraints, $\det G^{AB} \neq 0$ and we have the inverse of the Gram matrix $G_{AB}$.
Using these normal vectors and the inverse Gram matrix, the projector is given by 
\begin{align}
    M^i{}_j \coloneqq \delta^i_j - n^{Ai} G_{AB} n^B{}_j\,.
    \label{eq:projector}
\end{align}
$M^{ij}$ is given by $M^{ij}=\delta^{jk}M^i{}_k$.
From the definition \eqref{eq:projector}, the projector satisfies 
\begin{align}
    n^A{}_j M^j{}_i = 0\,,
    \qquad 
    M^i{}_k M^k{}_j = M^i{}_j\,.
    \label{eq:projection-property}
\end{align}
The first equation states that the tangent space is orthogonal to the normals, and the second is idempotency of the projector.

We note that the Hamiltonian~\eqref{eq:H-classical} is written as
\begin{align}
    H \approx \frac{P^2}{2m}  + V(\bm{x}) \big|_{\mc{M}_d}\,,
    \label{eq:H-classical-2}
\end{align}
where $P^2$ denotes the squared tangential momentum, $P_i \coloneqq M_i{}^j p_j$, as
\begin{align}
    P^2 = \delta^{ij} P_i P_j = M^{ij} p_i p_j\,,
\end{align}
and the second equality follows from the idempotency \eqref{eq:projection-property}.
$P_i$ is the projection of the momentum $p_i$ onto the directions compatible with the constraints, and the kinetic energy of the constrained particle is determined by this tangential component alone.
We also note that $P_i$ is not a canonical variable. Its Poisson and Dirac brackets with $x^i$ are
\begin{align}
    \{x^i, P_j\}_P = \{x^i,P_j\}_D = M^i{}_j\,,
    \label{eq:Dirac-bracket-xP}
\end{align}
and $\{P_i,P_j\}_D \neq 0$ [App.~\ref{app:hamiltonian}, Eq.~\eqref{eq:PP-dirac-app}].
The quantum counterparts of these brackets are discussed in the remark following Eq.~\eqref{commutation-relation-xp} in Sec.~\ref{sec:quantization}.

\subsection{Quantization of constrained system}
\label{sec:quantization}

We now quantize the system introduced in the previous subsection.
For an unconstrained system, quantization is formally given by replacing the Poisson bracket with the commutator.
For the constrained system, the Dirac bracket plays this role: we replace the Dirac bracket with the commutator.

\paragraph{Canonical commutation relations.}

Thanks to the terms of $\lambda_A \dot{f}^A$ in the Lagrangian~\eqref{eq:classical-L}, the Dirac brackets of $(x^i,p_j)$ are simple canonical form as in Eq.~\eqref{eq:Dirac-brakect-xp}.
The commutation relations of the operators $(\hat{x}^i, \hat{p}_j)$ then become\footnote{
If the Dirac brackets contained products of the canonical variables, such as ${\bm{x}}\cdot{\bm{p}}$, the operator-ordering ambiguity of quantization would appear. 
Since the Dirac brackets~\eqref{eq:Dirac-brakect-xp} are of the canonical form, their quantization~\eqref{commutation-relation-xp} involves no such ambiguity.
The kinetic energy~\eqref{eq:H-classical}, however, contains products of $M^{ij}(\bm{x})$ and $p_i$, and its quantization does require a choice of ordering, which we discuss in the paragraph ``Ambiguity of quantization'' at the end of this section.
}
\begin{align}
    [\hat{x}^i, \hat{p}_j] = \iu \delta^i_j,
    \qquad 
    [\hat{x}^i, \hat{x}^j] = [\hat{p}_i, \hat{p}_j] = 0\,.
    \label{commutation-relation-xp}
\end{align}
Note that $[f^A(\hat{\bx}),\hp_j]=\iu n^A{}_j\neq0$, which means that the canonical momentum $\hp_j$ generates translations off the constraint surface.
On the other hand, the tangential momentum $\hat{P}_i \coloneqq M_i{}^j(\hat{\bx})\hp_j$ introduced below satisfies $[f^A(\hat{\bx}),\hat{P}_i]=0$ by Eq.~\eqref{eq:projection-property}.
Its commutator with $\hat{x}^i$, however, is not canonical, as expected from the Dirac bracket~\eqref{eq:Dirac-bracket-xP}, and its explicit form is derived below.\footnote{
Eq.~\eqref{commutation-relation-xp} is understood as a relation among operators acting on smooth functions in a neighborhood of $\mc{M}_d$ in $\mbb{R}^N$.
The operators $\hat{p}_i = -\iu \del_i$ are not, by themselves, operators on the physical Hilbert space and generally differentiate off the manifold. 
In contrast, the tangential operators $\hat{P}_i = M_i{}^j\hat{p}_j$
and the Hamiltonian constructed below act intrinsically on $\mc{M}_d$, independently of how a wave function is extended away from it.}

\paragraph{Position eigenstates.}

Next, let us consider the states and the Hilbert space of the quantum constrained system to write the Schr\"odinger equation in the position basis.
We first introduce the position eigenstate $\ket{\bm{x}}$ satisfying 
\begin{align}
    \hat{x}^i \ket{\bm{x}} = x^i \ket{\bm{x}}\,.
\end{align}
We use the generalized position eigenstates $\ket{\bm{x}}$, labeled by $\bm{x} \in\mc{M}_d$.
The completeness relation then takes the form
\begin{align}
    \hat{I} = \int \dd^N \bm{x} \prod_A \delta(f^A) \frac{1}{\sqrt{\det G}} \ketbra{\bm{x}}{\bm{x}}\,.
    \label{eq:completness-relation-x}
\end{align}
Here, we write $\det G_{AB} = \det G$.
We call the Hilbert space spanned by the generalized position eigenstates, which satisfy the completeness relation~\eqref{eq:completness-relation-x}, the physical Hilbert space.
The integral measure $\dd^N \bm{x} \prod_A \delta(f^A) /\sqrt{\det G}$ is the induced volume element of the constrained manifold $\mc{M}_d$ in the ambient space $\mbb{R}^N$.
The factor $1/\sqrt{\det G}$ is required from the covariance under the transformation of the delta function $\delta (f^A(\bm{x}))$.
Accordingly, the normalization of the position eigenstate takes the following form
\begin{align}
    \braket{\bm{x}|\bm{x}'} = \delta_{\mc{M}_d}(\bm{x}, \bm{x}')\,,
\end{align}
where $\delta_{\mc{M}_d}$ is the delta function on $\mc{M}_d$ with respect to the induced volume element.
This delta function satisfies
\begin{align}
    \int \dd^N\bm{x} \prod_A \delta(f^A) \frac{1}{\sqrt{\det G}} \delta_{\mc{M}_d}(\bm{x}, \bm{x}') \varphi(\bm{x}) = \varphi(\bm{x}')\,.
\end{align}

\begin{figure}[t]
    \centering
    \includegraphics[width=0.55\textwidth]{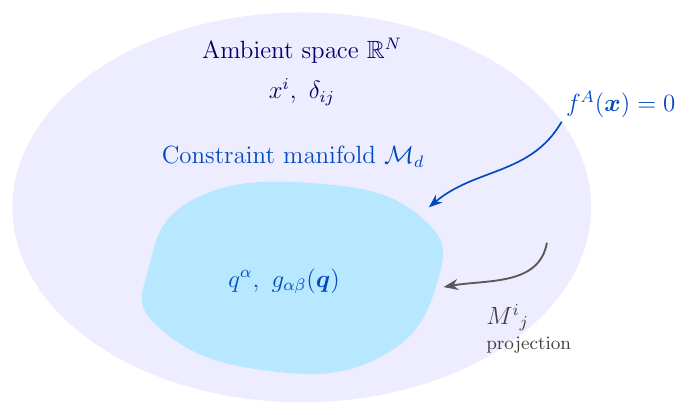}
    \caption{
    Schematic picture for our geometric setting.
    The ambient space $\mathbb{R}^N$ has Cartesian coordinates $\{x^i\}_{i=1}^N$ and the flat metric $\delta_{ij}$.
    The constraint manifold $\mc{M}_d$ is the zero set of the constraint functions $\{f^A(\bx)\}_{A=1}^M$.
    This subspace is parametrized by local coordinates $\{q^\alpha\}_{\alpha=1}^{d=N-M}$ and carries the induced metric $g_{\alpha\beta}$.
    The projector $M^i{}_j$ in Eq.~\eqref{eq:projector} maps vectors of $\mathbb{R}^N$ onto the tangent space of $\mc{M}_d$.
    }
    \label{fig:indices}
\end{figure}

Before moving to the wavefunction, we give some comments on the integral measure and the delta function discussed above.
If we introduce local coordinates $\bm{q} = \{q^\alpha\}_{\alpha=1}^d$ on $\mc{M}_d$ as in Fig.~\ref{fig:indices} and write the embedding of $\mc{M}_d$ as
\begin{align}
    x^i = X^i(\bm{q})\,,
    \label{eq:emb-Xq}
\end{align}
the integral measure in the ambient space is rewritten as 
\begin{align}
    \int \dd^N\bm{x} \prod_A\delta(f^A)\frac{1}{\sqrt{\det G}} = \int_{\mc{M}_d} \dd^d \bm{q} \sqrt{\det g}\,,
\end{align}
where $g_{\alpha\beta} = \delta_{ij}\,\del_\alpha X^i \del_\beta X^j$ is the metric induced on $\mc{M}_d$ and $\det g = \det g_{\alpha\beta}$ is its determinant.
The right-hand side is the covariant volume element of the curved space $\mc{M}_d$.
In the same coordinates, the delta function is written as 
\begin{align}
    \delta_{\mc{M}_d}(\bm{x}, \bm{x}') = g(\bm{q})^{-1/4} \delta^d(\bm{q}-\bm{q}') g(\bm{q}')^{-1/4}\,,
    \label{eq:deltafunction-rep}
\end{align}
which agrees with the delta function used in Ref.~\cite{2026-AbeNagai}.
The detailed calculations are found in App.~\ref{app:geometry}.

\paragraph{Wave function.}

The wave function in the position basis is 
\begin{align}
    \psi(\bm{x},t) = \braket{\bm{x}|\psi(t)}\,.
\end{align}
Using the completeness relation~\eqref{eq:completness-relation-x}, the inner product of states or wave functions is expressed as 
\begin{align}
    \braket{\psi_1|\psi_2} = \int \dd^N\bm{x} \prod_A \delta(f^A) \frac{1}{\sqrt{\det G}} \psi_1^\ast(\bm{x}) \psi_2(\bm{x})\,.
    \label{eq:inner-product-ambient}
\end{align}
We call the state satisfying Eq.~\eqref{eq:inner-product-ambient} the physical state.
If we use the local coordinate, this inner product reduces to 
\begin{align}
    \braket{\psi_1|\psi_2} = \int_{\mc{M}_d} \dd^d \bm{q} \sqrt{\det g} \psi^\ast_1(\bm{q}) \psi_2(\bm{q})\,,
    \label{eq:inner-product}
\end{align}
which is the same form introduced in Ref.~\cite{2026-AbeNagai}.\footnote{
The inner products \eqref{eq:inner-product-ambient} and \eqref{eq:inner-product} lead that a physical state $\ket{\psi}$ is therefore represented by a wave function on $\mc{M}_d$, where the position operators act as $(\hat{x}^i\psi)(\bm{q})=X^i(\bm{q})\psi(\bm{q})$.
Consequently, every physical state satisfies $f^A(\hat{\bm{x}}) \ket{\psi} = 0$.
This condition follows from the definition of the physical Hilbert space, rather than being imposed as an additional Dirac constraint. 
In particular, the second-class constraints are those in Eq.~\eqref{eq:Dirac-constraints}, whereas $f^A=0$ specifies the initial constraint sector preserved by the dynamics. 
}

\paragraph{Hamiltonian operator.}

From the commutation relation \eqref{commutation-relation-xp}, a simple choice of the conjugate momentum operator is\footnote{%
The commutation relations (\ref{commutation-relation-xp}) are satisfied by
$\hat{p}_i= -\iu \del_i + \mc{F}_i(\bm{x})$
for any $\mc{F}_i(\bm{x})$ with
$\del_i \mc{F}_j -\del_j \mc{F}_i = 0$.
Different choices of $\mc{F}_i$ correspond to different prescriptions of quantization.
If $\mc{F}_i$ is normal to $\mc{M}_d$, $M^{ij}\mc{F}_j = 0$, the tangential momentum operator \eqref{eq:hatP_i} is unchanged, so that Eqs.~\eqref{eq:commutation-relation-xP} and \eqref{eq:hatP^2-Delta_M} and the discussion below hold as they stand.
}
\begin{align}
    \hat{p}_i = -\iu \del_i\,.
    \label{eq:p-rep}
\end{align}
The tangential momentum operator then becomes
\begin{align}
    \hat{P}_i = M^j{}_i(\hat{\bm{x}}) \hat{p}_j = -\iu M^j{}_i(\bm{x}) \del_j \,.
    \label{eq:hatP_i}
\end{align}
This derivative operator satisfies
\begin{align}
    [\hat{x}^i,\hat{P}_j] = \iu M^i{}_j(\hat{\bm{x}})\,,
    \label{eq:commutation-relation-xP}
\end{align}
which is the quantum counterpart of the Dirac bracket~\eqref{eq:Dirac-bracket-xP}.
Its square is evaluated as
\begin{align}
    \hat{P}^2 \coloneqq \delta^{ij} \hat{P}_i \hat{P}_j = -\Delta_{\mc{M}_d}\,,
    \label{eq:hatP^2-Delta_M}
\end{align}%
where 
\begin{align}
    \Delta_{\mc{M}_d} \coloneqq \sqrt{\det G} \del_i \left(  \frac{M^{ij}}{\sqrt{\det G}} \del_j \right)\,.
    \label{eq:LB-x}
\end{align}
As shown below, $\Delta_{\mc{M}_d}$ is the Laplace--Beltrami operator on $\mc{M}_d$ written in the ambient coordinates $\bm{x}$.
We thus find that the Hamiltonian operator obtained by quantizing \eqref{eq:H-classical} is
\begin{align}
    \hat{H} = \frac{\hat{P}^2}{2m}  + V(\hat{\bm{x}})\big|_{\mc{M}_d} = - \frac{\Delta_{\mc{M}_d}}{2m}  + V(\bm{x})\big|_{\mc{M}_d}\,,
    \label{eq:H-quantum}
\end{align}%
which is the main result of this section.

It should be noted that 
the quantum Hamiltonian \eqref{eq:H-quantum} describes quantum mechanics on the manifold $\mathcal{M}_d$.
This is seen by rewriting the Hamiltonian operator in the local coordinates $\bm{q}$ on $\mc{M}_d$ introduced with the embedding \eqref{eq:emb-Xq}.
In these coordinates, the differential operator \eqref{eq:LB-x} takes the following form:
\begin{align}
    \Delta_{\mc{M}_d} = \frac{1}{\sqrt{\det g}} \del_\alpha \left( \sqrt{\det g} g^{\alpha \beta} \del_\beta \right)\,,
    \label{eq:LB-local}
\end{align}
where $g^{\alpha\beta}$ is the inverse of the induced metric $g_{\alpha\beta}$.
This is the well-known form of the Laplace--Beltrami operator on a curved space with the metric $g_{\alpha \beta}$.
The Hamiltonian operator is then written in the local coordinates as
\begin{align}
    \hat{H} = - \frac{1}{2m} \Delta_{\mc{M}_d} + V(\bm{X}(\bm{q}))\,,
    \qquad 
    \Delta_{\mc{M}_d} = \frac{1}{\sqrt{\det g}} \del_\alpha \left( \sqrt{\det g} g^{\alpha \beta} \del_\beta \right)\,,
    \label{eq:H-quantum-QRHD}
\end{align}
with the embedding \eqref{eq:emb-Xq}.
The derivation of the above results is summarized in App.~\ref{app:geometry}.

\paragraph{Momentum operator.}

Now we give some comments on the momentum operator.\footnote{%
The operator $\hat{P}_i=-\iu M^j{}_i\del_j$ contains the ambient derivatives $\del_j$, whereas the wave function $\psi$ is defined only on $\mc{M}_d$.
We note that, if $\psi$ is extended off $\mc{M}_d$ in an arbitrary way, the values of $\hat{P}_i\psi$ and $\hat{P}^2\psi$ on $\mc{M}_d$ do not depend on the extension, because $M^j{}_i$ removes the normal derivative.
}
We note that $\hat{P}_i$ itself is not Hermitian with respect to the inner product \eqref{eq:inner-product}.
The partial integration of $\braket{\psi_2|\hat{P}_i\psi_1}$ gives
\begin{align}
    \hat{P}^\dagger_i = \hat{P}_i - \iu C_i\,,
    \qquad 
    C_i \coloneqq \Delta_{\mc{M}_d} X_i\,,
\end{align}
with the embedding \eqref{eq:emb-Xq} [App.~\ref{app:geoprp-constraint-adapted}, Eq.~\eqref{eq:Pdagger-app}].
Because $C_i$ is a normal vector $M^{ij} C_j =0$, the derivative operator \eqref{eq:hatP^2-Delta_M} is also written as 
\begin{align}
    \delta^{ij} \hat{P}^\dagger_i \hat{P}_j = \delta^{ij} \hat{P}_i \hat{P}_j = -\Delta_{\mc{M}_d}\,,
\end{align}
and it is found that the kinetic term of \eqref{eq:H-quantum} is Hermitian and nonnegative.

\paragraph{Ambiguity of quantization.}

Before closing this section, we give some comments on the ambiguity of the quantization.
The quantization of the constrained system is not unique~\cite{1981-DaCosta-PRA,1992-IkegamiEtAl-PTP,1997-KleinertShabanov-PLA,2006-Golovnev-IJGMMP,2009-Golovnev-RoMP,2017-SilvaEtAl-AP}. 
Let us give a concrete example.  
One can define the quantum Hamiltonian as \cite{1990-HommaEtAl-PRD}
\begin{align}
    \hat{H} = \frac{\hat{\Pi}^2}{2m} +V({\bm{x}})\,,
    \label{eq:H-quantum-2}
\end{align}%
instead of Eq.~\eqref{eq:H-quantum}.
Here $\hat{\Pi}^2$ is defined as $\hat{\Pi}^2=\delta^{ij}\hat{\Pi}_i\hat{\Pi}_j$
where
\begin{align}
    \hat{\Pi}_i = \frac{1}{2} (\hat{P}_i + \hat{P}^\dagger_i) = \hat{P}_i - \frac{\iu}{2} C_i\,.
\end{align}%
We note that the quantum system described by Eq.~\eqref{eq:H-quantum-2}
also reduces to the constrained classical system \eqref{eq:classical-L} in the classical limit.
The two quantum Hamiltonians are therefore indistinguishable at the classical level, but they differ at the quantum level.
Concretely, we find
\begin{align}
    \hat{\Pi}^2 = \hat{P}^2 + \frac{1}{4}  K^A G_{AB}  K^B = - \Delta_{\mc{M}_d} + \frac{1}{4} K^A G_{AB} K^B\,,
    \label{eq:Pi2-main}
\end{align}
where $K^A = \tr K^A = g^{\alpha\beta} K^A{}_{\alpha\beta}$ is the trace of the second fundamental form or extrinsic curvature introduced in App.~\ref{app:geometry}.
If Eq.~\eqref{eq:H-quantum-2} is used instead of Eq.~\eqref{eq:H-quantum}, the second term of Eq.~\eqref{eq:Pi2-main} thus appears as a quantum correction $\frac{1}{8m} K^A G_{AB} K^B$ to the potential.
A similar quantum-level difference appears in the comparison between QCHD and QHD+Penalty, which we discuss in Sec.~\ref{sec:penalty}.

\section{Quantum Constrained Hamiltonian Descent}
\label{sec:qchd}

We now apply the formulation of quantum mechanics for constrained systems discussed in Sec.~\ref{sec:constrained} to the construction of a quantum algorithm for constrained optimization.
The concept is the same as that of QHD~\cite{2023-LengEtAl}.
We regard $\bm{x}$ as the parameters and $V(\bm{x})$ as the cost function, and find the minimum point $\bm{x}_\ast$ of $V$ by simulating the quantum ``particle'' motion.
As in QHD, we introduce a monotonically increasing factor $a(t)$, which produces the friction in the classical equations of motion (EOMs) to drive the wave function toward the potential minimum.
The convergence of the expected cost to its minimum has been proved, under suitable conditions on the schedule and on $V$, for QHD with convex $V$~\cite{2023-LengEtAl}.
For the Laplace--Beltrami kinetic term, it has been proved for geodesically convex $V$ on geodesically convex domains of spheres and simplices~\cite{2026-HermanEtAl}.

According to the procedure in Refs.~\cite{2023-LengEtAl,2026-AbeNagai}, we consider the following time-dependent Hamiltonian operator based on \eqref{eq:H-quantum}
\begin{align}
    \hat{H}_\QCHD(t) =  \frac{1}{a(t)}\biggl[\frac{-\Delta_{\mc{M}_d}}{2m}  \biggr] + a(t) V(\bm{x})\big|_{\mc{M}_d}\,,
    \label{eq:H-QCHD}
\end{align}
and the time evolution is governed by the Schr\"odinger equation
\begin{align}
    \iu \frac{\del}{\del t} \psi(\bm{x},t) = \hat{H}_\QCHD(t) \psi(\bm{x},t)\,.
    \label{eq:schrodinger-QCHD}
\end{align}
Here, $\Delta_{\mc{M}_d}$ is the differential operator determined by the constraints $f^A(\bm{x})$ in \eqref{eq:LB-x}, namely the Laplace--Beltrami operator on the constraint manifold.
We call the quantum optimization algorithm based on this time-dependent Hamiltonian simulation Quantum Constrained Hamiltonian Descent (QCHD).

In QCHD, the time-dependent factor $a(t)$ plays the same role as in Q(R)HD: it provides a friction term in the EOM and dissipates the energy of the system.
In the classical theory, the Lagrangian \eqref{eq:classical-L} contains $a(t)$ as\footnote{
Although the original QHD paper~\cite{2023-LengEtAl} allows an additional time-dependent factor in the potential term, $V(\bm{x}) \to \eta(t) V(\bm{x})$, we set $\eta(t)=1$ for simplicity.
} 
\begin{align}
    L = a(t) \left[ \frac{m}{2} \delta_{ij} \dot{x}^i \dot{x}^j - V(\bm{x}) + \lambda_A \dot{f}^A \right]\,,
\end{align}
and the classical EOM is read as 
\begin{align}
    m M^i{}_j \left( \ddot{x}^j + \frac{\dot{a}}{a} \dot{x}^j \right) = - M^{ij} \del_j V\,.
\end{align}
The projector $M^i{}_j$ comes from the constraints.
The particle thus descends along the gradient of $V$ in the tangential directions of the constraint manifold $\mc{M}_d$.
The friction decelerates it, and it settles at the minimum of $V|_{\mc{M}_d}$.
The QCHD Hamiltonian operator \eqref{eq:H-QCHD} is the quantization of this classical system by the prescription of the previous section.

\paragraph{Relation among QCHD, QHD, and QRHD.}

The essential difference between QCHD and QHD lies in the differential operator in the kinetic term.
QHD uses the ordinary Laplacian $\Delta=\delta^{ij}\partial_i\partial_j$, so that the quantum dynamics takes place in the whole unconstrained space $\mathbb{R}^N$.
QCHD instead uses the Laplace--Beltrami operator $\Delta_{\mc{M}_d}$, so that the quantum dynamics is restricted to the constraint manifold $\mc{M}_d$ from the outset.
As a result, the constraints are satisfied exactly throughout the optimization in QCHD, and no additional external parameters such as the penalty coefficient are required for enforcing the constraints.

We note, however, that the quantization of this paper is based on the velocity constraints $\dot{f}^A=0$, as explained in Sec.~\ref{sec:classical}.
It therefore presupposes that the initial state is a state of the constrained system, i.e., a wave function on $\mc{M}_d$ as defined in Sec.~\ref{sec:quantization}, rather than a wave packet in $\mathbb{R}^N$.
The time evolution \eqref{eq:schrodinger-QCHD} is defined on this space, so that the state remains on $\mc{M}_d$ by construction.

Rewriting QCHD in local coordinates on $\mc{M}_d$, we find from \eqref{eq:LB-local} that it takes the form of QRHD~\cite{2026-AbeNagai}.
That is, QCHD is a concrete realization of QRHD expressed in ambient coordinates.
The two approaches describe the same theory in different coordinate systems.
Fig.~\ref{fig:concept} in Introduction shows a schematic picture of this relation between QCHD and QRHD.

\section{Relation to QHD with a penalty term}
\label{sec:penalty}

An alternative to QCHD is to add a penalty term to the ordinary QHD and let the penalty confine the dynamics near the constraint manifold $\mc{M}_d$.
We call this approach QHD+Penalty.
In this section we ask how QHD+Penalty is related to QCHD.
For this purpose we derive, for a sufficiently strong penalty, the effective Hamiltonian that governs the tangential motion of QHD+Penalty on $\mc{M}_d$, and compare it with the QCHD Hamiltonian of Sec.~\ref{sec:qchd}.
We then examine how the result depends on the form of the penalty, which is not unique.
The geometric calculations are summarized in App.~\ref{app:geometry} and the treatment of the normal oscillation in App.~\ref{app:effective}.

\subsection*{Effective Hamiltonian for QHD+Penalty with strong penalty}

\paragraph{Setup for QHD+Penalty.}

Let us consider the QHD Hamiltonian and quadratic penalty term $\mathcal{P}({\bm{x}})$ as 
\begin{align}
    \hat{H}_{\rm pen}(t)
    &=
    \frac{1}{a(t)}\,\biggl[\frac{-\Delta}{2m}\biggr]
    +
    a(t)\,\Bigl[V(\bx)+\kappa\,\mathcal{P}(\bx)\Bigr]
    \,,
    \label{eq:QHD-penalty}
    \\
    \mathcal{P}({\bm{x}})
    &=
    \xi^2 =\delta_{ab}\xi^a \xi^b
    \label{eq:penalty-xi2}
\end{align}%
where $\Delta = \delta^{ij} \del_i \del_j$ is the Laplacian of $\mbb{R}^N$, $\kappa >0$ is the strength of the penalty, and $\xi^a$ is the normal displacement from the constraint manifold.
In the large-$\kappa$ regime, namely the $\kappa \to \infty$ limit, the relevant states are localized in a narrow neighborhood of the constraint manifold $\mc{M}_d$. 
It is therefore natural to resolve the ambient coordinates into tangential coordinates on $\mc{M}_d$ and normal displacements away from it. 
We accordingly work in a tubular neighborhood of $\mc{M}_d$, where the ambient coordinates can be written as
\begin{align}
    x^i = X^i(\bm{q}) + \xi^a \nu^i{}_a(\bm{q})\,.
\end{align}
Here, $q^\alpha$ is the local coordinate on $\mc{M}_d$, $X^i(\bm{q})$ is the embedding map, and $\nu^i{}_a(\bm{q})$ is an orthonormal frame spanning the normal space at each point of the constraint manifold satisfying 
\begin{align}
    \nu^i{}_a = \frac{\del x^i}{\del \xi^a}\,,
    \qquad
    \delta_{ij} \nu^i{}_a \nu^j{}_b = \delta_{ab}\,.
    \label{eq:nu-orthonormal}
\end{align}%
We also introduce the tangent basis vectors of $\mc{M}_d$ by 
\begin{align}
    e^i{}_\alpha = \frac{\del X^i}{\del q^\alpha}\,,
    \qquad
    \delta_{ij} e^i{}_\alpha e^j{}_\beta = M_{ij} e^i{}_\alpha e^j{}_\beta = g_{\alpha\beta}\,.
\end{align}
These frame fields satisfy
\begin{align}
     \delta^{ij} = g^{\alpha\beta} e^i{}_\alpha e^j{}_\beta + \delta^{ab} \nu^i{}_a \nu^j{}_b\,,
     \qquad 
     \delta_{ij} \nu^i{}_a e^j{}_\beta = 0\,.
\end{align}
The orthonormal frame can be constructed from the normal vector \eqref{eq:normal} restricted to $\mc{M}_d$, for example
\begin{align}
    \nu^i{}_a = b_{Aa} n^{Ai}\big|_{\mc{M}_d}\,,
    \qquad
    b^A{}_a = \frac{\del f^A}{\del \xi^a}\bigg|_{\xi=0}\,.
\end{align}
This $b^A{}_a$ satisfies
\begin{align}
    \delta^{ab} b^A{}_a b^B{}_b = G^{AB}\,.
    \label{eq:bb-G}
\end{align}
The details of the geometric formulation are found in App.~\ref{app:geometry}.

The penalty $\mathcal{P}=\xi^2$ in Eq.~(\ref{eq:penalty-xi2}) is one of many possible choices.
For example, $\mathcal{P}=\delta_{AB}f^A f^B$ is used in Sec.~\ref{sec:numerics}.
We adopt $\xi^2$ because it makes the normal part of the problem as simple as possible:
$\xi^a$ are flat coordinates of the normal space [Eq.~\eqref{eq:nu-orthonormal}], so that the penalty is an isotropic harmonic potential with the same stiffness at every point of $\mc{M}_d$.
The dependence on this choice is discussed at the end of this section.

\paragraph{Measure transformation and leading Hamiltonian.}

We now rewrite $\hat{H}_\pen$ in the tubular coordinates $(\bm{q},\bm{\xi})$ and separate it into a tangential part, a normal part, and their coupling.
The details of the tubular coordinates we use are summarized in App.~\ref{app:geometry}.
The first point to note is the volume element, which in these coordinates is
\begin{align}
    \dd^N\bm{x} = \dd^d \bm{q} \dd^M \bm{\xi} \sqrt{\det g} \mc{J}\,,
    \qquad 
    \mc{J} = \det (\delta_\alpha^\beta + \xi^a  K_{a\alpha}{}^\beta)\,,
\end{align}
where the tube is chosen so that $\mc{J} > 0$.\footnote{
If we introduce the so-called shape operator by $\ms{S}_a{}^\alpha{}_\beta \coloneqq - g^{\alpha \gamma}K_{a\gamma \beta} = - K_a{}^\alpha{}_\beta$, the Jacobian is also written as $\mc{J} = \det (\delta^\alpha_\beta - \xi^a \ms{S}_a{}^\alpha{}_\beta)$.
}
Here, $K_{a\alpha}{}^\beta$ is the extrinsic curvature given by
\begin{align}
     K_a{}^\alpha{}_{\beta} = g^{\alpha \gamma}  K_{a\gamma \beta}\,,
    \qquad
     K_{a\alpha\beta} = e_{i\beta} \del_\alpha \nu^i{}_a\,.
\end{align}
The factor $\mc{J}$ is a ratio of tubular and surface volume densities.
Because of this $\bm{\xi}$-dependent factor, the inner product of $\psi$ does not factorize into tangential and normal parts.
After the wave function redefinition $\tilde{\psi} \coloneqq \mc{J}^{1/2} \psi$, however, the inner product of $\tilde{\psi}$ factorizes into the part with the induced measure $\dd^d \bm{q} \sqrt{\det g}$ and that with the flat measure $\dd^M \bm{\xi}$, so that the normal direction can be treated as ordinary $M$-dimensional quantum mechanics in flat space.

For this redefinition, the Hamiltonian acting on $\tilde{\psi}$ is $\hat{\tilde{H}}_\pen(t) \coloneqq \mc{J}^{1/2} \hat{H}_\pen(t) \mc{J}^{-1/2}$, which is decomposed as
\begin{align}
    \hat{\tilde{H}}_\pen(t) &= \hat{H}_q(t) + \hat{H}_\xi(t) + \hat{R}(t)\,,
    \label{eq:H-tilde-decomp}
    \\
    \hat{H}_q(t) &= 
    \frac{1}{a(t)}\biggl[\frac{-\Delta^{(\mc{A})}_{\mc{M}_d}}{2m} \biggr]
    + a(t) \biggl[
        V(\bm{X}(\bm{q})) + \frac{V_\geo(\bm{q})}{a(t)^2}
    \biggr]\,,
    \label{eq:H-q}
    \\
    \hat{H}_\xi(t) &= 
    \frac{1}{a(t)}\biggl[\frac{-\Delta_\xi}{2m} \biggr]
    + a(t) \kappa \xi^2\,,
    \qquad 
    \Delta_\xi = \delta^{ab} \del_a \del_b\,,
    \label{eq:H-xi}
    \\
    \hat{R}(t) &= \mc{O}(\xi)\,.
    \label{eq:H-R}
\end{align}
Here, $V_\geo$ is the geometric potential and $\Delta^{(\mc{A})}_{\mc{M}_d}$ is the covariant Laplacian with the normal-bundle connection, which we define in turn.
The geometric potential is given by
\begin{align}
    V_\geo(\bm{q}) = \frac{1}{8m}\bigl(
        \delta^{ab} K_a K_b - 2 \delta^{ab} K_{a \alpha\beta} K_b{}^{\alpha\beta}\,,
    \bigr)
    \label{eq:Vgeo}
\end{align}
where we define the mean curvature components by
\begin{align}
    K_a \coloneqq \tr K_{a\alpha\beta} = g^{\alpha\beta} K_{a \alpha\beta}\,.
\end{align}
The derivative operator
\begin{align}
    \Delta_{\mc{M}_d}^{(\mc{A})} = \frac{1}{\sqrt{\det g}} D_\alpha \bigl( \sqrt{\det g} g^{\alpha \beta} D_\beta \bigr)\,,
    \quad 
    D_\alpha \coloneqq \del_\alpha + \frac{1}{2} \mc{A}_{\alpha ab}(\bm{q}) \ms{L}^{ab}\,,
    \quad
    \ms{L}_{ab} \coloneqq \xi_a \del_b - \xi_b \del_a\,,
\end{align}
is the covariant Laplacian on $\mc{M}_d$ with the normal-bundle connection 
\begin{align}
    \mc{A}_{\alpha ab} = - \nu^i{}_b \del_\alpha \nu_{ia} = - \mc{A}_{\alpha ba}\,,
\end{align}
as a gauge field.
$\ms{L}_{ab}$ is the orbital angular momentum of $\xi^a$ and generates the $\mr{SO}(M)$ rotations in the normal space.\footnote{
From the definition of $\ms{L}_{ab}$, the commutation relation is evaluated as 
$[\ms{L}_{ab}, \ms{L}_{cd}] = - (\delta_{ac} \ms{L}_{bd} - \delta_{bc} \ms{L}_{ad} - \delta_{ad} \ms{L}_{bc} + \delta_{bd} \ms{L}_{ac})$,
which is nothing but the well-known commutation relation of the angular momentum operator.
}
If $M=1$, this connection automatically vanishes and $\Delta^{(\mc{A})}_{\mc{M}_d} = \Delta_{\mc{M}_d}$.

The decomposition \eqref{eq:H-tilde-decomp}--\eqref{eq:H-R} is known from the quantum mechanics of a particle confined to a manifold by a time-independent potential~\cite{1988-Tolar-LNP,1993-FujiiOgawa-PTP,1993-MaranerDestri-MPLA,2003-SchusterJaffe-AP}.
The geometric potential \eqref{eq:Vgeo} was derived in Ref.~\cite{1988-Tolar-LNP} and reduces to the result of Ref.~\cite{1981-DaCosta-PRA} for a surface ($M=1$) and a curve ($M=2$) in $\mbb{R}^3$.
The normal-bundle connection as a gauge field was obtained in Refs.~\cite{1993-FujiiOgawa-PTP,1993-MaranerDestri-MPLA}.
Systematic expansions in powers of $\xi$ are given in Refs.~\cite{1995-Maraner-JPAMG,2017-SilvaEtAl-AP}.
Here the time-dependent factor $a(t)$ is carried through the derivation.

\paragraph{Normal evolution and reduced effective Hamiltonian.}

To reduce the dynamics to $\mc{M}_d$, we next solve the normal part $\hat{H}_\xi$ in Eq.~\eqref{eq:H-xi},
whose eigenstates serve as the basis in which the tangential effective Hamiltonian is defined.
$\hat{H}_\xi$ is an $M$-dimensional isotropic harmonic oscillator with mass $ma(t)$ and spring constant $2\kappa a(t)$,
so that its frequency $\Omega_\perp$ and eigenvalues $E_{n_\perp}$ do not depend on $a(t)$:
\begin{align}
    \Omega_\perp = 
    \sqrt{\frac{2\kappa a(t)}{ma(t)}}
    = 
    \sqrt{\frac{2\kappa}{m}}\,,
    \qquad 
    E_{n_\perp} = \Omega_\perp \left( n_\perp + \frac{M}{2} \right) \,.
\end{align}
Here $n_\perp=0,1,2,\ldots$ is the total excitation number of the $M$ normal directions 
and $M\Omega_\perp/2$ is the zero-point energy.
The eigenfunctions, on the other hand, depend on time through $a(t)$.
For example, the position variance per normal direction in the instantaneous ground state is
\begin{align}
    \sigma^2_\perp(t) = \frac{1}{M}\braket{\xi^2}_0 = \frac{1}{2 m a(t) \Omega_\perp}\,,
\end{align}
which shrinks as $a(t)$ grows.
Consequently, the oscillator level $n_\perp$ is not exactly conserved.
This is seen from the Hamiltonian in the instantaneous oscillator basis, derived in App.~\ref{app:effective},
\begin{align}
    \hat{H}_{\perp,\mr{inst.}} &= \Omega_\perp \left(
        \sum_a \hat{\ms{b}}^\dagger_a \hat{\ms{b}}_a + \frac{M}{2}
    \right) + \frac{\dot{a}}{2a} \hat{\ms{D}}_\perp\,,
    \label{eq:inst-osc-ham}
    \\
    \hat{\ms{D}}_\perp &=\frac{1}{2} \sum_a (\xi^a p_a + p_a \xi^a) = \frac{\iu}{2} \sum_a \left( (\hat{\ms{b}}^\dagger_a)^2 - \hat{\ms{b}}_a^2 \right)\,.
    \label{eq:dilatation-op}
\end{align}
Here, $p_a = -\iu \del_a$, and $\hat{\ms{b}}_a$ is the annihilation operator of the instantaneous oscillator basis.
In addition to the ordinary time-independent harmonic-oscillator Hamiltonian, the second term of Eq.~\eqref{eq:inst-osc-ham}, $\frac{\dot{a}}{2a}\hat{\ms{D}}_\perp$, arises from the time dependence of $a(t)$ through the unitary transformation to the instantaneous oscillator basis.
We call it the dilatation term, and it mixes levels differing by two, as seen from Eq.~\eqref{eq:dilatation-op}.
The instantaneous levels are therefore conserved only in the adiabatic approximation, in which the dilatation term is neglected because $a(t)$ varies slowly compared with the normal oscillation, $|\dot{a}|/(a\Omega_\perp)\ll1$.

Although the level is not conserved, the Hamiltonians for $q$ and $\xi$ commute with each other,
\begin{align}
    [\hat{H}_q(t), \hat{H}_\xi(s)] = 0\,.
\end{align}
This leads to the factorized unitary evolution 
\begin{align}
    \hat{U}(t) \simeq \hat{U}_q (t) \hat{U}_\xi(t)\,,
\end{align}
if the mixing term between $q$ and $\xi$ is small.
In addition, $\hat{H}_\xi$ commutes with the angular momentum operator in the normal space,
\begin{align}
    [\hat{H}_\xi, \ms{L}_{ab}] =0\,,
\end{align}
which implies that the time evolution generated by $\hat{H}_\xi$ preserves the $\mr{SO}(M)$ representation in the normal space even when the dilatation term changes the level $n_\perp$.
For a fixed instantaneous normal level and $|\dot{a}|/(a\Omega_\perp) \ll 1$, the adiabatic reduction gives, after removing the spatially uniform phase $\ee^{-\iu E_{n_\perp} t}$,
\begin{align}
    \hat{H}_\eff^{(n_\perp)}(t) = - \frac{1}{2ma} \Delta_{\mc{M}_d}^{(\mc{A},n_\perp)} + a V(\bm{X}(\bm{q})) + \frac{V_\geo(\bm{q})}{a} + \mc{O}(\kappa^{-1/2})\,,
    \label{eq:H_eff^n_perp}
\end{align}
where $\ms{L}_{ab}$ acts as a matrix on the eigenspace of level $n_\perp$.
Here, $\Delta_{\mc{M}_d}^{(\mc{A} ,n_\perp)}$ denotes the covariant Laplace–Beltrami operator acting on the degenerate eigenspace of the $n_\perp$-th normal oscillator level. 
The dependence on $n_\perp$ enters through the matrix representation of the normal rotation generators $\ms{L}_{ab}$ in this eigenspace.
For a singlet state for the rotation, the derivative operator reduces to the ordinary Laplace--Beltrami operator due to $\ms{L}_{ab}=0$:
\begin{align}
    \Delta^{(\mc{A},n_\perp)}_{\mc{M}_d} \to \Delta_{\mc{M}_d}\,.
\end{align}

The reduction to Eq.~\eqref{eq:H_eff^n_perp} rests on the following three assumptions.
\begin{enumerate}[label=(\roman*)]
    
    \item The penalty is the squared normal distance, Eq.~\eqref{eq:penalty-xi2}.
    The normal oscillator \eqref{eq:H-xi} then has the same frequency $\Omega_\perp$ at every point of $\mc{M}_d$, and its energy $E_{n_\perp}$ is independent of the position on $\mc{M}_d$ and can be removed by a phase.
    
    \item The normal motion is confined within the tube. For the levels $n_\perp$,
    \begin{align}
        \Braket{\xi^2}_{n_\perp}^{1/2} = \Bigl(\frac{2n_\perp+M}{2ma(t)\Omega_\perp}\Bigr)^{1/2}
        \label{eq:xi2-level}
    \end{align}
    is small compared with the radii of curvature of $\mc{M}_d$ and with the scale on which $V$ varies in the normal directions.
    The remainder $\hat R$ in Eq.~\eqref{eq:H-tilde-decomp} is then a small perturbation.
    Its terms linear in $\xi^a$ have no matrix elements within a level and contribute only in second order, $\mc{O}(\braket{\xi^2}/\Omega_\perp)=\mc{O}(\kappa^{-1})$, and its quadratic terms give the $\mc{O}(\kappa^{-1/2})$ in Eq.~\eqref{eq:H_eff^n_perp}.
    These estimates refer to matrix elements between states of a fixed level whose tangential derivatives are bounded independently of $\kappa$, not to the operator norm.

    \item The factor $a(t)$ varies slowly compared with the normal oscillation, $|\dot a|/(a\Omega_\perp)\ll1$, so that the dilatation term in Eq.~\eqref{eq:inst-osc-ham} does not change the level $n_\perp$.
    For $a(t)=\ee^{2\gamma t}$ this condition reads $2\gamma/\Omega_\perp=\sqrt{2m\gamma^2/\kappa}\ll1$.

\end{enumerate}
Under (ii) and (iii) the level $n_\perp$ is conserved, and $[\hat H_q,\hat H_\xi]=0$ factorizes the time evolution, so that Eq.~\eqref{eq:H_eff^n_perp} holds for each level separately.
A superposition of levels evolves as a mixture of these tangential dynamics, which differ from one another only through the representation of $\ms{L}_{ab}$ and through the $\mc{O}(\kappa^{-1/2})$ term.

\paragraph{Comparison with QCHD.}

Under the assumptions (i)--(iii), the difference from the QCHD Hamiltonian~\eqref{eq:H-QCHD} is
\begin{align}
    \hat{H}^{(n_\perp)}_\eff(t) - \hat{H}_\QCHD(t) = \frac{1}{a(t)}\biggl[ V_\geo - \frac{\Delta^{(\mc{A},n_\perp)}_{\mc{M}_d} - \Delta_{\mc{M}_d}}{2 m} \biggr] + \mc{O}(\kappa^{-1/2})\,.
    \label{eq:difference}
\end{align}
The terms in the bracket have no explicit dependence on $\kappa$.
They are non-decoupling effects that survive in the large-$\kappa$ limit, in contrast to the $\mc{O}(\kappa^{-1/2})$ correction.
The degrees of freedom in the normal direction are frozen by the penalty, but they still have effects on the dynamics on the constraint manifold through $V_\geo$ and the normal-bundle connection.

This difference illustrates the quantization ambiguity discussed in Sec.~\ref{sec:quantization}.
Quantizing the dynamics directly on the constraint manifold (QCHD) and adding a strong penalty after quantization (QHD+Penalty) generally lead to different quantum theories.
While QCHD depends only on the intrinsic geometry of the constraint manifold, QHD+Penalty retains information about its embedding through the geometric potential and the normal-bundle connection. 
When $V_{\geo}$ is constant on the manifold and the connection term is absent, the difference \eqref{eq:difference} reduces to an overall phase, and the leading tangential dynamics agrees with QCHD.
Otherwise, the geometric potential modifies the effective potential and can affect the transient dynamics, although its relative contribution to the cost function is suppressed by $a(t)^{-2}$.

\paragraph{Dependence on the form of the penalty.}

We finally examine how the effective Hamiltonian depends on the form of the penalty.
Assumption (i) fixes the penalty to the squared normal distance $\xi^2$.
In practice, however, the penalty is often written with the constraint functions, $\kappa\sum_A(f^A)^2$.
Moreover, the functions $f^A$ that define a given $\mc{M}_d$ are not unique, since $f^A\to \ms{h}(\bx)f^A$ with $\ms{h}>0$ leaves $\mc{M}_d$ unchanged.
Classically, this choice does not affect the constrained motion in the limit $\kappa\to\infty$, in which the normal motion is frozen.
Quantum mechanically, however, the choice enters the effective Hamiltonian through the level energies of the normal oscillator, as we now show.
This is the ambiguity of the penalty approach mentioned in Sec.~\ref{sec:introduction}, which QCHD does not have.

Let us see this concretely.
Since the constraint is expanded as
\begin{align}
    f^A(\bm{x}) = f^A(\bm{X}(\bm{q}) + \xi^b \bm{\nu}_b(\bm{q})) = b^A{}_b(\bm{q}) \xi^b + \mc{O}(\xi^2)\,,
\end{align}
due to $f^A(\bm{X}(\bm{q})) = 0$,\footnote{
From this expansion of $f^A(\bm{x})$, it is found that the penalty term in \eqref{eq:QHD-penalty} is regarded as the leading order contribution of the penalty $G_{AB} f^A f^B$, since $G_{AB} f^A f^B = G_{AB} b^A{}_a b^B{}_b \xi^a \xi^b + \mc{O}(\xi^3) = \delta_{ab} \xi^a \xi^b + \mc{O}(\xi^3)$.
}
the penalty near $\mc{M}_d$ is a quadratic form in $\xi^a$ with position-dependent coefficients,
\begin{align}
    \kappa\sum_A (f^A)^2 \,=\, \kappa\,\delta_{AB}\,b^A{}_a(\bm{q})\,b^B{}_b(\bm{q})\,\xi^a\xi^b + \mc{O}(\xi^3)\,.
    \label{eq:penalty-f2}
\end{align}
The normal Hamiltonian \eqref{eq:H-xi} then becomes an anisotropic oscillator whose squared frequencies $\Omega_a^2(\bm{q})$ are the eigenvalues of the $M\times M$ matrix $(2\kappa/m)\,\delta_{AB}\,b^A{}_a b^B{}_b$
[equivalently, $2\kappa/m$ times the eigenvalues of the Gram matrix $G^{AB}(\bm{q})$].
Its level energies
\begin{align}
    E_{\bm{n}}(\bm{q}) \,=\, \sum_a \Bigl(n_a+\frac{1}{2}\Bigr)\,\Omega_a(\bm{q})\,,
    \label{eq:E-level-f2}
\end{align}
depend on the position and can no longer be removed by a phase, and they remain in the effective Hamiltonian as an additional potential.
For a single constraint, 
Eq.~(\ref{eq:H_eff^n_perp}) is replaced by
\begin{align}
    \hat{H}_\eff^{(n_\perp)}(t) \,=\, - \frac{1}{2ma} \Delta_{\mc{M}_d} + a V(\bm{X}(\bm{q})) + \frac{V_\geo(\bm{q})}{a}
    + \Bigl(n_\perp+\frac{1}{2}\Bigr)\,\Omega^{(f)}_\perp(\bm{q}) + \mc{O}(\kappa^{-1/2})\,,
    \label{eq:H_eff-f2}
\end{align}
where
\begin{align}
    \Omega^{(f)}_\perp(\bm{q})
    = |\nabla f(\bm{X}(\bm{q}))|\sqrt{\frac{2\kappa}{m}}
    = |\nabla f(\bm{X}(\bm{q}))| \Omega_\perp\,.
\end{align}
We here ignore derivative terms that originate from the position dependence of the normal wave function.\footnote{\label{fn:normalization}
Besides the additional potential $(n_\perp+\tfrac12)\Omega^{(f)}_\perp(\bm{q})$ in Eq.~\eqref{eq:H_eff-f2}, there are two further corrections of order $\kappa^0$, both proportional to $1/a$ like $V_\geo/a$.
One is the Born--Huang potential $(16ma)^{-1}g^{\alpha\beta}\partial_\alpha\log|\nabla f|\,\partial_\beta\log|\nabla f|$~\cite{2010-WachsmuthTeufel-PRA}, which arises from the position dependence of the normal width $\propto|\nabla f|^{-1/2}$ and is the derivative term omitted in Eq.~\eqref{eq:H_eff-f2}.
The other is the shift of the level energies by the anharmonic terms of the normal potential, i.e., the $\mc{O}(\xi^3)$ terms in Eq.~\eqref{eq:penalty-f2}.
For a penalty $\kappa F^2$ with $F=\xi+c(\bm{q})\xi^2+d(\bm{q})\xi^3+\mc{O}(\xi^4)$ near $\mc{M}_d$, this shift is $-c^2/(ma)+3d/(4ma)$ for the lowest level.
Normalizing the constraint function as $F=f/|\nabla f|$ makes the frequency uniform on $\mc{M}_d$ and removes the position dependence of the level energies and the Born--Huang potential, but the anharmonic shift remains and depends on the choice of $f$ through $c$ and $d$.
}
The additional potential, the term $(n_\perp+\tfrac12)\Omega^{(f)}_\perp(\bm{q})$ in Eq.~\eqref{eq:H_eff-f2}, has two properties.
It does not depend on $a(t)$, so that its weight relative to the cost $aV$ decreases as $1/a$.
Its variation along $\mc{M}_d$ grows as $\sqrt{\kappa}$, because $\Omega_\perp^{(f)}(\bm{q})=|\nabla f(\bm{X}(\bm{q}))|\sqrt{2\kappa/m}$.
For a position-dependent $|\nabla f|$, the large-$\kappa$ limit of QHD+Penalty therefore does not approach QCHD.
The variation also depends on the choice of $f$ through $|\nabla f|$.
The comparison \eqref{eq:difference} therefore refers to the distance penalty of assumption (i).
For a penalty $\kappa\sum_A(f^A)^2$ with a position-dependent $|\nabla f|$, the quantum corrections derived above are added to the difference, as illustrated by the elliptic example of Sec.~\ref{sec:num-ellipse}.

\section{Numerical comparison of QCHD and QHD+Penalty}
\label{sec:numerics}

In this section, we apply QCHD and QHD+Penalty to four constrained optimization problems and compare their performance numerically.

\subsection{Setup}

We chose the four examples to examine the following points.
\begin{enumerate}[label=(\arabic*)]
    \item \textbf{Spherical constraint with a quadratic potential} (Sec.~\ref{sec:num-sphere}).
    This is the simplest example: $V_\geo$ is constant, the normal-bundle connection vanishes ($M=1$), and $|\nabla f|$ is constant, so that the large-$\kappa$ effective theory of Sec.~\ref{sec:penalty} coincides with QCHD up to a phase.
    We quantify the difference between the two methods at finite $\kappa$.

    \item \textbf{Spherical constraint with the Ackley potential} (Sec.~\ref{sec:num-ackley}).
    The geometry is the same as in (1), but the potential is nonconvex with many local minima on the constraint manifold.
    We test whether the global search capability survives under the constraint.

    \item \textbf{Elliptic constraint with a quadratic potential} (Sec.~\ref{sec:num-ellipse}).
    Here $V_\geo$ and $|\nabla f|$ depend on the position, so that the large-$\kappa$ limit of the penalty $\kappa f^2$ does not approach QCHD as shown in Sec.~\ref{sec:penalty}.
    We examine how this difference shows up at finite $\kappa$.

    \item \textbf{Two constraints in $\mathbb{R}^3$ with a nonconvex potential} (Sec.~\ref{sec:num-twocons}).
    This example has two constraints ($M=2$) and a one-dimensional constraint manifold in a three-dimensional ambient space.
    Its geometry is again such that the large-$\kappa$ effective theory coincides with QCHD up to a constant.
    We examine the effect of two normal directions.
\end{enumerate}

\paragraph{Common setup.}
In this section, we write the Cartesian coordinates as $\bx=(x,y)^{\mathrm{T}}$ for $N=2$ and $\bx=(x,y,z)^{\mathrm{T}}$ for $N=3$.
We numerically solve the Schr\"odinger equation with the Hamiltonian \eqref{eq:H-QCHD} for QCHD and with
\begin{align}
    \hat{H}_\pen(t) = \frac{1}{a(t)}\biggl[
        \frac{-\Delta}{2m}
    \biggr] + a(t) \left[ V(\bm{x}) + \kappa \mc{P}(\bm{x}) \right]\,,
    \qquad 
    \mc{P}(\bm{x}) = f(\bm{x})^2 = \delta_{AB} f^A(\bm{x}) f^B(\bm{x})\,,
    \label{eq:H-pen-num}
\end{align}%
for QHD+Penalty.
We set $m=1$ and $a(t)=\ee^{2\gamma t}$.
The value of $\gamma$ is specified for each example, and the simulations are run up to $t_{\rm end}=20$ ($t_{\mr{end}} = 40$ in Sec.~\ref{sec:num-twocons}).
For the comparison between QCHD and QHD+Penalty, both methods start from the same initial state,
a tube state that is uniform along $\mc{M}_d$ (with respect to the arc length)
and Gaussian in the normal directions with width $\sigma_r=0.05$,
\begin{align}
    \psi_0(\bm{x}) \propto \exp \left( - \frac{d(\bm{x})^2}{2 \sigma_r^2} \right)\,,
    \label{eq:initial-tube}
\end{align}
where $d(\bx)^2$ is the squared normal distance from $\bx$ to $\mc{M}_d$, so that $d(\bx)^2=\xi^2$ in the notation of Sec.~\ref{sec:penalty}.
The time evolution is obtained by solving the time-dependent Schr\"odinger equation. 
See App.~\ref{app:numerics} for the details of the calculation.

We remark that the numerical examples are set in a situation different from that of Sec.~\ref{sec:penalty} in two respects.
First, the penalty is $\mathcal{P}=f^2$ instead of Eq.~\eqref{eq:penalty-xi2}.
As discussed in Sec.~\ref{sec:penalty}, the difference gives an additional potential in QHD+Penalty when $|\nabla f|$ depends on the position, which is the case in the elliptic example (Sec.~\ref{sec:num-ellipse}).
Second, the normal motion of QHD+Penalty does not satisfy assumption (ii) of Sec.~\ref{sec:penalty}, i.e., $\braket{\xi^2}^{1/2}$ is not small compared with the radii of curvature of $\mc{M}_d$:
the narrow initial tube \eqref{eq:initial-tube} releases a normal kinetic energy $1/(4m\sigma_r^2)=100\gg\Omega_\perp$,
and the density spreads beyond the curvature radius of $\mc{M}_d$ during the transient.
The examples below therefore compare the two methods in the finite-$\kappa$ regime, where the assumptions of Sec.~\ref{sec:penalty} are not satisfied, rather than in the large-$\kappa$ regime.

\paragraph{Indicators.}

We measure the progress of the optimization by the cost evaluated after the projection $\mathsf{P}$ onto $\mc{M}_d$,
\begin{align}
    E_{\mr{proj}}(t) \coloneqq \Braket{V(\ms{P}\bm{x})}_t - V^\ast\,,
    \label{eq:def-Eproj}
\end{align}
where $V^*$ is the minimum of $V|_{\mc{M}_d}$ and $\mathsf{P}\bx$ is the (radial) projection of $\bx$ onto $\mc{M}_d$.
$E_{\rm proj}$ is the tangential optimization error and contains no contribution from the constraint violation.\footnote{%
We evaluate the cost after the projection, Eq.~\eqref{eq:def-Eproj}, rather than using $\langle V\rangle_t-V^*$, because $V$ takes values below $V^*$ off $\mc{M}_d$ whenever the constraint is active, i.e., whenever the unconstrained minimizer of $V$ does not lie on $\mc{M}_d$ (in Secs.~\ref{sec:num-sphere}--\ref{sec:num-ellipse} it lies inside the circle).
The unprojected cost would therefore reward the constraint violation.
}
We define the convergence time $t_\varepsilon$ as the first time at which $E_{\rm proj}(t)$ falls below a threshold $\varepsilon$.
The constraint violation is measured by $\braket{d^2}$ at a reference time $t_{\rm ref}$, the last time up to which all displayed quantities (the density, $\braket{d^2}$, and $E_{\rm proj}$) have been verified to be converged with respect to the grid (App.~\ref{app:numerics}, which also checks $t_\varepsilon$ beyond $t_{\rm ref}$).
We take $t_{\rm ref}=10$ in Secs.~\ref{sec:num-sphere} and \ref{sec:num-ellipse}, $8$ in Sec.~\ref{sec:num-ackley}, and $40$ in Sec.~\ref{sec:num-twocons}.
In the nonconvex examples we also use the success probability $P_{\rm succ}$, the probability of finding the particle in the basin of the global minimum.
The convergence of $t_\varepsilon$, $\braket{d^2}(t_{\rm ref})$, and $P_{\rm succ}$ with respect to the grid, the time step, and the computational domain is examined in App.~\ref{app:numerics}.

For each example we show, in a single figure, snapshots of the probability density $|\psi|^2$ for $t\le t_{\rm ref}$ and a Pareto plot
with the constraint violation on the horizontal axis and $t_\varepsilon$ (or the failure probability) on the vertical axis.
Since QCHD satisfies the constraints exactly, it appears in the Pareto plot as a horizontal line indicating $t_\varepsilon$ only.

\subsection{Spherical constraint with a quadratic potential}
\label{sec:num-sphere}

We take $N=2$ and the circle
\begin{align}
    \mc{M}_1:\quad x^2 + y^2 = R^2,
    \qquad 
    R=0.5\,,
\end{align}
as the constraint manifold.
It is represented by the constraint function $f(\bx)=x^2+y^2-R^2$.
The potential is
\begin{align}
    V(\bm{x}) = \frac{m}{2} \bm{x}^\TT A_1 \bm{x}\,,
    \qquad 
    A_1 = \begin{pmatrix}
        1 & R'\\
        R' & 1
    \end{pmatrix}\,,
    \qquad 
    R' = -0.9\,.
  \label{eq:V-quad}
\end{align}
The minimum of $V$ is at the origin, which violates the constraint.
On the circle the minimum value is $V^*=\frac{m}{2}R^2(1+R')=0.0125$, attained at the two points $\theta^*=\pi/4,\,5\pi/4$.
The penalty is $\kappa f^2=\kappa(x^2+y^2-R^2)^2$ with $\kappa=1,2,5,10$, and we set $\gamma=0.3$.
In this example $V_\geo=-1/(8mR^2)$ is constant, the normal-bundle connection vanishes ($M=1$),
and $|\nabla f|=2R$ is constant, so that $\kappa f^2$ agrees with the distance penalty $\kappa\xi^2$ to leading order in $\xi$ (with $4R^2=1$).
The large-$\kappa$ effective theory of Sec.~\ref{sec:penalty} therefore coincides with QCHD up to a phase, and the differences found below are finite-$\kappa$ effects.

Fig.~\ref{fig:num-sphere-snap} shows the evolution of the probability density.
QCHD stays on the circle and concentrates on the two minima by $t=8$.
QHD+Penalty first spreads inside and outside the circle (the released normal motion), is recaptured the faster the larger $\kappa$,
and gathers around the minima by $t=10$ for $\kappa\ge5$.
In the Pareto plot ($\varepsilon=3\times10^{-2}$), $t_\varepsilon$ is 4.45 for QCHD and 11.51, 10.80, 9.54, and 8.45 for QHD+Penalty with $\kappa=1,2,5,10$:
it decreases monotonically with $\kappa$, but even at $\kappa=10$ QHD+Penalty takes 1.9 times longer than QCHD.
The constraint violation at $t_{\rm ref}=10$ decreases from $2.5\times10^{-2}$ ($\kappa=1$) to $4.0\times10^{-3}$ ($\kappa=10$).
In the range $\kappa\le10$ examined here, QCHD thus converges faster than QHD+Penalty for every $\kappa$ while satisfying the constraint exactly.

\begin{figure}[tbp]
  \centering
  \includegraphics[width=\textwidth]{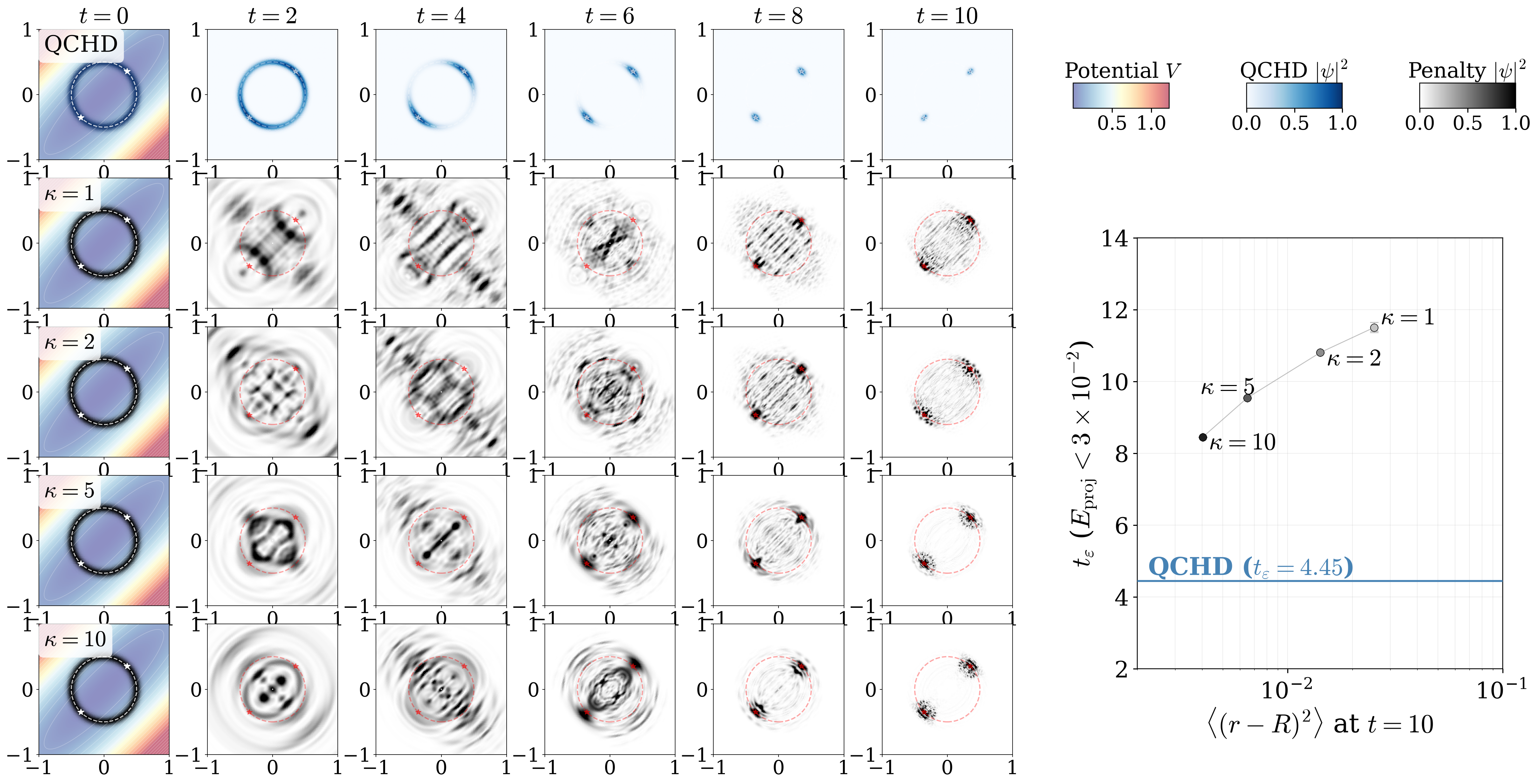}
  \caption{Spherical constraint with a quadratic potential (Sec.~\ref{sec:num-sphere}).
    Left: time evolution of the probability density $|\psi|^2$.
    The top row is QCHD and the second to fifth rows are QHD+Penalty with $\kappa=1,2,5,10$; the columns are $t=0,2,4,6,8,10$.
    The $t=0$ panels show the heat map of the potential $V$ in the background and the initial density as the opacity of the dark overlay;
    the white dashed line is the constraint circle and the stars are the minima of $V|_{\mc{M}_d}$.
    In the $t>0$ panels the shading is $|\psi|^2$ normalized to the maximum of each panel, and the dashed line (white in the QCHD row, red in the QHD+Penalty rows) is the constraint circle.
    The color bars are shown at the upper right.
    Right: Pareto plot. The horizontal axis is the constraint violation $\langle(r-R)^2\rangle$ at $t_{\rm ref}=10$,
    and the vertical axis is the time $t_\varepsilon$ at which $E_{\rm proj}<\varepsilon=3\times10^{-2}$ is reached.
    Since QCHD satisfies the constraint exactly, only its $t_\varepsilon$ is shown as a horizontal line.
    The changes of the points under grid refinement are smaller than the marker size (App.~\ref{app:numerics}).}
  \label{fig:num-sphere-snap}
\end{figure}

\subsection{Spherical constraint with the Ackley potential}
\label{sec:num-ackley}

The constraint manifold is the same circle as in the previous subsection ($R=0.5$), and the potential is replaced by the nonconvex Ackley function
\begin{align}
    V(\bx)&=\frac{1}{L}\,V_{\rm Ackley}\bigl(L\,A_1^{1/2}\bx\bigr)\,,
    \qquad L=8\,,
    \label{eq:V-ackley}
    \\
    V_{\rm Ackley}(\bm{X})
     &=
    -20\exp\Bigl(-0.2\sqrt{\tfrac{1}{2}|\bm{X}|^2}\Bigr)
    -\exp\Bigl(\tfrac{1}{2}\textstyle\sum_{i}\cos(2\pi X^i)\Bigr)
    +20+\ee\,,
\end{align}
where $X^i$ ($i=1,2$) are the components of $\bm{X}=LA_1^{1/2}\bx$, $|\bm{X}|^2=\sum_i (X^i)^2$, and $A_1$ is the matrix of Eq.~\eqref{eq:V-quad} ($R'=-0.9$), which gives the same shear as in the QRHD paper~\cite{2026-AbeNagai}.
The global minimum of $V$ is at the origin, inside the circle, and there are 16 local minima on the circle.
Among them the global minima on the circle are $\theta^*=\pi/4,\,5\pi/4$ ($V^*=0.474$),
and the difference from the neighboring local minima ($\theta^*\pm14^\circ$, $V=0.670$) is small.
The basins of the global minima (the intervals $36^\circ$--$54^\circ$ and $216^\circ$--$234^\circ$ bounded by the neighboring maxima)
occupy 10\% of the circle, so that the uniform initial state has $P_{\rm succ}=0.10$.
The penalty is $\kappa f^2$ with the same $f$ as in Sec.~\ref{sec:num-sphere}, $\kappa=1,2,5,10$, and $\gamma=0.3$.
The main indicator in this example is the failure probability $1-P_{\rm succ}$ evaluated at $t_{\rm ref}=8$, by which time QCHD has converged.
The geometry is the same as in Sec.~\ref{sec:num-sphere}, so that again the large-$\kappa$ effective theory coincides with QCHD up to a phase.

Fig.~\ref{fig:num-ackley-snap} shows the evolution of the probability density.
QCHD concentrates at the global minima by $t=8$ without leaking into local minima ($P_{\rm succ}=0.998$).
QHD+Penalty fails in two different ways:
for $\kappa=1,2$ the global minimum of $V+\kappa f^2$ lies near the origin and more than 30\% of the density leaves the circle toward it,
whereas for $\kappa=5,10$ the density is recaptured near the circle but splits into the local minima during the delayed recapture.
The failure probability at $t=8$ is $2\times10^{-3}$ for QCHD and $0.82$, $0.84$, $0.74$, and $0.56$ for QHD+Penalty with $\kappa=1,2,5,10$:
it decreases with $\kappa$ for $\kappa\ge2$ but remains above $0.5$, more than two orders of magnitude larger than for QCHD.
The failures of QHD+Penalty thus originate from the landscape $V+\kappa f^2$ in the ambient space, which differs from $V|_{\mc{M}_d}$ at finite $\kappa$.
Its global minimum lies off the circle for small $\kappa$, and its local minima on the circle trap the density during the recapture.
QCHD, which moves only on the circle, has neither problem.

\begin{figure}[tbp]
  \centering
  \includegraphics[width=\textwidth]{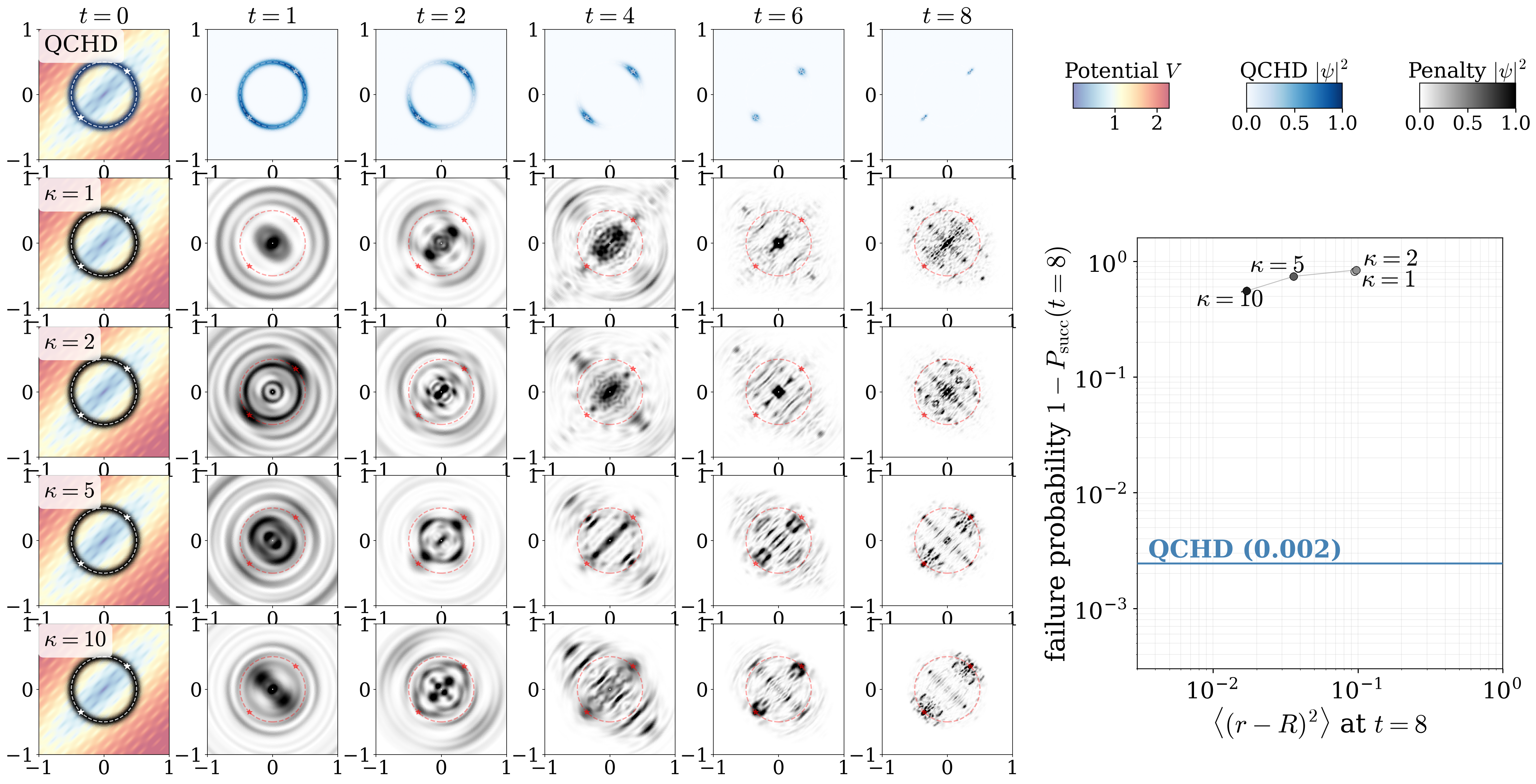}
  \caption{Spherical constraint with the Ackley potential (Sec.~\ref{sec:num-ackley}).
    The layout is the same as in Fig.~\ref{fig:num-sphere-snap}, with the columns $t=0,1,2,4,6,8$.
    In the background of the $t=0$ panels, the blue band along the diagonal is the low-$V$ valley (the global minimum is at the origin).
    In the Pareto plot on the right, the horizontal axis is the constraint violation $\langle(r-R)^2\rangle$ at $t=8$
    and the vertical axis is the failure probability $1-P_{\rm succ}(t=8)$.
    Only the failure probability of QCHD is shown as a horizontal line.}
  \label{fig:num-ackley-snap}
\end{figure}

\subsection{Elliptic constraint with a quadratic potential}
\label{sec:num-ellipse}

We take the ellipse
\begin{align}
    \mc{M}_1: \quad \Bigl(\frac{x}{R_a} \Bigr)^2 + \Bigl( \frac{y}{R_b} \Bigr)^2 =1\,,
    \qquad 
    R_a = 0.6\,,
    \qquad 
    R_b=0.4\,,
\end{align}
as the constraint manifold and the same quadratic potential as in Eq.~(\ref{eq:V-quad}).
The ellipse is represented by the constraint function 
\begin{align}
    f(\bx)=\frac{R_b}{R_a}x^2+\frac{R_a}{R_b}y^2-R_aR_b\,,    
\end{align} 
which reduces to $x^2+y^2-R^2$ for $R_a=R_b=R$.
The minimum on the ellipse is $V^*=0.0110$, attained at the two points with polar angles $\theta^*\simeq46^\circ,\,226^\circ$.
We set $\gamma=0.3$ and $\kappa=1,2,5,10$, and use the penalty $\kappa f^2$.
In this example the normal-bundle connection vanishes ($M=1$), but both $V_\geo$ and $|\nabla f|$ depend on the position.
With the parametrization $\bx=(R_a\cos\theta,R_b\sin\theta)$, the curvature is $k(\theta)=R_aR_b/(R_a^2\sin^2\theta+R_b^2\cos^2\theta)^{3/2}$,
so that $V_\geo=-k^2/(8m)$ varies from $-1.76$ at the ends of the major axis ($\theta=0,\pi$) to $-0.15$ at the ends of the minor axis, and $|\nabla f|=2(R_b^2\cos^2\theta+R_a^2\sin^2\theta)^{1/2}$ varies from $0.8$ to $1.2$ between the same points.
By Sec.~\ref{sec:penalty}, the zero-point energy $\frac12|\nabla f|\sqrt{2\kappa/m}$ of the normal oscillation is therefore position dependent and grows as $\sqrt{\kappa}$, so that the large-$\kappa$ limit of QHD+Penalty does not approach QCHD.
A distance penalty $\kappa d^2$ would remove this term but not $V_\geo$, and normalizing $f$ by $|\nabla f|$ leaves the corrections of footnote~\ref{fn:normalization}.
QCHD involves no such choice, since its Hamiltonian depends only on the shape of the ellipse and on $V|_{\mc{M}_d}$, not on the function $f$ that represents the ellipse.

Fig.~\ref{fig:num-ellipse-snap} shows the evolution of the probability density.
QCHD concentrates on the two minima by $t=8$, whereas QHD+Penalty first bursts out of the ellipse and is recaptured near the minima by $t=10$ for $\kappa\ge5$.
The additional potential in Eq.~\eqref{eq:H_eff-f2} is, for $n_\perp=0$, the zero-point energy $\frac12\Omega^{(f)}_\perp(\bm{q})$, whose variation along the ellipse is $\frac12(|\nabla f|_{\max}-|\nabla f|_{\min})\sqrt{2\kappa/m}=0.2\sqrt{2\kappa}\le0.9$ for $\kappa\le10$.
This variation is smaller than the variation $a(t)\,\Delta V$ of the cost term along the ellipse ($\Delta V\simeq0.23$) for $t\gtrsim t_c$, where $t_c$ is defined by $a(t_c)\,\Delta V=0.2\sqrt{2\kappa}$ and is at most $2.3$ for $\kappa\le10$, so that the additional potential affects the dynamics only at early times.
The variation of $V_\geo/a$ along the ellipse, $1.6/a$, is smaller than $a\,\Delta V$ for $t\gtrsim1.6$.
In the Pareto plot ($\varepsilon=3\times10^{-2}$), $t_\varepsilon$ is 4.37 for QCHD and 10.95, 10.17, 9.62, and 9.09 for QHD+Penalty ($\kappa=1,2,5,10$), i.e., 2.1--2.5 times longer.
As in Sec.~\ref{sec:num-sphere}, $t_\varepsilon$ decreases with $\kappa$ because the recapture becomes faster.
Since the zero-point variation grows as $\sqrt{\kappa}$, $t_c$ grows as $\ln\kappa$ and reaches the convergence time of QCHD ($4.4$) at $\kappa\approx10^2$.
Beyond this value, the additional potential is expected to delay the convergence, so that $t_\varepsilon$ should increase with $\kappa$ instead of decreasing.
Confirming the reversal numerically, and reaching the large-$\kappa$ regime of Sec.~\ref{sec:penalty}, in which the released normal energy $1/(4m\sigma_r^2)=100$ is below $\Omega_\perp^{(f)}$ ($\kappa\gtrsim10^3$), is beyond the present computation.
This example thus shows that the result of QHD+Penalty depends on the form of the penalty, as discussed in Sec.~\ref{sec:penalty}, while QCHD has no such dependence.

\begin{figure}[tbp]
  \centering
  \includegraphics[width=\textwidth]{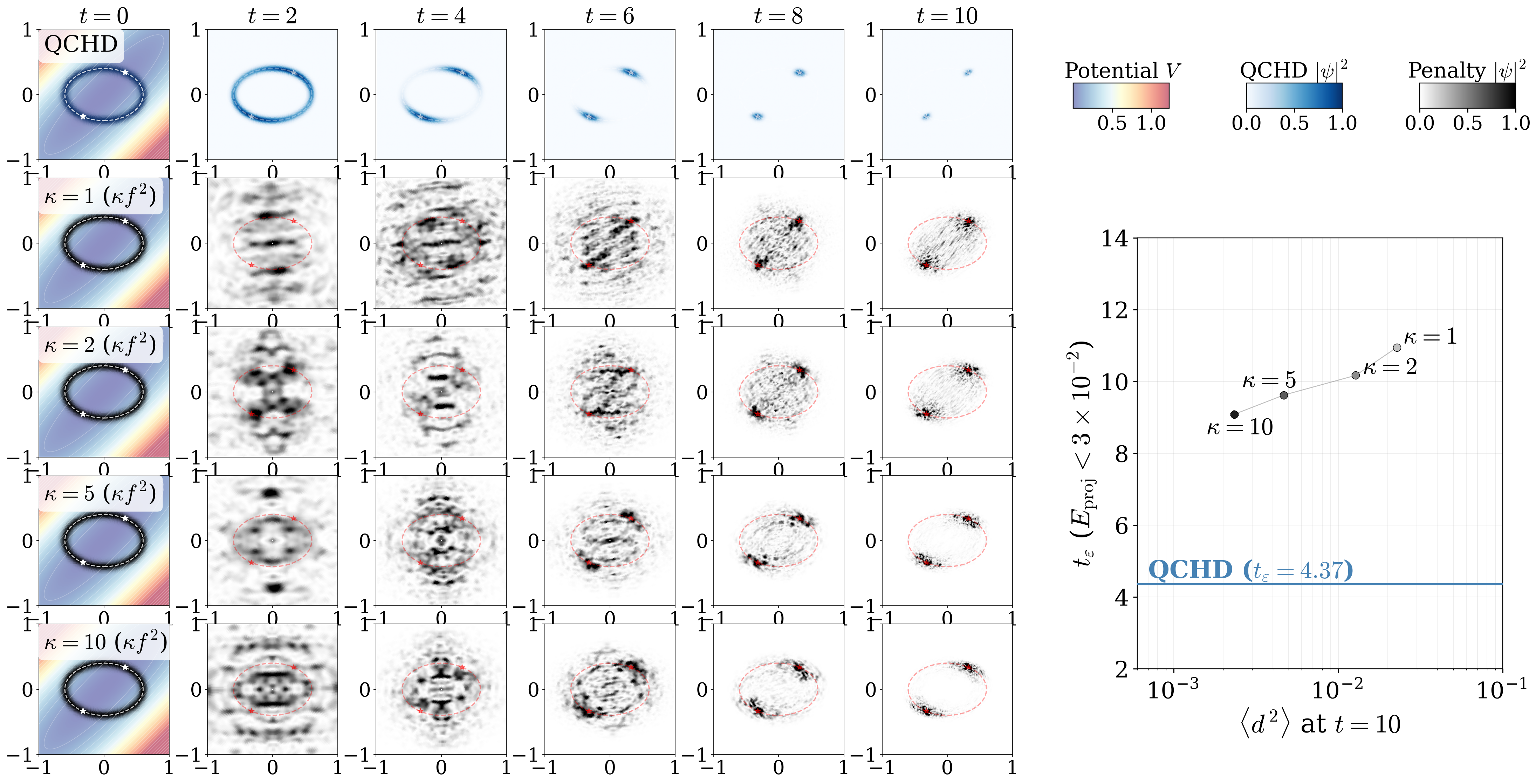}
  \caption{Elliptic constraint with a quadratic potential (Sec.~\ref{sec:num-ellipse}).
    Left: time evolution of $|\psi|^2$. The top row is QCHD and the second to fifth rows are QHD+Penalty with $\kappa=1,2,5,10$; the columns are $t=0,2,4,6,8,10$.
    The background of the $t=0$ panels is the potential $V$, the dashed line is the ellipse, and the stars are the minima of $V|_{\mc{M}_d}$
    (the conventions are the same as in Fig.~\ref{fig:num-sphere-snap}).
    Right: Pareto plot ($\varepsilon=3\times10^{-2}$).
    The horizontal axis is the constraint violation $\langle d^2\rangle$ at $t_{\rm ref}=10$. QCHD is the horizontal line.}
  \label{fig:num-ellipse-snap}
\end{figure}

\subsection{Two constraints with a nonconvex potential}
\label{sec:num-twocons}

Finally, we consider an example with $N=3$ and $M=2$.
We take the circle
\begin{align}
    \mc{M}_1: \quad x + y + z = 1\,,
    \qquad 
    x^2 + y^2 + z^2 =1\,,
\end{align}
the intersection of a plane and a sphere, with center $(1,1,1)/3$ and radius $\sqrt{2/3}$, as the constraint manifold.
It is represented by the two constraint functions
\begin{align}
    f^1(\bm{x}) = x + y + z -1 \,,
    \qquad 
    f^2(\bm{x}) = x^2 + y^2 + z^2 -1\,.
\end{align}
The potential is the nonconvex quartic function
\begin{align}
    V(\bm{x}) = \bm{x}^\TT A_1 \bm{x} + (\bm{x}^\TT A_2 \bm{x})^2\,,
    \qquad 
    A_1 = \begin{pmatrix}
        - 0.5 & -1 & 0.5 \\
        -1 & -0.5 & 0 \\
        0.5 & 0 & 1
    \end{pmatrix}\,,
    \qquad 
    A_2 = \begin{pmatrix}
        0 & 0 &  1 \\
        0 & 1 & 0 \\
        1 & 0 & -0.5
    \end{pmatrix}\,,
    \label{eq:V-twocons}
\end{align}
The circle has three local minima: the global minimum $V^*=-1.442$ at the angle $\varphi^*=89^\circ$ along the circle,
and two others, $V=0.27$ at $195^\circ$ and $V=1.16$ at $285^\circ$.
The basin of the global minimum occupies 53\% of the circle, so that the success probability is a weak indicator in this example (it is already 0.53 for the initial state). 
We therefore compare $t_\varepsilon$.
The penalty is $\kappa[(f^1)^2+(f^2)^2]$ with $\kappa=1,2,5,10$, and we set $\gamma=0.05$, a slower schedule than in the previous examples, which is needed for convergence in this landscape.
In this example the constraint manifold is a planar curve, whose torsion is zero, so that the normal-bundle connection vanishes.
The geometric potential $V_\geo=-k^2/(8m)$ with $k=\sqrt{3/2}$ is constant, and the Gram matrix ($G^{11}=3$, $G^{22}=4$, $G^{12}=2$) is constant along the circle, so that the zero-point energy is independent of the position on the circle.
The large-$\kappa$ effective theory of Sec.~\ref{sec:penalty} therefore coincides with QCHD up to a constant, and the differences found below are finite-$\kappa$ effects.

\begin{figure}[tbp]
  \centering
  \includegraphics[width=\textwidth]{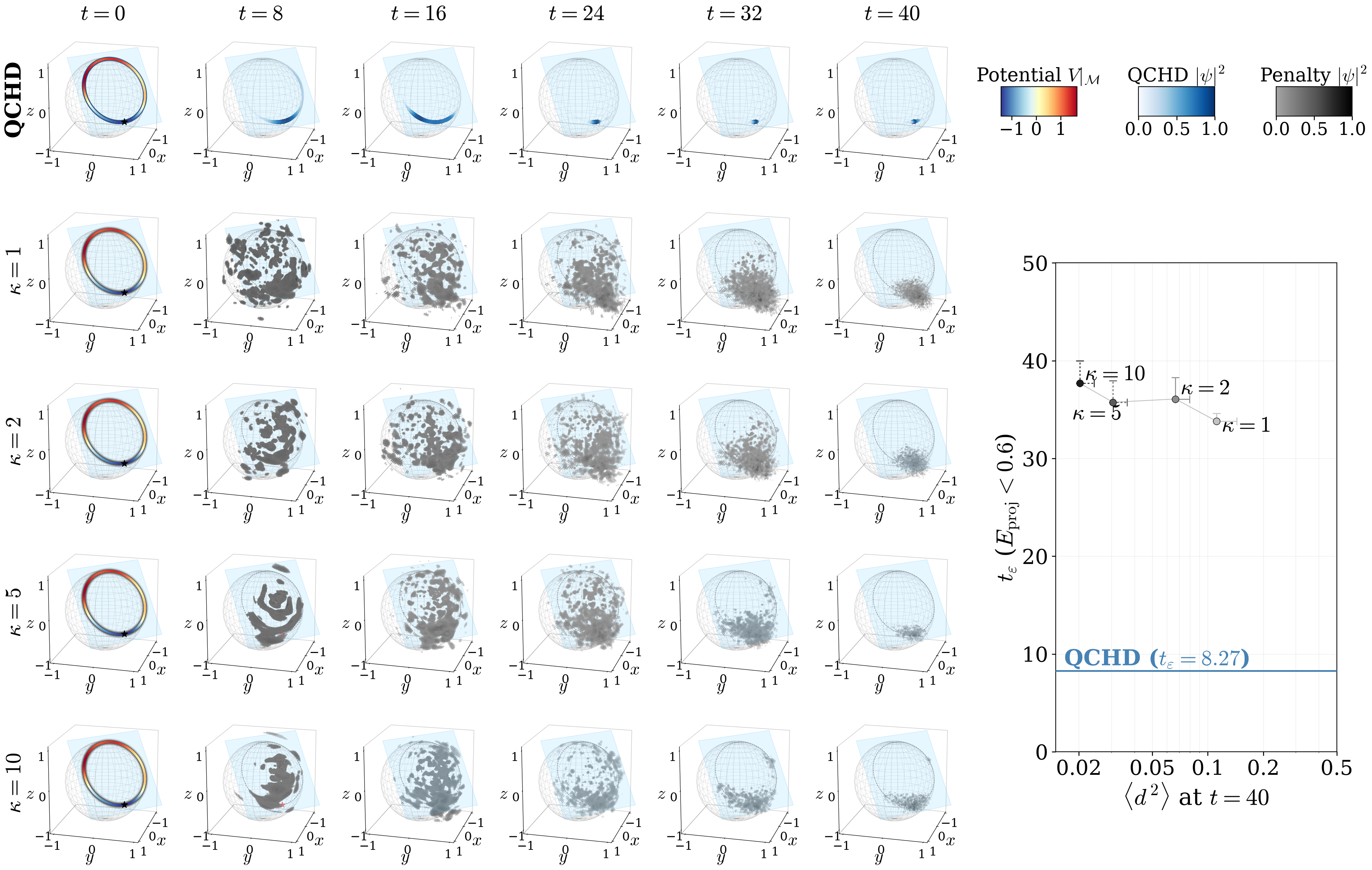}
  \caption{Two constraints in $\mathbb{R}^3$ with a nonconvex potential (Sec.~\ref{sec:num-twocons}).
    Left: time evolution of $|\psi|^2$. The columns are $t=0,8,16,24,32,40$; the top row is QCHD and the second to fifth rows are QHD+Penalty with $\kappa=1,2,5,10$.
    In the $t=0$ panels the potential $V|_{\mc{M}_d}$ on the circle is shown as the color of the circle on top of the (translucent) initial tube
    (color bar at the upper right; the star is the minimum).
    The light blue plane is $x+y+z=1$ and the gray sphere is $x^2+y^2+z^2=1$.
    In the QHD+Penalty rows, voxels whose density exceeds 5\% of the maximum are shown in gray scale
    (color bar at the upper right, normalized to the maximum of each panel).
    Right: Pareto plot ($\varepsilon=0.6$). The horizontal axis is the squared-distance to the circle $\langle d^2\rangle$ at $t_{\rm ref}=40$. QCHD is the horizontal line.
    The error bars are estimates of the systematic error due to the finite computational domain (App.~\ref{app:numerics}):
    for $\kappa=1,2$ they are the changes from the original to the enlarged domain (solid), and for $\kappa=5,10$ the relative change of $\kappa=2$ is assumed (dashed).
    The converged values are expected to lie in the direction of the bars.}
  \label{fig:num-twocons-snap}
\end{figure}

Fig.~\ref{fig:num-twocons-snap} shows the evolution of the probability density in ambient space $\mathbb{R}^3$.
QCHD concentrates on an arc at $t=8$ and on a single point at $t=24$.
In QHD+Penalty the motion released in the two normal directions spreads the density over almost the whole computational domain at $t=8$, and the recapture near the circle is delayed until $t\simeq24$--$32$ even for $\kappa=10$; for $\kappa=1,2$ a cloud remains at $t=40$ around the minimum of $V+\kappa\sum_A(f^A)^2$, which lies off the circle (at normal distance 0.27 for $\kappa=1$).
In the Pareto plot ($\varepsilon=0.6$; the initial value is $E_{\rm proj}(0)=1.76$),
$t_\varepsilon$ is 8.3 for QCHD and 33.8, 36.1, 35.8, and 37.7 for QHD+Penalty with $\kappa=1,2,5,10$, i.e., more than 4 times longer.
Increasing $\kappa$ reduces the constraint violation at $t_{\rm ref}=40$ from 0.11 to 0.020, but $t_\varepsilon$ stays between 34 and 38.
For large $\kappa$ the density converges while confined near the circle, whereas for small $\kappa$ it converges while retaining a component off the circle.
In this example the off-circle component happens to lower $E_{\rm proj}$, which does not count the normal displacement, so that the threshold is reached at similar times for all $\kappa$.
The large-$\kappa$ regime of Sec.~\ref{sec:penalty}, in which the released normal energy $1/(4m\sigma_r^2)=100$ is below the normal frequencies $\Omega_a\simeq1.7\sqrt{\kappa}$ and $3.3\sqrt{\kappa}$, requires $\kappa\gtrsim10^3$ and is beyond the present computation.

\paragraph{Summary of the numerical results.}

Tab.~\ref{tab:num-summary} summarizes the results of the four examples.
In all examples QCHD converges faster than QHD+Penalty while satisfying the constraints exactly.
The performance of QHD+Penalty depends on $\kappa$, and the dependence varies from example to example:
large $\kappa$ is favorable on the sphere and for the Ackley potential, where the failure probability nevertheless remains above $0.5$ for $\kappa\le10$,
and increasing $\kappa$ does not speed up the convergence with two constraints.
QCHD requires no such tuning parameter.

\begin{table}[tbp]
  \centering
  {\small
  \begin{tabular}{llccc}
    \hline
    Example & Indicator & QCHD & QHD+Penalty & Ratio \\
    \hline
    (1) sphere, quadratic & $t_\varepsilon$ ($\varepsilon=3\times10^{-2}$) & 4.5 & 8.5--11.5 & 1.9--2.6 \\
    (2) sphere, Ackley & failure probability ($t=8$) & $2\times10^{-3}$ & 0.56--0.84 & $>200$ \\
    (3) ellipse, quadratic & $t_\varepsilon$ ($\varepsilon=3\times10^{-2}$) & 4.4 & 9.1--11.0 & 2.1--2.5 \\
    (4) two constraints, quartic & $t_\varepsilon$ ($\varepsilon=0.6$) & 8.3 & 33.8--37.7 & 4.1--4.6 \\
    \hline
  \end{tabular}}
  \caption{Comparison of QCHD and QHD+Penalty in the four numerical examples.
    $t_\varepsilon$ is the time to reach $E_{\rm proj}<\varepsilon$, and the failure probability is $1-P_{\rm succ}$ at $t=8$.
    The values for QHD+Penalty are the ranges over $\kappa=1,2,5,10$.
    Ratio is the value for QHD+Penalty divided by that for QCHD.}
  \label{tab:num-summary}
\end{table}

\section{Conclusions}
\label{sec:conclusions}

In this paper, we have proposed Quantum Constrained Hamiltonian Descent (QCHD), a coherent quantum approach to continuous optimization under equality constraints. 
Starting from a classical formulation in which Lagrange multipliers enforce the velocity constraints $\dot{f}^A=0$, we applied Dirac’s method and constructed a quantum Hamiltonian whose kinetic term is the Laplace–Beltrami operator on the constraint manifold $\mc{M}_d$. 
Introducing the time-dependent factor $a(t)$ of QHD then defines the QCHD dynamics. For an initial state on $\mc{M}_d$, the evolution preserves the constraints exactly without introducing penalty terms. 
In local coordinates, the resulting Hamiltonian takes the form of QRHD, establishing QCHD as a realization of Riemannian quantum optimization constructed directly from equality constraints.

We have also clarified the relation between QCHD and QHD+Penalty.
For a strong penalty, the wave function of QHD+Penalty is squeezed into a thin neighborhood of $\mc{M}_d$, and its motion separates into an oscillation perpendicular to $\mc{M}_d$ and a motion along $\mc{M}_d$.
We derived the effective Hamiltonian for the motion along $\mc{M}_d$ in the strong penalty limit ($\kappa\to \infty$), assuming that the penalty is the squared-distance from $\mc{M}_d$, that the perpendicular oscillation stays within the thin neighborhood, and that $a(t)$ varies slowly compared with this oscillation.
The effective Hamiltonian differs from that of QCHD by the geometric potential $V_\geo$, which is determined by the extrinsic curvature of $\mc{M}_d$, and, for two or more constraints, by a coupling to the normal-bundle connection, which describes how the perpendicular directions rotate along $\mc{M}_d$.
Both terms survive the limit $\kappa\to\infty$.
Thus QCHD, which quantizes the dynamics on $\mc{M}_d$, and QHD+Penalty, which keeps the quantum dynamics in $\mathbb{R}^N$ near $\mc{M}_d$ by a penalty, lead to different quantum dynamics in general, even though they describe the same classical constrained motion in this limit.
The two agree up to an overall phase when $V_\geo$ is constant and the connection term is absent, as for a sphere.
For the penalty $\kappa\sum_A(f^A)^2$ used in practice, the zero-point energy of the perpendicular oscillation adds a position-dependent potential whenever $|\nabla f|$ varies along $\mc{M}_d$.
This potential grows as $\sqrt{\kappa}$ and is not removed by normalizing the constraint functions by their gradients, so that the result of QHD+Penalty depends on the functions $f^A$ chosen to represent $\mc{M}_d$.
The QCHD Hamiltonian, which is determined by $\mc{M}_d$ and $V|_{\mc{M}_d}$ alone, has no such dependence.

Our numerical simulations compared QCHD with QHD+Penalty in four examples involving spherical and elliptic constraints, two simultaneous constraints, and both quadratic and nonconvex cost functions. For the initial states and penalty strengths considered, QCHD satisfied the constraints exactly and reached the convergence threshold 1.9--4.6 times faster than QHD+Penalty in the three examples compared by $t=t_\varepsilon$. 
In the spherical Ackley example, its failure probability was smaller by more than two orders of magnitude. The dependence of QHD+Penalty on the penalty strength varied across the examples, and increasing the strength did not always improve its performance. 
These results imply that the effectiveness of enforcing equality constraints through the kinetic operator in the finite-penalty regime explored here.

Several questions remain for the development of QCHD as an optimization algorithm on a quantum computer. 
Efficient preparation of an initial state on $\mc{M}_d$ is an essential ingredient for QCHD since the dynamics keeps the state on $\mc{M}_d$ once such a state has been supplied. 
In order to show the actual advantages of QCHD, we need to estimate the resource and cost of quantum circuit realizations including the implementation of the Laplace--Beltrami operator associated with general constraint functions. 
Further numerical studies should examine a wider range of geometries and systematically vary the number of constraints, the initial state, and the evolution schedule. 
Comparisons with classical Riemannian optimization methods and other constrained quantum approaches will help identify the problems for which QCHD is most effective. 
Extending the framework to inequality constraints, potentially through barrier-induced metrics, is another natural direction. 

\section*{Acknowledgments}
\noindent
 The authors thank Naoki Yamamoto and Katsuaki Ishigaki for helpful comments and discussions.
This work is supported by MEXT Quantum Leap Flagship Program Grant No. JPMXS0120319794.
Y.A. acknowledges support by JST Grant No. JPMJPF2221.
In preparing this work, the authors used the generative AI tools Claude Code (Anthropic) with the Claude Fable 5.1 model and Codex (OpenAI) with the GPT-6 Astra model to check the equations, to implement and run the numerical calculations, and to proofread the manuscript.
The authors take full responsibility for the content of the manuscript.

\appendix

\section[Hamiltonian formalism of constrained systems and the Dirac bracket]{Hamiltonian formalism of constrained systems\\ and the Dirac bracket}
\label{app:hamiltonian}

In this appendix, we apply Dirac's procedure~\cite{1950-Dirac-CJM,1958-Dirac-PRSLA,1964-Dirac} to the Lagrangian \eqref{eq:classical-L}
\begin{align}
    L = \frac{m}{2} \delta_{ij} \dot{x}^i \dot{x}^j - V(\bm{x}) + \lambda_A \dot{f}^A(\bm{x})\,,
\end{align}
and derive the Dirac brackets \eqref{eq:Dirac-brakect-xp} and the Hamiltonian \eqref{eq:H-classical}.

\paragraph{Primary constraint.}

Let us begin with the Lagrangian \eqref{eq:classical-L}.
The canonical momentum of $x^i$ is given by \eqref{eq:classical-p}, but the canonical momentum of the multiplier $\lambda_A$ vanishes automatically because of the lack of $\dot{\lambda}_A$
\begin{align}
    \pi^A_\lambda \coloneqq \frac{\del L}{\del \dot{\lambda_A}} = 0\,.
\end{align}
The primary constraint for \eqref{eq:classical-L} then becomes 
\begin{align}
    \phi^A_1 \coloneqq \pi^A_\lambda \approx 0\,.
    \label{eq:primary-constraint-app}
\end{align}
The Hamiltonian of a system with a primary constraint is ambiguous in the terms involving $\dot{\lambda}_A$, so we write the Hamiltonian by introducing new multipliers $\mu_A$ as 
\begin{align}
    \tilde{H} = H' + \mu_A \phi^A_1 = H' + \mu_A \pi^A_\lambda\,,
    \quad 
    H' =  p_i \dot{x}^i - L\,,
\end{align}
where $\mu_A$ is also regarded as $\dot{\lambda}_A$.
Using Eqs.~\eqref{eq:classical-L} and \eqref{eq:classical-p}, this Hamiltonian is rewritten as 
\begin{align}
    \tilde{H} = \frac{1}{2m} \delta^{ij}(p_i - \lambda_A n^A{}_i)(p_j - \lambda_B n^B{}_j) + V(\bm{x}) + \mu_A \pi^A_\lambda\,.
    \label{eq:tilde-H-derivation}
\end{align}
We note that this Hamiltonian does not depend on $\dot{x}^i$ in spite of the singular system.

\paragraph{Secondary constraint.}

Although we get the Hamiltonian, we have issues that need to be considered.
The consistency of the primary constraint \eqref{eq:primary-constraint-app} with the time evolution requires the conditions of 
\begin{align}
    \dot{\phi}^A_1 = 
    \{\phi^A_1, H'\}_P = 
    - \frac{\del H'}{\del \lambda_A} = \frac{1}{m} n^{Ai}(p_i - \lambda_B n^B{}_i) \approx 0\,.
    \label{eq:secondary-constraint}
\end{align}%
This new constraint is called a secondary constraint:
\begin{align}
    \phi^A_2 \coloneqq n^{Ai}(p_i - \lambda_B n^B{}_i) = n^{Ai}p_i - G^{AB} \lambda_B \approx 0\,.
\end{align}
This constraint means that $p_i - \lambda_A n^A{}_i$ has no normal component with respect to the constraint manifold.
If we impose the constraints, we then find that $p_i -\lambda_A n^A{}_i$ is rewritten as the momentum in the tangential direction:
\begin{align}
    p_i - \lambda_A n^A{}_i \approx M^j{}_i (p_j - \lambda_A n^A{}_j) = M^j{}_i p_j = P_i\,.
\end{align}
Using this equation, the Hamiltonian of this constraint system becomes from \eqref{eq:tilde-H-derivation}
\begin{align}
    H \approx \frac{1}{2m} \delta^{ij} P_i P_j + V(\bm{x})\big|_{\mc{M}_d} 
    =\frac{1}{2m} M^{ij} p_i p_j + V(\bm{x}) \big|_{\mc{M}_d}\,, 
\end{align}
which is \eqref{eq:H-classical} we show in Sec.~\ref{sec:classical}.

The consistency of $\phi^A_2$ with the time evolution does not provide any additional constraints and gives the equation to determine the multiplier $\mu_A$:
\begin{align}
    \dot{\phi}^A_2 = \{\phi^A_2,\tilde{H}\}_P = \{\phi^A_2, H'\} + \{ \phi^A_2,\phi^B_1\} \mu_B \approx 0\,,
\end{align}
and the multiplier is expressed as 
\begin{align}
    \mu_A \approx G_{AB} \left[ \frac{1}{m} P^i P^j \del_i \del_j f^B - n^{Bi} \del_i V \right]\,.
\end{align}
We then find that the constraints in the Dirac's sense are $\phi^A_1$ and $\phi^A_2$.

\paragraph{Dirac bracket.}

Since all constraints have now been identified, let us consider the Dirac brackets \eqref{eq:dirac-bracket-def}.
The correction coming from the constraints are the matrix $C_{\ms{AB}} = \{ \phi_{\ms{A}}, \phi_{\ms{B}}\}_P$ with $\phi_{\ms{A}} = (\phi^A_1, \phi^A_2)^{\mr{t}}$.
This matrix elements are evaluated as 
\begin{align}
    C = \begin{pmatrix}
        \{\phi^A_1, \phi^B_1\}_P & \{\phi^A_1, \phi^B_2\}_P
        \\
        \{\phi^A_2, \phi^B_1\}_P & \{\phi^A_2, \phi^B_2\}_P
    \end{pmatrix}
    = \begin{pmatrix}
        0 & G^{AB}
        \\
        - G^{AB} & S^{AB}
    \end{pmatrix}\,,
\end{align}
where $G^{AB}$ is the Gram matrix.
In this calculation, we use 
\begin{align}
    \{\phi^A_1, \phi^B_2\}_P = \{\pi^A_\lambda, - G^{BC} \lambda_C\}_P = G^{AB}\,.
\end{align}
The Poisson bracket between $\phi^A_2$'s becomes
\begin{align}
    S^{AB} \coloneqq \{ \phi^A_2 , \phi^B_2\}_P = (n^{Bj} \del_j n^{Ai} - n^{Aj} \del_j n^{Bi}) (p_i- \lambda_C n^C{}_i)\,.
\end{align}
This non-vanishing result comes from the non-commutativity among the normal vectors defined by the different constraint conditions.

Since $\det C = (\det G^{AB})^2 \neq 0$, all of constraints $\phi^A_{1,2}$ are second-class constraints and the Dirac bracket \eqref{eq:dirac-bracket-def} is well-defined.
The inverse matrix of $C$ is 
\begin{align}
    C^{-1} = \begin{pmatrix}
        G^{-1} S G^{-1} & - G^{-1} \\
        G^{-1} & 0
    \end{pmatrix}\,.
\end{align}
We note that the lower-right block vanishes $(C^{-1})_{22} = 0$, and the non vanishing Poisson bracket among the canonical variables and the constraints are $\{x^i,\phi^A_2\}_P$ and $\{p_i, \phi^A_2\}_P$.
Because the canonical variables $(x^i,p_i)$ commute with the conjugate momentum of the multiplier $\pi^A_\lambda = \phi^A_1$, the non-trivial corrections in \eqref{eq:dirac-bracket-def} could come only from terms with $(C^{-1})_{22}$, which vanishes identically.
We thus find that the Dirac brackets of the canonical variables are the same as the Poisson brackets,
\begin{align}
    \{x^i,p_j\}_D = \{x^i,p_j\}_P = \delta^i{}_j\,,
    \qquad 
    \{x^i,x^j\}_D = \{p_i, p_j\}_D = 0\,.
\end{align}
These simplified Dirac brackets result from the multiplier term $\lambda_A \dot{f}^A$ in the Lagrangian \eqref{eq:classical-L}.
If we take the multiplier term as $\lambda_A' f^A$, the set of constraints changes and the Dirac brackets become more complicated \cite{1990-HommaEtAl-ZPC-PaF}.

It would be instructive to consider the Dirac brackets with the tangential momentum $P_i$ instead of the conjugate momentum $p_i$.
Since these are related by the projector as $P_i = M^j{}_i p_j$, the Dirac bracket with $\bm{x}$ becomes 
\begin{align}
    \{x^i, P_j\}_D = M^i{}_j\,.
\end{align}
In addition, this $P_i$ contains the $\bm{x}$ dependence through the projector and the Dirac bracket between $\bm{P}$'s is no longer trivial:
\begin{align}
    \{P_i, P_j\}_D = G_{AB} (n^A{}_j M^k{}_i - n^A{}_i M^k{}_j) (\del_k \del_l f^B) P^l\,.
    \label{eq:PP-dirac-app}
\end{align}

\section{Geometric formulation for constrained quantum mechanics}
\label{app:geometry}

In this appendix, we summarize our geometric notation and show some geometric properties used in Secs.~\ref{sec:constrained} and \ref{sec:penalty}.
Our notation for the coordinates and metrics is shown in Tab.~\ref{tab:coordinates-tab}, and the geometric setting is illustrated in Fig.~\ref{fig:indices}.

\begin{table}[t]
    \centering
    \begin{tabular}{ccccc}
        \toprule
         Space & Dimension & Coordinate & Metric & Sec./App. \\
         \midrule 
         \multirow{3}{*}{$\mbb{R}^N$} & \multirow{3}{*}{$N$} & $x^i$ & $\delta_{ij}$ & Main\\
         & & $y^I=(q^\alpha,\eta^A)$ & $\mc{G}_{IJ}$ & App.~\ref{app:geoprp-constraint-adapted}\\
         & & $z^\mu = (q^\alpha, \xi^a)$ & $h_{\mu\nu}$ & App.~\ref{app:geoprp-tubular}\\
         \hline
         $\mc{M}_d$ & $d=N-M$ & $q^\alpha$ & $g_{\alpha\beta}$ & Main \\
         \hline
         \multirow{2}{*}{$N_q \mc{M}_d$} & \multirow{2}{*}{$M$} & $\eta^A = f^A$ & $G_{AB}$ & App.~\ref{app:geoprp-constraint-adapted}\\
         & & $\xi^a$ & $\delta_{ab}$ & Sec.~\ref{sec:penalty}, App.~\ref{app:geoprp-tubular} \\
         \bottomrule
    \end{tabular}
    \caption{A table of coordinates and their notation we use in this paper.
    $N_{\bm{q}} \mc{M}_d$ denotes the normal space to $\mc{M}_d$ at $\bm{q}$.}
    \label{tab:coordinates-tab}
\end{table}

\subsection{Preliminary}

In the $N$-dimensional Euclidean space $\mbb{R}^N$ with coordinates $x^i$, we consider $M$ independent constraints
\begin{align}
    f^A(\bm{x}) = 0\,,
\end{align}
and submanifold characterized by
\begin{align}
    \mc{M}_d = \{\bm{x} \in \mbb{R}^N | f^A(\bm{x}) =0 \}\,.
\end{align}
We call this the constraint manifold.
The normal vector and the Gram matrix are written as 
\begin{align}
    n^A{}_i = \del_i f^A\,,
    \qquad 
    G^{AB} = \mr{Gram}^{AB} = \delta^{ij} n^A{}_i n^B{}_j\,.
\end{align}
Using these normal vectors and the Gram matrix, the projector onto the constrained manifold is given by 
\begin{align}
    M^i{}_j = \delta^i_j - n^{Ai} G_{AB} n^B{}_j\,.
\end{align}

We write the local coordinate on the constraint manifold $\mc{M}_d$ as $q^\alpha$ and the induced metric as $g_{\alpha\beta}$.
The coordinate embedding is written by 
\begin{align}
    x^i = X^i(\bm{q})\,,
\end{align}
and the frame field is defined by 
\begin{align}
    e^i{}_\alpha(\bm{q}) = \del_\alpha X^i(\bm{q})\,,
    \qquad
    \delta_{ij} e^i{}_\alpha e^j{}_\beta = M_{ij} e^i{}_\alpha e^j{}_\beta = g_{\alpha\beta}\,.
\end{align}

\subsubsection*{Constraint-adapted coordinates}

In the neighborhood of the constraint manifold, we take the constraints as coordinates $\eta^A = f^A$ and consider the $N$-dimensional coordinates $y^I = (q^\alpha, \eta^A)$, which we call constraint-adapted coordinates.
In this coordinate, the constraint surface is just given by $\eta^A = 0$.
The relation between $x^i$ and $y^I$ is assumed to be given by\footnote{
We note that this $\mc{X}^i(\bm{q},\bm{\eta})$ is different from the embedding denoted by $X^i(\bm{q})$.
}
\begin{align}
    x^i = \mc{X}^i(\bm{q}, \bm{\eta})\,,
\end{align}
and the tangent vector is 
\begin{align}
    E^i{}_\alpha(\bm{q},\bm{\eta}) = \del_\alpha \mc{X}^i(\bm{q},\bm{\eta})\,.
\end{align}
These are related to the local coordinate fields by 
\begin{align}
    X^i(\bm{q}) = \mc{X}^i(\bm{q}, \bm{\eta}=\bm{0})\,,
    \qquad 
    e^i{}_\alpha(\bm{q}) = E^i{}_\alpha(\bm{q},\bm{\eta}=\bm{0})\,.
\end{align}
The local coordinate metric is also given by
\begin{align}
    \gamma_{\alpha\beta}(\bm{q},\bm{\eta}) = \delta_{ij} E^i{}_\alpha E^j{}_\beta\,,
    \qquad 
    g_{\alpha\beta}(\bm{q}) = \gamma_{\alpha\beta}(\bm{q},\bm{\eta}=\bm{0})\,.
    \label{eq:relation-gamma-g}
\end{align}

The $N$-dimensional metric tensors are introduced by
\begin{align}
    \dd s^2 = \delta_{ij} \dd x^i \dd x^j = \mc{G}_{IJ} \dd y^I \dd y^J\,,
    \qquad 
    \mc{G}_{IJ} =\del_I \mc{X}^i \del_J \mc{X}_i\,.
    \label{eq:M-dim-metric-mathcalGIJ}
\end{align}
It is noted that from the definition~\eqref{eq:M-dim-metric-mathcalGIJ}, this metric satisfies
\begin{align}
    \mc{G}_{\alpha\beta} =\del_\alpha \mc{X}^i \del_\beta \mc{X}_i = \gamma_{\alpha\beta}\,,
    \qquad 
    \mc{G}^{AB} = \delta^{ij} \frac{\del \eta^A}{\del x^i} \frac{\del \eta^B}{\del x^j} = G^{AB}\,.
\end{align}

Here, we also defined
\begin{align}
    F^i{}_A(\bm{q},\bm{\eta}) = \del_A \mc{X}^i\,,
\end{align}
which satisfies 
\begin{align}
    n^A{}_i F^i{}_B = \delta^A_B\,.
    \label{eq:nF-delta}
\end{align}
The relation between $F^i{}_A$ and $n_A{}^i = \delta^{ij} G_{AB} n^B{}_j$ is given by the follows.

\begin{thm}
$F^i{}_A$ and $n_A{}^i$ are related by 
\begin{align}
    F^i{}_A = \del_A \mc{X}^i = n_A{}^i + E^i{}_\alpha \gamma^{\alpha\beta} \mc{G}_{\beta A}\,.
    \label{eq:relation-FiA-niA}
\end{align}
\end{thm}

\proof

$F^i{}_A - n^i{}_A$ is tangent as
\begin{align}
    n^B{}_i(F^i{}_A - n_A{}^i) = \delta^B_A - \delta^B_A =0\,
\end{align}
where we use \eqref{eq:nF-delta} and we parametrize it as $F^i{}_A - n_A{}^i = E^i{}_\alpha \ms{V}^\alpha{}_A$.
From the calculation of
\begin{align}
    \delta_{ij} E^i{}_\beta (F^j{}_A - n_A{}^j) = E_i{}_\beta E^i{}_\alpha \ms{V}^\alpha{}_A = \gamma_{\beta \alpha} \ms{V}^\alpha{}_A\,,
\end{align}
and 
\begin{align}
    \delta_{ij} E^i{}_\beta (F^j{}_A - n_A{}^j) = \delta_{ij} E^i{}_\beta F^j{}_A = \delta_{ij} \frac{\del \mc{X}^i}{\del q^\beta} \frac{\del \mc{X}^j}{\del \eta^A} = \mc{G}_{\beta A}\,,
\end{align}
we find that
\begin{align}
    \ms{V}^\alpha{}_A = \gamma^{\alpha\beta} \mc{G}_{\beta A}\,,
\end{align}
and 
\begin{align}
    F^i{}_A = n_A{}^i + E^i{}_\alpha \gamma^{\alpha\beta} \mc{G}_{\beta A}\,.
\end{align}
\qed

\begin{thm}
In constraint-adapted coordinates, the ambient metric and its inverse have the following block forms:
\begin{align}
    \mc{G}_{IJ} = \begin{pmatrix}
        \gamma_{\alpha\beta} & \mc{G}_{\alpha B} \\
        \mc{G}_{A\beta} & G_{AB} + \mc{G}_{A\alpha} \gamma^{\alpha\beta} \mc{G}_{\beta B}
    \end{pmatrix}\,,
    \quad 
    \mc{G}^{IJ} = \begin{pmatrix}
        \gamma^{\alpha\beta} + \gamma^{\alpha \gamma} \mc{G}_{\gamma A} G^{AB} \mc{G}_{B \delta} \gamma^{\delta \beta} & - \gamma^{\alpha \gamma} \mc{G}_{\gamma C} G^{CB}
        \\
        - G^{AC} \mc{G}_{C\gamma} \gamma^{\gamma \beta} & G^{AB}
    \end{pmatrix}\,.
    \label{eq:mathcalG_IJ-mathcalG^IJ}
\end{align}
\end{thm}

\proof

From \eqref{eq:relation-FiA-niA}, $\dd x^i$ is expressed as
\begin{align}
    \dd x^i = E^i{}_\alpha \dd q^\alpha + F^i{}_A \dd \eta^A
    = E^i{}_\alpha \dd q^\alpha +(n_A{}^i + E^i{}_\alpha \gamma^{\alpha \beta} \mc{G}_{\beta A}) \dd \eta^A\,,
\end{align}
and we find that the metric has the following form by substituting this $\dd x^i$
\begin{align}
    \dd s^2 &= \delta_{ij} E^i{}_\alpha E^j{}_\beta (\dd q^\alpha + \gamma^{\alpha \gamma}\mc{G}_{\gamma A} \dd \eta^A)(\dd q^\beta + \gamma^{\beta \delta} \mc{G}_{\delta B} \dd \eta^B) + \delta_{ij} n_A{}^i n_B{}^j \dd \eta^A \dd \eta^B
    \nn \\
    &= \gamma_{\alpha\beta} (\dd q^\alpha + \gamma^{\alpha \gamma}\mc{G}_{\gamma A} \dd \eta^A)(\dd q^\beta + \gamma^{\beta \delta} \mc{G}_{\delta B} \dd \eta^B) + G_{AB} \dd \eta^A \dd \eta^B\,.
\end{align}
The metric components are then read as
\begin{align}
    \mc{G}_{IJ} = \begin{pmatrix}
        \gamma_{\alpha\beta} & \mc{G}_{\alpha B} \\
        \mc{G}_{A\beta} & G_{AB} + \mc{G}_{A\alpha} \gamma^{\alpha\beta} \mc{G}_{\beta B}
    \end{pmatrix}\,.
\end{align}

This metric can be decomposed as 
\begin{align}
    \mc{G}_{IJ} = \begin{pmatrix}
        \delta^\gamma{}_\alpha & 0\\
        \mc{G}_{A\delta} \gamma^{\delta \gamma} & \delta^C{}_A
    \end{pmatrix} \begin{pmatrix}
        \gamma_{\gamma \epsilon} & 0 \\
        0 & G_{CD}
    \end{pmatrix} \begin{pmatrix}
        \delta^\epsilon{}_\beta & \gamma^{\epsilon\zeta} \mc{G}_{\zeta B} \\
        0 & \delta^D{}_B
    \end{pmatrix}\,,
    \label{eq:mathcalG-decomposition}
\end{align}
we thus obtain the inverse of $\mc{G}_{IJ}$ as 
\begin{align}
    \mc{G}^{IJ} = \begin{pmatrix}
        \gamma^{\alpha\beta} + \gamma^{\alpha \gamma} \mc{G}_{\gamma A} G^{AB} \mc{G}_{B \delta} \gamma^{\delta \beta} & - \gamma^{\alpha \gamma} \mc{G}_{\gamma C} G^{CB}
        \\
        - G^{AC} \mc{G}_{C\gamma} \gamma^{\gamma \beta} & G^{AB}
    \end{pmatrix}\,.
\end{align}

\qed

From the decomposition of \eqref{eq:mathcalG-decomposition}, the determinant is written as 
\begin{align}
    \det \mc{G} = \det \gamma \times \det G\,.
    \label{eq:decompose-detmathcalG}
\end{align}

If we write the off-diagonal component of $\mc{G}_{IJ}$ as a vector like
\begin{align}
    \Theta^\alpha_A \coloneqq \gamma^{\alpha\beta} \mc{G}_{\beta A} = \gamma^{\alpha\beta} \delta_{ij} E^i{}_\beta \del_A \mc{X}^j\,,
\end{align}
the metric becomes
\begin{align}
    \dd s^2 = \gamma_{\alpha \beta}(\dd q^\alpha + \Theta^\alpha_A \dd \eta^A)( \dd q^\beta + \Theta^\beta_B \dd \eta^B) + G_{AB} \dd \eta^A \dd \eta^B\,,
\end{align}
and the components are written as 
\begin{align}
    \mc{G}_{IJ} = \begin{pmatrix}
        \gamma_{\alpha \beta} & \gamma_{\alpha \gamma} \Theta^\gamma_B \\
        \gamma_{\beta \gamma} \Theta^\gamma_A & G_{AB} + \gamma_{\gamma \delta} \Theta^\gamma_A \Theta^\delta_B
    \end{pmatrix}\,,
    \qquad 
    \mc{G}^{IJ} = \begin{pmatrix}
        \gamma^{\alpha\beta} + \Theta^\alpha{}_A G^{AB} \Theta^\beta{}_B & - \Theta^\alpha{}_C G^{CB} \\
        - G^{AC} \Theta^\beta{}_C & G^{AB}
    \end{pmatrix}\,.
    \label{eq:mathcalG_IJ-mathcalG^IJ-2}
\end{align}

\subsubsection*{Tubular coordinates}

\begin{figure}
    \centering
    \includegraphics[width=0.62\textwidth]{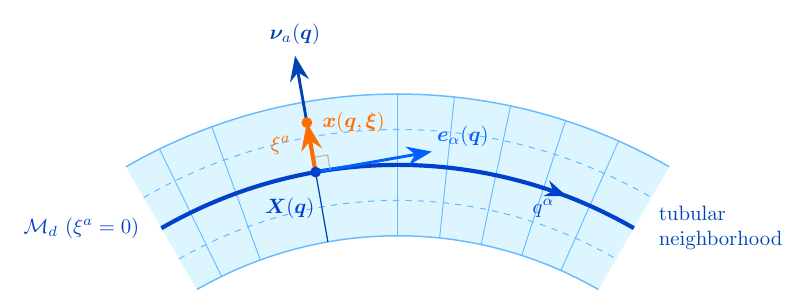}
    \caption{A schematic picture of tubular coordinates for the constraint manifold.}
    \label{fig:tubular}
\end{figure}

For the $M$-dimensional normal directions to the constraint manifold, we take the normal coordinate denoted by $\xi^a$ with the flat metric $\delta_{ab}$.
We write the coordinates together as $z^\mu =(q^\alpha, \xi^a)$ and its metric is assumed to be 
\begin{align}
    \dd s^2 = \delta_{ij} \dd x^i \dd x^j = h_{\mu\nu} \dd z^\mu \dd z^\nu\,,
    \label{eq:metric-tubular-saisyo}
\end{align}
which is called a tubular coordinate.
The schematic picture for the tabular coordinates is shown in Fig.~\ref{fig:tubular}.
The $N$-dimensional ambient coordinate is assumed to be related to $z^\mu$ by 
\begin{align}
    x^i = X^i(\bm{q}) + \xi^a \nu^i{}_a(\bm{q})\,,
    \label{eq:coordinate-decomposition-x-tubular}
\end{align}
where $X^i(\bm{q})$ is the embedding we have used on $\mc{M}_d$.
This kind of coordinate system is discussed in Ref.~\cite{2003-SchusterJaffe-AP}.
The frame fields for the tangent space and the normal space are
\begin{align}
    e^i{}_\alpha = \del_\alpha X^i\,,
    \qquad 
    \nu^i{}_a = \del_a x^i\,,
\end{align}
and satisfy
\begin{align}
    \delta_{ij} e^i{}_\alpha e^j{}_\beta = g_{\alpha\beta}\,,
    \qquad 
    \delta_{ij} \nu^i{}_a \nu^j{}_b = \delta_{ab}\,.
\end{align}
We also introduce the extrinsic curvature associated with $\nu^i{}_a$ by
\begin{align}
     K^a{}_{\alpha\beta} \coloneqq e^i{}_\alpha e^j{}_\beta \nabla_i \nu_j{}^a\,.
    \label{eq:extrinsic-curvature-Ka}
\end{align}
The sign convention for the extrinsic curvature follows Ref.~\cite{2004-Poisson}.

In this coordinate system, the Gauss formula becomes
\begin{align}
    \del_\alpha e^i{}_\beta = \Gamma^\gamma{}_{\alpha\beta} e^i{}_\gamma -  K^a{}_{\alpha\beta} \nu^i{}_a\,,
\end{align}
and the Weingarten equation~\cite{1992-doCarmo,1999-Spivak} is given by
\begin{align}
    \del_\alpha \nu^i{}_a =  K_{a \alpha}{}^\beta e^i{}_\beta - \mc{A}_{\alpha a}{}^b \nu^i{}_b\,.
\end{align}
$\Gamma^\gamma{}_{\alpha\beta}$ is the Levi-Civita connection calculated by the metric $g_{\alpha\beta}$:
\begin{align}
    \Gamma^\gamma_{\alpha\beta} = \frac{1}{2} g^{\gamma \delta} (\del_\alpha g_{\delta \beta} + \del_\beta g_{\alpha \delta} - \del_\delta g_{\alpha\beta})\,.
\end{align}
In addition, we also introduce the normal-bundle connection or the $\mr{SO}(M)$ connection of the normal bundle by 
\begin{align}
    \mc{A}_{\alpha ab} \coloneqq - \nu^i{}_b \del_\alpha \nu_{ia} = \nu^i{}_a \del_\alpha \nu_{ib} = - \mc{A}_{\alpha ba}\,.
    \label{eq:def-connection-A}
\end{align}
The analogy to the gauge field will be directly seen when we consider the derivative operator below.\footnote{
If we consider a local flat frame in the tangent direction with a frame field $g_{\alpha \beta} =\delta_{mn} \hat{e}^m{}_\alpha \hat{e}^n{}_\beta$,
we have a connection like \eqref{eq:def-connection-A} associated with the local $\mr{SO}(d) = \mr{SO}(N-M)$ rotation given by 
$\hat{\omega}_{\alpha mn} = - \hat{e}_{n \beta} \nabla_\alpha \hat{e}_m{}^\beta = - \hat{e}_{n\beta} (\del_\alpha \hat{e}_m{}^\beta + \Gamma^\beta{}_{\alpha \gamma} \hat{e}_m{}^\gamma)$.
This connection is called a spin connection in the high energy physics or gravitational theory context~\cite{2004-Ortin,2012-FreedmanVanProeyen}.
This connection appears when the covariant derivative acts on the object which has a nontrivial representation of the local tangent rotation $\mr{SO}(d)$.
In this paper, we consider the derivative operator acting mainly on a scalar, namely a wave function, so we do not find the direct appearance of the spin connection.
}

\begin{thm}
In the tubular coordinates $z^\mu=(q^\alpha,\xi^a)$,
the ambient metric and its inverse take the following block forms:
\begin{align}
    h_{\mu\nu} = \begin{pmatrix}
        \bar{g}_{\alpha\beta} + \mc{B}^c_{\alpha} \mc{B}_{\beta c} & \mc{B}_{\alpha b} \\
        \mc{B}_{\beta a} & \delta_{ab}
    \end{pmatrix}\,, 
    \qquad 
    h^{\mu\nu} = \begin{pmatrix}
        \bar{g}^{\alpha \beta} & - \bar{g}^{\alpha \gamma} \mc{B}^b_\gamma \\
        -\mc{B}^a_\gamma \bar{g}^{\gamma \beta} & \delta^{ab} + \mc{B}^a_\alpha \bar{g}^{\alpha\beta} \mc{B}^b_\beta
    \end{pmatrix}\,,
    \label{eq:metric-hmunu-decompose}
\end{align}
where we define 
\begin{align}
    \mc{B}^a_\alpha &\coloneqq \mc{A}_\alpha{}^a{}_b \xi^b\,,
    \\
    \bar{g}_{\alpha\beta} &\coloneqq g_{\gamma \delta}(\delta^\gamma_\alpha + \xi^a  K_{a\alpha}{}^\gamma)(\delta^\delta_\beta + \xi^b  K_{b \beta}{}^\delta)\,.
\end{align}
\end{thm}

\proof

From \eqref{eq:coordinate-decomposition-x-tubular}, $\dd x^i$ is related to $\dd q^\alpha$ and $\dd \xi^a$ as
\begin{align}
    \dd x^i &= (\delta^\beta{}_\alpha + \xi^a  K_{a\alpha}{}^\beta) e^i{}_\beta \dd q^\alpha + \nu^i{}_a (\dd \xi^a + \mc{A}_\alpha{}^a{}_b \xi^b \dd q^\alpha)
    \\
    &=(\delta^\beta{}_\alpha + \xi^a  K_{a\alpha}{}^\beta) e^i{}_\beta \dd q^\alpha + \nu^i{}_a (\dd \xi^a + \mc{B}^a_{\alpha} \dd q^\alpha)\,,
    \label{eq:dxiexpansion}
\end{align}
where the contracted connection is defined by 
\begin{align}
    \mc{B}^a_\alpha \coloneqq \mc{A}_\alpha{}^a{}_b \xi^b\,.
\end{align}
Substituting \eqref{eq:dxiexpansion} into \eqref{eq:metric-tubular-saisyo}, the metric $h_{\mu\nu}$ is read as 
\begin{align}
    \dd s^2&= \delta_{ij} \dd x^i \dd x^j
    \nn \\
    &= g_{\gamma \delta}(\delta^\gamma_\alpha + \xi^a  K_{a\alpha}{}^\gamma)(\delta^\delta_\beta + \xi^b  K_{b \beta}{}^\delta) \dd q^\alpha \dd q^\beta
     +\delta_{ab} (\dd \xi^a + \mc{B}^a_\alpha \dd q^\alpha) (\dd \xi^b + \mc{B}^b_\beta \dd q^\beta)
    \\
    &= \bar{g}_{\alpha\beta}(\bm{q},\bm{\xi}) \dd q^\alpha \dd q^\beta + \delta_{ab}(\dd \xi^a + \mc{B}^a_\alpha \dd q^\alpha)(\dd \xi^b + \mc{B}^b_\beta \dd q^\beta)\,,
\end{align}
and its matrix form is 
\begin{align}
    h_{\mu\nu} = \begin{pmatrix}
        \bar{g}_{\alpha\beta} + \mc{B}^c_{\alpha} \mc{B}_{\beta c} & \mc{B}_{\alpha b} \\
        \mc{B}_{\beta a} & \delta_{ab}
    \end{pmatrix}\,.
\end{align}
Here, we define $\bar{g}_{\alpha\beta}$ for $h_{\alpha\beta}$ by 
\begin{align}
    \bar{g}_{\alpha\beta} = g_{\gamma \delta}(\delta^\gamma{}_\alpha + \xi^a  K_{a\alpha}{}^\gamma)(\delta^\delta{}_\beta + \xi^b  K_{b \beta}{}^\delta)\,.
\end{align}
This matrix is decomposed as 
\begin{align}
    h_{\mu\nu} = \begin{pmatrix}
        \delta^\gamma{}_\alpha & \mc{B}^c_\alpha 
        \\
        0 & \delta^c{}_a
    \end{pmatrix} \begin{pmatrix}
        \bar{g}_{\gamma \delta} & 0 
        \\
        0 & \delta_{cd}
    \end{pmatrix}\begin{pmatrix}
        \delta^\delta_\beta & 0
        \\
        \mc{B}^d_\beta & \delta^d_b
    \end{pmatrix}\,,
    \label{eq:hmunu-decomposition}
\end{align}
Because the inverse of the upper/lower triangular matrix is given by
\begin{align}
    \begin{pmatrix}
        \delta^\gamma_\alpha & \mc{B}^c_\alpha 
        \\
        0 & \delta^c_a
    \end{pmatrix}^{-1} = \begin{pmatrix}
        \delta^\alpha_\gamma & - \mc{B}^b_\gamma
        \\
        0 & \delta^b_c
    \end{pmatrix}\,,
    \qquad 
    \begin{pmatrix}
        \delta^\delta_\beta & 0
        \\
        \mc{B}^d_\beta & \delta^d_b 
    \end{pmatrix}^{-1} = \begin{pmatrix}
        \delta^\alpha_\delta & 0 
        \\
        - \mc{B}^a_\delta & \delta^a_d
    \end{pmatrix}\,,
\end{align}
we find the inverse metric has the following form:
\begin{align}
    h^{\mu\nu} = \begin{pmatrix}
        \bar{g}^{\alpha \beta} & - \bar{g}^{\alpha \gamma} \mc{B}^b_\gamma \\
        -\mc{B}^a_\gamma \bar{g}^{\gamma \beta} & \delta^{ab} + \mc{B}^a_\alpha \bar{g}^{\alpha\beta} \mc{B}^b_\beta
    \end{pmatrix}\,.
\end{align}

\qed

From the decomposition \eqref{eq:hmunu-decomposition}, the determinant of $h_{\mu\nu}$ is directly calculated and the volume form becomes 
\begin{align}
    \sqrt{\det h} = \sqrt{\det \bar{g}} = \sqrt{\det g} \mc{J}\,,
    \qquad 
    \mc{J} = \det (\delta_\alpha^\beta + \xi^a K_{a\alpha}{}^\beta)\,.
    \label{eq:jacobian-turblar}
\end{align}

\paragraph{Coordinate system normal to $\mc{M}_d$.}

We give some comments on the relation between $\xi^a$ and $\eta^A$.
The label of the constraint $A,B,\ldots$ and the normal coordinate indices $a,b,\ldots$ are related by
\begin{align}
    b^A{}_a \coloneqq \del_a f^A\,.
\end{align}
The extrinsic curvatures defined in \eqref{eq:extrinsic-curvature-KA} and \eqref{eq:extrinsic-curvature-Ka} are related by
\begin{align}
     K^A{}_{\alpha\beta} = b^A{}_a  K^a{}_{\alpha\beta}\,.
\end{align}
If we write the corresponding mean curvature components by
\begin{align}
    K^a = \tr K^a{}_{\alpha\beta} = g^{\alpha\beta} K^a{}_{\alpha\beta}\,,
    \qquad 
    K^A = \tr K^A{}_{\alpha\beta} = g^{\alpha\beta} K^A{}_{\alpha\beta}\,,
\end{align}
the extrinsic curvature term in \eqref{eq:Pi^2-app} is also rewritten as 
\begin{align}
    \frac{1}{4} K^A G_{AB} K^B = \frac{1}{4} K^a \delta_{ab} K^b\,.
\end{align}

\paragraph{Relations among coordinates.}

\begin{figure}
    \centering
    \begin{tikzpicture}
        \node[draw, shape=circle] (A) at (0,0) {$x^i$};
        \node[draw, shape=circle] (B) at (-2.31,0) {$q^\alpha$};
        \node[draw, shape=circle] (C) at (1.15,2) {$\eta^A$};
        \node[draw, shape=circle] (D) at (1.15,-2) {$\xi^a$};
        \draw[-, thick] (A) -- (B);
        \draw[-, thick] (A) -- (C);
        \draw[-, thick] (A) -- (D);
        \draw[-, thick] (C) -- (D);
        \node at (-1.1,0) [above] {$e^i{}_\alpha / E^i{}_\alpha$};
        \node at (0.5,1.2) [left] {$F^i{}_A$};
        \node at (0.5,-1.2) [left] {$\nu^i{}_a$};
        \node at (1.15,0) [right] {$b^A{}_a$};
    \end{tikzpicture}
    \caption{Schematic figure illustrating the coordinates and their frame fields.}
    \label{fig:vielbeins}
\end{figure}
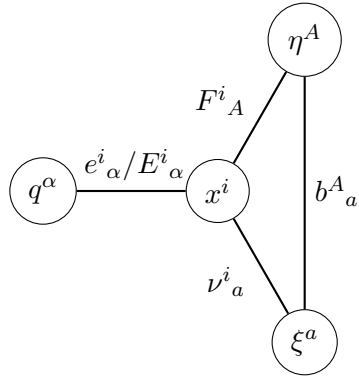

The relations among the coordinates shown in Tab.~\ref{tab:coordinates-tab} and their frame fields are schematically summarized in Fig.~\ref{fig:vielbeins}.

\subsubsection*{Comments on relations to theoretical physics}

Before closing this subsection, we give comments on the relations among our two formalisms and those used in theoretical physics.
In high energy physics and gravitational theory, the two types of the metric decomposition and their components are discussed.

\paragraph{Constraint-adapted coordinates.}

In the constraint-adapted coordinates, the metric is decomposed as \eqref{eq:mathcalG_IJ-mathcalG^IJ} or \eqref{eq:mathcalG_IJ-mathcalG^IJ-2}.
This type of decomposition is similar to the ADM decomposition in gravitational theory \cite{2008-ArnowittEtAl-GRG,1991-CorichiNunez-RMF,2004-Poisson}, which decomposition respect the foliation structure.
$G_{AB}$ is a normal metric or regarded as generalized lapse matrix and $\Theta^\alpha_A$ is the shift vector.
If we write $G_{11} = N^2$ when $M=1$, this $N$ corresponds to the lapse function in the ADM decomposition.

\paragraph{Tubular coordinates.}

In the tubular coordinates, the metric is decomposed as \eqref{eq:metric-hmunu-decompose}.
This kind of metric is used in the context of the dimensional reduction and the form like \eqref{eq:metric-hmunu-decompose} is often discussed in so-called Kaluza--Klein reduction~\cite{1921-Kaluza-SPAW,1926-Klein-ZP,1987-BailinLove-RPP,1986-DuffEtAl-PR,1995-Duff-OKC}.
In that context, the vector field appearing the off-diagonal component of the metric is often called graviphoton, which also behaves as a vector component in the reduced space.

\subsection{Geometric properties in constraint-adapted coordinates}
\label{app:geoprp-constraint-adapted}

\begin{thm}
For independent constraints $f^A(\bm{x})=0$, the constrained ambient measure coincides with the volume measure induced on $\mc{M}_d$:
\begin{align}
    \dd^N \bm{x} \prod_A \delta(f^A(\bm{x})) \frac{1}{\sqrt{\det G}} = \dd^d \bm{q} \sqrt{\det g(\bm{q})}\,.
\end{align}
\end{thm}
\proof

From \eqref{eq:decompose-detmathcalG}, the determinants satisfy the following relation on the constraint manifold $\mc{M}_d$:
\begin{align}
    \det \mc{G}|_{\eta^A=0} = \det g \times \det G\,,
\end{align}
because of \eqref{eq:relation-gamma-g}.
In addition, the corresponding integral measures satisfy
\begin{align}
    \dd^N\bm{x} =\dd^d\bm{q} \dd^M \bm{\eta} \sqrt{\det \mc{G}}\,,
\end{align}
which directly comes from \eqref{eq:M-dim-metric-mathcalGIJ}.
We then obtain
\begin{align}
    \dd^N \bm{x} \prod_A \delta (f^A) \frac{1}{\sqrt{\det G}} = \dd^d\bm{q} \dd^M \bm{\eta} \sqrt{\det g} \sqrt{\det G} \prod_A \delta(\eta^A) \frac{1}{\sqrt{\det G}}
    = \dd^d\bm{q} \sqrt{\det g}\,.
\end{align}
\qed

\begin{thm}
The integral measure
\begin{align}
    \dd^N\bm{x} \prod_A \delta(f^A) \frac{1}{\sqrt{\det G}}\,,
\end{align}
is invariant under the redefinition of the constraints.
\end{thm}

\proof

Let us consider the redefinition of the constraints which give the same constraint manifold as
\begin{align}
    f'^A(\bm{x}) = F^A(f(\bm{x}))\,,
\end{align}
and this redefinition is assumed to satisfy
\begin{align}
    F^A(\bm{0}) = 0\,,
    \qquad 
    \mf{J}^A{}_B = \left. \frac{\del F^A}{\del f^B} \right|_{f=0}\,,
    \qquad 
    \det \mf{J} \neq 0\,.
\end{align}
Under this redefinition, the delta function term changes as 
\begin{align}
    \prod_A \delta(f'^A) = \frac{1}{|\det \mf{J}|} \prod_A \delta (f^A)\,,
\end{align}
and $G^{AB}$ also transforms as 
\begin{align}
    G'^{AB} = \mf{J}^A{}_C \mf{J}^B{}_D G^{CD}\,,
    \qquad 
    G'^{AB} = \delta^{ij} \del_i f'^A \del_j f'^B\,,
\end{align}
then we obtain
\begin{align}
    \det G' = \frac{\det G}{(\det \mf{J})^2}\,.
\end{align}
Using these equations, we find that 
\begin{align}
    \prod_A \delta(f'^A) \frac{1}{\sqrt{\det G'}} = \prod_A \delta(f^A) \frac{1}{\sqrt{\det G}}\,.
\end{align}
\qed

\begin{thm}
The Laplace--Beltrami operator on $\mathcal{M}_d$ admits the following equivalent expressions in ambient and intrinsic coordinates, with the ambient expression evaluated on the constraint manifold:
\begin{align}
    \Delta_{\mc{M}_d} = \sqrt{\det G} \del_i \left.\left(\frac{M^{ij}}{\sqrt{\det G}} \del_j \right)\right|_{\mc{M}_d}
    = \frac{1}{\sqrt{\det g}} \del_\alpha \left( \sqrt{\det g} g^{\alpha \beta} \del_\beta \right)\,.
    \label{eq:Laplacian-rewriting}
\end{align}
\end{thm}

\proof

In the constraint-adapted coordinate $y^I$, the frame field $E^i{}_\alpha = \del_\alpha \mc{X}^i$ satisfies the orthogonal condition
\begin{align}
    n^A{}_i E^i{}_\alpha = \del_i f^A E^i{}_\alpha =\del_\alpha \eta^A = 0\,.
\end{align}
The projector is decomposed as
\begin{align}
    M^{ij} = \gamma^{\alpha\beta} E^i{}_\alpha E^j{}_\beta\,.
\end{align}

The derivative operator projected by $M^{ij}$ and $E^i{}_\alpha$ are written in terms of the frame field and the local coordinate as 
\begin{align}
    M^{ij} \del_j = E^i{}_\alpha \gamma^{\alpha\beta} \del_\beta\,,
    \qquad 
    E^i{}_\alpha \del_i = \del_\alpha\big|_\eta\,.
\end{align}
Now we evaluate
\begin{align}
    \sqrt{\det G} \del_i \left( \frac{1}{\sqrt{\det G}} E^i{}_\alpha \right)
    &= \del_i E^i{}_\alpha - \frac{1}{\sqrt{\det G}} \del_\alpha \sqrt{\det G}
    \nn \\
    &= \del_\alpha \log \sqrt{\det \mc{G}} - \del_\alpha \log \sqrt{\det G}
    \nn \\
    &= \del_\alpha \log \sqrt{\det \gamma}\,,
\end{align}
where we use \eqref{eq:decompose-detmathcalG} and
\begin{align}
    \del_i E^i{}_\alpha = \frac{\del y^I}{\del x^i} \del_I \del_\alpha x^i = \frac{\del y^I}{\del x^i} \del_\alpha \frac{\del x^i}{\del y^I}\,.
\end{align}
Using these equations, we find that
\begin{align}
    \sqrt{\det G}\del_i \left( \frac{M^{ij}}{\sqrt{\det G}} \del_j \right)
    &=\sqrt{\det G} \del_i \left(\frac{1}{\sqrt{\det G}} E^i{}_\alpha \gamma^{\alpha\beta} \del_\beta \right)
    \nn \\
    &=\gamma^{\alpha\beta} \del_\alpha \del_\beta + (\del_\alpha \gamma^{\alpha\beta}) \del_\beta + (\del_\alpha \log \sqrt{\det \gamma}) \gamma^{\alpha \beta} \del_\beta
    \nn \\
    &=\frac{1}{\sqrt{\det \gamma}} \del_\alpha\left( \sqrt{\det \gamma} \gamma^{\alpha\beta} \del_\beta \right)\,,
\end{align}
and this equation reduces on the constraint manifold to 
\begin{align}
    \sqrt{\det G} \del_i \left(\frac{M^{ij}}{\sqrt{\det G}} \del_j \right)\Bigg|_{\mc{M}_d}
    = \frac{1}{\sqrt{\det g}} \del_\alpha \left( \sqrt{\det g} g^{\alpha \beta} \del_\beta \right)\,.
\end{align}
\qed

In this proof, we show \eqref{eq:Laplacian-rewriting} by the component calculation.
On the other hand, \eqref{eq:Laplacian-rewriting} is also shown by considering the partial integration.

\proof

\begin{align}
    &\int \dd^N\bm{x} \prod_A \delta(f^A)\frac{1}{\sqrt{\det G}} \psi_1^\ast \left( \sqrt{\det G} \del_i \left( \frac{M^{ij}}{\sqrt{\det G}} \del_j \right) \psi_2 \right)
    \nn \\
    &=- \int \dd^N \bm{x} \prod_A \delta(f^A) \frac{1}{\sqrt{\det G}} (\del_i \psi^\ast_1) M^{ij} \del_j \psi_2
    \nn \\
    &= -\int \dd^d \bm{q} \sqrt{\det g} g^{\alpha \beta} (\del_\alpha \psi_1^\ast) \del_\beta \psi_2
    \\
    &= \int \dd^d\bm{q} \sqrt{\det g} \psi_1^\ast \left(\frac{1}{\sqrt{\det g}} \del_\alpha \sqrt{\det g} g^{\alpha\beta} \del_\beta \psi_2 \right)
    = \int \dd^d\bm{q} \sqrt{\det g} \psi_1^\ast \Delta_{\mc{M}_d} \psi_2\,,
\end{align}
where we use 
\begin{align}
    \left( \del_i \prod_A \delta(f^A) \right) M^{ij} = 0\,,
\end{align}
coming from $M^{ij} \del_i f^A = 0$.
We then obtain \eqref{eq:Laplacian-rewriting}:
\begin{align}
    \Delta_{\mc{M}_d} = \sqrt{\det G} \del_i \left.\left(\frac{M^{ij}}{\sqrt{\det G}} \del_j \right)\right|_{\mc{M}_d}
    = \frac{1}{\sqrt{\det g}} \del_\alpha \left( \sqrt{\det g} g^{\alpha \beta} \del_\beta \right)\,.
\end{align}

\qed

\begin{thm}
The squared tangential momentum operator $\hat{P}^2=\delta^{ij}\hat{P}_i\hat{P}_j$ equals the negative Laplace--Beltrami operator on $\mc{M}_d$:
\begin{align}
    \hat{P}^2 = - \Delta_{\mc{M}_d}
\end{align}
\end{thm}

\proof

Because the tangential momentum operator is written as $\hat{P}_i =  - \iu M^j{}_i \del_j = -\iu e_i{}^\alpha \del_\alpha$, the derivative operator $\hat{P}^2$ is expressed as 
\begin{align}
    \hat{P}^2 = - \delta^{ij}M^k{}_i \del_k (M^l{}_j \del_l) = - \delta^{ij} e_i{}^\alpha \del_\alpha(e_j{}^\beta \del_\beta) 
    = - g^{\alpha\beta} \del_\alpha \del_\beta - \delta^{ij}(e_i{}^\alpha \del_\alpha e_j{}^\beta) \del_\beta\,.
\end{align}
From the Gauss formula~\cite{2004-Poisson,2019-DajczerTojeiro},\footnote{
While we use the Gauss formula on $\mc{M}_d$ in the proof, we note that a similar relation holds in the neighborhood coordinates, $\del_\alpha E^i{}_\beta = \del_\alpha \del_\beta \mc{X}^i = \Gamma^\gamma_{\alpha\beta}[\gamma] E^i{}_\gamma -  K^A{}_{\alpha\beta} G_{AB} n^{Bi}$,
where $\Gamma^\gamma_{\alpha\beta}[\gamma]$ is the Levi-Civita connection calculated by the metric $\gamma_{\alpha\beta}$.
}
the derivative of the frame field is given by 
\begin{align}
    \del_\alpha e^i{}_\beta = \del_\alpha \del_\beta X^i = \Gamma^\gamma_{\alpha\beta} e^i{}_\gamma -  K^A{}_{\alpha\beta} n^i{}_A\,,
\end{align}
where we introduce the second fundamental form or extrinsic curvature associated with $n^A{}_i$ by
\begin{align}
     K^A{}_{\alpha\beta} = e^i{}_\alpha e^j{}_\beta \nabla_i n^A{}_j
    =e^i{}_\alpha e^j{}_\beta \del_i n^A{}_j\,.
    \label{eq:extrinsic-curvature-KA}
\end{align}
Using this, we find that
\begin{align}
    \hat{P}^2 = - g^{\alpha\beta} \del_\alpha \del_\beta + g^{\alpha\beta} \Gamma^\gamma_{\alpha\beta} \del_\gamma
    =- g^{\alpha\beta}\nabla_\alpha \del_\beta
    = - \frac{1}{\sqrt{\det g}} \del_\alpha \left( \sqrt{\det g} g^{\alpha\beta} \del_\beta \right)\,.
\end{align}

\qed

\begin{thm}
With respect to the induced inner product on $\mc{M}_d$ in \eqref{eq:inner-product-ambient}, the formal adjoint of the tangential momentum operator is given below:
\begin{align}
    \hat{P}^\dagger_i = \hat{P}_i - \iu C_i\,,
    \label{eq:Pdagger-app}
\end{align}
where
\begin{align}
    C_i \coloneqq \nabla_\alpha e_i{}^\alpha =\Delta_{\mathcal{M}_d}X_i\,.
\end{align}
\end{thm}

\proof

Let us evaluate the expectation value of $\hat{P}^\dagger_i$ by using $\braket{\psi_1|\hat{P}^\dagger_i|\psi_2} = \braket{\hat{P}_i \psi_1|\psi_2}$ as:
\begin{align}
    \braket{\hat{P}_i\psi_1|\psi_2} &= \iu \int \dd^d\bm{q} \sqrt{\det g} e_i{}^\alpha (\del_\alpha \psi_1^\ast) \psi_2
    \nn \\
    &= -\iu \int \dd^d\bm{q} \sqrt{\det g}\psi_1^\ast\bigl[ e_i{}^\alpha \del_\alpha + \nabla_\alpha e_i{}^\alpha\bigr] \psi_2
    \\
    &=\braket{\psi_1| \hat{P}_i - \iu C_i| \psi_2}\,.
\end{align}
Here, the additional vector contribution is evaluated as
\begin{align}
    C_i = \nabla_\alpha e_i{}^\alpha = \Delta_{\mc{M}_d} X_i = - K^A G_{AB} n^B{}_i\,,
    \label{eq:Ci_in_P^dagger}
\end{align}
where we use that the Laplacian of the embedding coordinates is the normal vector and written by using the extrinsic curvature as
\begin{align}
    \Delta_{\mc{M}_d} X^i = - K^A n^i{}_A\,.
\end{align}
Here, we write the trace of the extrinsic curvature or mean curvature components as 
\begin{align}
    K^A = g^{\alpha\beta} K^A{}_{\alpha \beta}\,.
\end{align}

\qed

Since the additional term in $P^\dagger_i$ is in the normal direction from \eqref{eq:Ci_in_P^dagger}, $C_i$ is orthogonal to the tangent momentum $C_i P^i =0$ and we find that 
\begin{align}
    P^\dagger_i P^i = P_i P^i = - \Delta_{\mc{M}_d}\,.
\end{align}
However, we note that this tangent momentum operator is not Hermitian as we show above.

\begin{thm}
The geometric momentum operator defined by
\begin{align}
    \hat{\Pi}_i \coloneqq \hat{P}_i - \frac{\iu}{2} C_i\,,
\end{align}
is formally Hermitian with respect to the induced inner product on $\mathcal{M}_d$.
Its square differs from the negative Laplace--Beltrami operator by an extrinsic-curvature term:
\begin{align}
    \hat{\Pi}^2 = \hat{P}^2 + \frac{1}{4} C_i C^i = -\Delta_{\mc{M}_d} + \frac{1}{4} K^A G_{AB} K^B\,.
    \label{eq:Pi^2-app}
\end{align}
\end{thm}

\proof

Taking the Hermitian conjugate of $\hat{\Pi}_i$ with \eqref{eq:Pdagger-app},
\begin{align}
    \hat{\Pi}^\dagger_i = \hat{P}_i^\dagger + \frac{\iu}{2} C_i = \hat{P}_i - \frac{\iu}{2} C_i = \hat{\Pi}_i,
\end{align}
the Hermiticity of $\hat{\Pi}_i$ is directly checked.
Using the projector, we write this geometric momentum operator as 
\begin{align}
    \hat{\Pi}_i = - \iu \left(M^j{}_i \del_j + \frac{1}{2} C_i \right)\,,
\end{align}
and its squared becomes 
\begin{align}
    \hat{\Pi}_i \hat{\Pi}^i &=- M^j{}_i \del_j (M^{ik}\del_k) - \frac{1}{2}M^j{}_i \del_j C^i - \frac{1}{2} C^i M^j{}_i \del_j - \frac{1}{4} C_i C^i
    \nn \\
    &=- M^j{}_i \del_j (M^{ik}\del_k) - \frac{1}{2}(M^j{}_i \del_j C^i) -  C^i M^j{}_i \del_j - \frac{1}{4} C_i C^i
    \nn \\
    &= - M^j{}_i \del_j (M^{ik}\del_k) + \frac{1}{4} C_i C^i\,,
\end{align}
where we use the conditions coming from $C_i$ being normal $C^i M^j{}_i \del_j = 0$ and 
\begin{align}
    M^j{}_i \del_j C^i = - C_i C^i\,.
\end{align}
The second condition is derived from the orthogonal condition $C_i e^{i\alpha} =0$ coming from $C_i = \Delta_{\mc{M}_d} X_i = - K^A n_A{}^i$ as
\begin{align}
    0 = \nabla_\alpha(C_i e^{i\alpha}) = (\del_\alpha C_i) e^{i\alpha} + C_i \nabla_\alpha e^{i\alpha}
    =(M^j{}_i \del_j C^i) + C_i C^i\,.
    \label{eq:nablaCe}
\end{align}
It is then found that
\begin{align}
    \hat{\Pi}_i \hat{\Pi}^i = \hat{P}^2 + \frac{1}{4} C_i C^i
    = -\Delta_{\mc{M}_d} + \frac{1}{4} K^A G_{AB} K^B \,.
\end{align}

\qed

\subsection{Geometric properties in tubular coordinates}
\label{app:geoprp-tubular}

\begin{thm}
The ambient Laplacian decomposes into tangential and normal parts in tubular coordinates as 
\begin{align}
    \Delta &= \Delta_{\mr{tan}} + \Delta_{\mr{nor}}\,,
    \\
    \Delta_{\mr{tan}} &= \frac{1}{\sqrt{\det g} \mc{J}} D_\alpha \left( \sqrt{\det g} \mc{J} \bar{g}^{\alpha \beta} D_\beta \right)\,,
    \label{eq:Delta-tan}
    \\
    \Delta_{\mr{nor}}&= \frac{1}{\mc{J}} \del_a(\mc{J} \del_a)\,.
    \label{eq:Delta-nor}
\end{align}
The tangential part involves the normal-bundle covariant derivative given by 
\begin{align}
    D_\alpha = \del_\alpha - \mc{B}^a_\alpha \del_a = \del_\alpha - \mc{A}_{\alpha}{}^a{}_b \xi^b \del_a\,.
\end{align}
\end{thm}

\proof 

In the $z^\mu$ coordinates, the Laplace operator is expressed as
\begin{align}
    \Delta = \frac{1}{\sqrt{\det h}} \del_\mu \left( \sqrt{\det h} h^{\mu\nu} \del_\nu \right)\,.
\end{align}
If we expand this as
\begin{align}
    \Delta = \frac{1}{\sqrt{\det h}} \left[
        \del_\alpha \left( \sqrt{\det h} h^{\alpha \beta} \del_\beta + \sqrt{\det h} h^{\alpha b} \del_b \right)
        + \del_a \left( \sqrt{\det h} h^{a \beta} \del_\beta + \sqrt{\det h} h^{ab} \del_b \right)
    \right]\,,
\end{align}
the first bracket terms are summarized as 
\begin{align}
    h^{\alpha \beta} \del_\beta + h^{\alpha b} \del_b =\bar{g}^{\alpha \beta} (\del_\beta - \mc{B}^b_\beta \del_b)
    \eqqcolon \bar{g}^{\alpha\beta} D_\beta\,,
\end{align}
which is regarded as a covariant derivative with the (non-abelian) gauge field $\mc{B}^a_\alpha$.
Since the terms in the second bracket become
\begin{align}
    h^{a\beta} \del_\beta + h^{ab} \del_b
    &= - \mc{B}^a{}_\alpha \bar{g}^{\alpha\beta} \del_\beta + (\delta^{ab} + \mc{B}^a_\alpha \bar{g}^{\alpha\beta} \mc{B}^b_\beta) \del_b
    \nn \\
    &=\delta^{ab} \del_b - \mc{B}^a_\alpha \bar{g}^{\alpha\beta}(\del_\beta - \mc{B}^b_\beta \del_b)
    =\delta^{ab} \del_b - \mc{B}^a_\alpha \bar{g}^{\alpha\beta} D_\beta\,,
\end{align}
the Laplace operator is
\begin{align}
    \Delta = \frac{1}{\sqrt{\det h}} \left[
        \del_\alpha (\sqrt{\det h} \bar{g}^{\alpha\beta} D_\beta) + \del_a(\sqrt{\det h} \del_a) - \del_a(\sqrt{\det h} \mc{B}^a{}_\alpha \bar{g}^{\alpha\beta} D_\beta)
    \right]\,.
    \label{eq:Laplacian-totyuu}
\end{align}
The last term in \eqref{eq:Laplacian-totyuu} is simplified as
\begin{align}
    \del_a \left( \sqrt{\det h} \mc{B}^a_\alpha \bar{g}^{\alpha\beta} D_\beta \right)
    &= \mc{B}^a_\alpha \del_a \left(\sqrt{\det h} \bar{g}^{\alpha\beta} D_\beta \right) + (\del_a \mc{B}^a_\alpha) \sqrt{\det h} \bar{g}^{\alpha\beta} D_\beta
    \nn \\
    &=\mc{B}^a_\alpha \del_a \left(\sqrt{\det h} \bar{g}^{\alpha\beta} D_\beta \right)\,,
\end{align}
due to the condition of $\del_a \mc{B}^a_\alpha = \mc{A}_{\alpha ab} \delta^{ab} =0$.
It is then found that 
\begin{align}
    \Delta &= \frac{1}{\sqrt{\det h}} \left[
        (\del_\alpha - \mc{B}^a_\alpha \del_a) \left( \sqrt{\det h} \bar{g}^{\alpha\beta} D_\beta \right) + \del_a \left( \sqrt{\det h} \del_a \right)
    \right]
    \\
    &= \frac{1}{\sqrt{\det g}\mc{J}} D_\alpha \left( \sqrt{\det g} \mc{J} \bar{g}^{\alpha\beta} D_\beta \right) + \frac{1}{\mc{J}} \del_a (\mc{J} \del_a)\,,
\end{align}
where \eqref{eq:jacobian-turblar} is used.

\qed

The covariant derivative introduced in the above proof is also written as
\begin{align}
    D_\alpha = \del_\alpha -\mc{B}^a_\alpha \del_a
    = \del_\alpha - \mc{A}_{\alpha}{}^a{}_b \xi^b \del_a
    = \del_\alpha + \frac{1}{2} \mc{A}_{\alpha}{}^{ab} \ms{L}_{ab}\,,
    \qquad 
    \ms{L}_{ab} = \xi_a \del_b - \xi_b \del_a\,,
\end{align}
and $\ms{L}_{ab}$ is regarded as the generator of $\mr{SO}(M)$ or the $M$-dimensional angular momentum operator.
This generator satisfies the following algebra:
\begin{align}
    [\ms{L}_{ab}, \ms{L}_{cd}] = - (\delta_{ac} \ms{L}_{bd} - \delta_{bc} \ms{L}_{ad} - \delta_{ad} \ms{L}_{bc} + \delta_{bd} \ms{L}_{ac})\,.
\end{align}
On the constraint manifold, the tangent Laplacian reduces to the covariant Laplace operator with the gauge field:
\begin{align}
    \Delta_{\mc{M}_d}^{(\mc{A})} \coloneqq \Delta_{\mr{tan}} \big|_{\xi=0} = \frac{1}{\sqrt{\det g}} D_\alpha \sqrt{\det g} g^{\alpha \beta} D_\beta\,.
\end{align}
Here $\xi=0$ is taken in the coefficients $\bar{g}^{\alpha\beta}$ and $\mc{J}$, whereas the operator $\xi^b \del_a$ in $D_\alpha$ is kept: it acts on the normal oscillator states as a rotation generator and is of order one.

\begin{thm}
Rescaling the wave function by $\mc{J}^{1/2}$ removes the first-order normal derivatives from $\Delta_{\mathrm{nor}}$,
leaving the flat normal Laplacian and a scalar term $W$:
\begin{align}
    \mc{J}^{1/2} \Delta_{\mr{nor}} \mc{J}^{-1/2} = \del_a \del_a + W\,,
    \qquad 
    W = \frac{1}{4}\left( \frac{\del_a \mc{J}}{\mc{J}} \right)^2  - \frac{1}{2} \frac{\del_a \del_a \mc{J}}{\mc{J}}\,.
\end{align}
\end{thm}

\proof

The similarity transformation of $\Delta_{\mr{nor}}$ in \eqref{eq:Delta-nor} is 
\begin{align}
    \mc{J}^{1/2} \Delta_{\mr{nor}} \mc{J}^{-1/2}
    = \mc{J}^{1/2} \bigl[ \del_a \del_a + (\del_ a \log \mc{J}) \del_a \bigr] \mc{J}^{-1/2}\,,
\end{align}
and we find
\begin{align}
    \del_a \mc{J}^{-1/2} = \mc{J}^{-1/2} \left[ \del_a - \frac{1}{2} (\del_a \log \mc{J}) \right]\,.
\end{align}
The similarity transformation then becomes
\begin{align}
    \mc{J}^{1/2} \Delta_{\mr{nor}} \mc{J}^{-1/2}
    &=\mc{J}^{1/2}[\del_a + (\del_a \log \mc{J})] \mc{J}^{-1/2} \left[ \del_a - \frac{1}{2} (\del_a \log \mc{J}) \right]
    \nn \\
    &= \del_a \del_a + W\,,
\end{align}
with
\begin{align}
    W = \frac{1}{4} \left(\frac{\del_a \mc{J}}{ \mc{J}} \right)^2 - \frac{1}{2} \frac{\del_a \del_a \mc{J}}{\mc{J}}\,.
\end{align}

\qed

\begin{thm}
On the constraint manifold, the scalar term $W$ is determined by the extrinsic curvature and yields the geometric potential $V_{\geo}$:
\begin{align}
    W\big|_{\xi=0} = - \frac{1}{4} \delta^{ab} \bigl[ K_a K_b - 2 K_{a\alpha\beta} K_b{}^{\alpha\beta} \bigr] = -2 m V_\geo\,.
\end{align}
\end{thm}

\proof

The expansion of $\mc{J}$ around $\xi =0$ is 
\begin{align}
    \mc{J} = 1 + \xi^a K_a + \frac{1}{2} \xi^a \xi^b \bigl( K_a K_b - K_{a\alpha\beta}  K_b{}^{\alpha\beta}) \bigr) + \mc{O}(\xi^3)\,,
\end{align}
and we then obtain
\begin{align}
    W\big|_{\xi=0} = \frac{1}{4} \delta^{ab} K_a K_b - \frac{1}{2} \delta^{ab} \bigl[ K_a K_b - K_{a\alpha\beta} K_b{}^{\alpha\beta} \bigr]
    = - \frac{1}{4} \delta^{ab} \bigl[
        K_a K_b - 2 K_{a\alpha\beta} K_b{}^{\alpha\beta}
    \bigr]\,.
\end{align}

\qed

\begin{thm}
At leading order in the tubular expansion, the rescaled tangential Laplacian reduces to the Laplace--Beltrami operator coupled to the normal-bundle connection:
\begin{align}
    \mc{J}^{1/2} \Delta_{\mr{tan}} \mc{J}^{-1/2} = \Delta_{\mc{M}_d}^{(\mc{A})} + \mc{O}(\xi)
    = \frac{1}{\sqrt{\det g}} D_\alpha(\sqrt{\det g} g^{\alpha \beta} D_\beta) + \mc{O}(\xi)\,.
\end{align}
\end{thm}

\proof

Since the covariant derivative acts on the Jacobian as 
\begin{align}
    D_\alpha \mc{J}^{-1/2} = \mc{J}^{-1/2} \left( D_\alpha - \frac{1}{2} (D_\alpha \log \mc{J}) \right)\,,
\end{align}
the similarity transformation of $\Delta_{\mr{tan}}$ in \eqref{eq:Delta-tan} becomes 
\begin{align}
    \mc{J}^{1/2} \Delta_{\mr{tan}} \mc{J}^{-1/2}
    = \frac{1}{\sqrt{\det g}} D_\alpha \sqrt{\det g}\bar{g}^{\alpha \beta} D_\beta + W'\,,
\end{align}
with
\begin{align}
    W' = - \frac{1}{2 \sqrt{\det g}} D_\alpha(\sqrt{\det g} \bar{g}^{\alpha \beta }D_\beta \log\mc{J}) - \frac{1}{4}\bar{g}^{\alpha \beta} (D_\alpha \log \mc{J}) (D_\beta \log \mc{J})\,.
\end{align}
The order of $\xi^a$ is
\begin{align}
    \bar{g}^{\alpha\beta} = g^{\alpha\beta} + \mc{O}(\xi)\,,
    \qquad 
    \log \mc{J} = \xi^a \tr K_a + \mc{O}(\xi^2)\,,
    \qquad 
    D_\alpha \log \mc{J} = \mc{O}(\xi)\,,
\end{align}
we thus find 
\begin{align}
    W' = \mc{O}(\xi)\,.
\end{align}
As a result, we have the following equation:
\begin{align}
    \mc{J}^{1/2} \Delta_{\mr{tan}} \mc{J}^{-1/2} = \frac{1}{\sqrt{\det g}}D_\alpha (\sqrt{\det g} g^{\alpha \beta} D_\beta) + \mc{O}(\xi)\,.
\end{align}

\qed

From the results, the similarity transformation of the ambient Laplacian becomes
\begin{align}
    \mc{J}^{1/2} \Delta \mc{J}^{-1/2} = \Delta^{(\mc{A})}_{\mc{M}_d} + \del_a \del_a - 2 m V_\geo + \mc{O}(\xi)\,,
\end{align}
where the geometric potential is defined by 
\begin{align}
    V_\geo = \frac{1}{8m} \delta^{ab} \bigl[ K_a K_b - 2 K_{a\alpha\beta} K_b{}^{\alpha\beta} \bigr]\,.
\end{align}

\section{Effective Hamiltonian for QHD+Penalty}
\label{app:effective}

We show the derivation of the instantaneous oscillator Hamiltonian in \eqref{eq:inst-osc-ham}.
The Hamiltonian in the normal direction is
\begin{align}
    \hat{H}_\xi(t) = - \frac{1}{2m a(t)} \Delta_\xi + a(t) \kappa \xi^2\,.
\end{align}
If we write the derivative operator as $\hat{p}_a = - \iu \del_a$, the Hamiltonian describes the harmonic oscillator with a time dependent mass
\begin{align}
    \hat{H}_\xi = \sum_a \biggl[
        \frac{1}{2ma(t)} \hat{p}_a^2 + \frac{1}{2} ma(t) \Omega_\perp^2 (\hat{\xi}^a)^2
    \biggr]\,,
\end{align}
where the frequency is 
\begin{align}
    \Omega^2_\perp = \frac{2\kappa}{m}\,.
\end{align}
If we introduce the creation/annihilation operator
\begin{align}
    \hat{\ms{b}}_a(t) = \sqrt{\frac{ma \Omega_\perp}{2}} \hat{\xi}^a + \frac{\iu}{\sqrt{2 m a \Omega_\perp}} \hat{p}_a\,,
    \qquad 
    \hat{\ms{b}}_a^\dagger(t) = \sqrt{\frac{ma \Omega_\perp}{2}} \hat{\xi}^a - \frac{\iu}{\sqrt{2 m a \Omega_\perp}} \hat{p}_a\,,
\end{align}
they satisfy
\begin{align}
    [\hat{\ms{b}}_a(t), \hat{\ms{b}}^\dagger_b] = \delta_{ab}\,,
\end{align}
and the Hamiltonian becomes
\begin{align}
    \hat{H}_\xi = \Omega_\perp \left( \hat{\ms{N}}_\perp(t) + \frac{M}{2} \right)\,,
    \qquad 
    \hat{\ms{N}}_\perp(t) = \sum_a \hat{\ms{b}}^\dagger_a(t) \hat{\ms{b}}_a(t)\,.
\end{align}

In this system, the energy eigenvalues are time independent, but the eigenstates depend on time through the time-dependent mass.
This dependence is generated by a so-called dilatation operator
\begin{align}
    \hat{\ms{D}}_\perp = \frac{1}{2} (\hat{\xi}^a \hat{p}_a + \hat{p}_a \hat{\xi}^a)
    = -\iu \left( \xi^a \del_a + \frac{M}{2} \right)\,.
\end{align}
If we introduce the reference time $t_0$
\begin{align}
    a_0\coloneqq a(t_0)\,,
    \qquad 
    s(t) \coloneqq \sqrt{\frac{a(t)}{a_0}}\,,
\end{align}
and the unitary operator
\begin{align}
    \ms{U}_s(t) = \exp \left[ \iu \log s(t) \hat{\ms{D}}_\perp\right]\,,
\end{align}
we map the time dependent Hamiltonian to that at the reference frame by the unitary transformation and remove the time dependence:
\begin{align}
    \hat{\ms{U}}^\dagger_s \hat{H}_\xi(t) \hat{\ms{U}}_s
    = \sum_a \left[\frac{s^2}{2ma(t)} \hat{p}_a^2 + \frac{a(t) \kappa}{s^2} (\hat{\xi}^a)^2\right]
    = \sum_a \left[ \frac{1}{2ma_0} \hat{p}_a^2 + a_0 \kappa (\hat{\xi}^a)^2 \right] = \hat{H}_\xi(t_0)\,.
\end{align}
However, since $\hat{\ms{U}}_s$ depends on time, the Hamiltonian in the transformed frame acquires an additional term,
\begin{align}
    \hat{H}_\xi \mapsto \hat{\ms{U}}_s^\dagger \hat{H}_\xi \hat{\ms{U}}_s - \iu \hat{\ms{U}}_s^\dagger \frac{\del \hat{\ms{U}}_s}{\del t} \eqqcolon \hat{H}_{\perp,\mr{inst.}}(t),
\end{align}
where the additional term is
\begin{align}
    -\iu \hat{\ms{U}}_s^\dagger \frac{\del \hat{\ms{U}}_s}{\del t} = \frac{\dot{a}}{2a} \hat{\ms{D}}_\perp\,.
\end{align}
We then find that the instantaneous oscillator Hamiltonian becomes
\begin{align}
    \hat{H}_{\perp,\mr{inst.}}(t) = \hat{H}_\xi(t_0) + \frac{\dot{a}(t)}{2a(t)} \hat{\ms{D}}_\perp
    = \Omega_\perp \left( \hat{\ms{N}}_\perp + \frac{M}{2} \right) + \frac{\iu \dot{a}}{4a} \sum_a \left[ (\hat{\ms{b}}^\dagger_a)^2 - \hat{\ms{b}}_a^2 \right]\,.
\end{align}

\section{Details of the numerical calculations}
\label{app:numerics}

This appendix summarizes the discretization, the initial state, the indicators, the accuracy checks, and the parameters (Tab.~\ref{tab:num-params}) of the numerical calculations in Sec.~\ref{sec:numerics}.

\paragraph{Spatial discretization.}
To prevent differences in the discretization from being mistaken for differences in performance, QHD+Penalty was solved on the same tangential grid as QCHD.
For $N=2$ (Secs.~\ref{sec:num-sphere}--\ref{sec:num-ellipse}) we used polar coordinates $(r,\theta)$ with a uniform periodic grid in $\theta$ ($N_\theta$ points, FFT) and a uniform grid on $[0,r_{\max}]$ with a Dirichlet boundary at $r_{\max}$ ($N_r$ points).
With $\psi=\varphi/\sqrt{r}$, the kinetic operator of the angular mode $n$ is 
\begin{align}
    \frac{1}{2m}\left[-\del_r^2 + \frac{n^2-\tfrac14}{r^2}\right]\,,
\end{align}
which we diagonalized once in the sine basis, so that the kinetic propagation $e^{-\iu K_{\rm kin}\Delta t/a(t)}$ is exact and unitary.
The $n=0$ mode was expanded in a Fourier--Bessel basis regular at the origin, because an inner Dirichlet node at a small $r_{\min}$ acts on this mode as a hard disk and distorts the wave passing through the origin.
Since the wave function of QHD+Penalty bursts out to $r\simeq4$ during the initial transient, the interval $t\le3$ was computed on the enlarged domain $r\le5.6$ with the same grid spacing, after which the wave function was restricted to $r\le2.8$ (the mass outside $r=2.8$ at $t=3$ is at most $3\times10^{-3}$, for $\kappa=1$).
QCHD is a one-dimensional problem on the same $\theta$ grid, with 
\begin{align}
    \Delta_{\mc{M}_d}=R^{-2}\partial_\theta^2\,,
\end{align}
diagonalized by the FFT on the circle and 
\begin{align}
    \Delta_{\mc{M}_d}=h^{-1}\partial_\theta(h^{-1}\partial_\theta)\,,
    \qquad 
    h = \frac{\dd s}{\dd \theta}\,,
\end{align}
diagonalized once as an $N_\theta\times N_\theta$ matrix on the ellipse.
For $N=3$ (Sec.~\ref{sec:num-twocons}) we used spherical coordinates $(r,\vartheta,\varphi)$ with the polar axis along $(1,1,1)/\sqrt3$, a uniform grid in $\varphi$ ($N_\varphi$ points, FFT), Gauss--Legendre nodes in $\vartheta$ ($N_\vartheta$ points, Legendre expansion with $\ell\le N_\vartheta-1$), and a Dirichlet interval in $r$ with $u=r\psi$, where the radial operator 
\begin{align}
    \frac{1}{2m}\left[-\del_r^2 + \frac{\ell(\ell+1)}{r^2} \right]\,,
\end{align}
of each $\ell$ was diagonalized once.
QCHD uses the same $\varphi$ grid on the circle.
In all cases the time evolution is the Strang splitting 
\begin{align}
    \ee^{-\iu aW\Delta t/2}\, \ee^{-\iu K_{\rm kin}\Delta t/a}\, \ee^{-\iu aW\Delta t/2}
\end{align}
with $W=V|_{\mc{M}_d}$ for QCHD and $W=V+\kappa\mathcal{P}$ for QHD+Penalty, with $a$ evaluated at the midpoint of the time step.
The norm is conserved to round-off ($10^{-10}$), and no mask operation that rewrites the wave function was used.

\paragraph{Initial state.}

The initial state \eqref{eq:initial-tube} is uniform along $\mc{M}_d$ with respect to the arc length and Gaussian with width $\sigma_r=0.05$ in the normal directions, so that the initial constraint violation is $\langle d^2\rangle_0=\sigma_r^2/2=1.25\times10^{-3}$ for $M=1$ and $\sigma_r^2=2.5\times10^{-3}$ for $M=2$.
For QCHD the initial state is the uniform wave function on $\mc{M}_d$, and the two-dimensional pictures of QCHD in Figs.~\ref{fig:num-sphere-snap}--\ref{fig:num-ellipse-snap} are obtained by multiplying it by the frozen normal profile.

\paragraph{Indicators.}

The projection $\mathsf{P}$ in Eq.~\eqref{eq:def-Eproj} is the radial projection $r\to R$ on the circle and $r\to r_{\mc{M}_d}(\theta)$ on the ellipse, and the projection onto the nearest point of the circle in Sec.~\ref{sec:num-twocons}.
The values of $E_{\rm proj}$ and $\langle d^2\rangle$ at $t=t_{\rm end}$ depend on the grid, because at late times the tangential width of the distribution becomes smaller than the grid spacing (about 0.5 cell for QCHD and 5 cells for QHD+Penalty) and the values decrease as $(\Delta\theta)^2$.
We therefore use $t_\varepsilon$, $\langle d^2\rangle(t_{\rm ref})$, and $P_{\rm succ}$, the probability that $\theta$ (or $\varphi$) lies in the basin of the global minimum bounded by the adjacent maxima.
The reference time $t_{\rm ref}$ is the last time up to which $\langle d^2\rangle(t)$ and the density are converged under the grid refinements.
Halving the radial grid spacing changes $\langle d^2\rangle(t)$ for $\kappa=1$ by at most $0.6\%$ up to $t=10$ in Secs.~\ref{sec:num-sphere} and \ref{sec:num-ellipse} and up to $t=8$ in Sec.~\ref{sec:num-ackley}, but by $+10\%$ at $t=11$ and $+27\%$ at $t=12$ (sphere), $+5\%$ at $t=12$ (ellipse), and $+3$--$4\%$ at $t=9$ (Ackley), as the normal width approaches one radial cell.
Hence $t_{\rm ref}=10$, $10$, and $8$, while in Sec.~\ref{sec:num-twocons} the checks extend to $t_{\rm ref}=t_{\rm end}=40$.
The angular grid sets a similar limit.
The effective mass $ma(t)$ in Eq.~\eqref{eq:H-pen-num} makes a tangential velocity $\dot\theta$ correspond to the angular mode $n\simeq ma(t)R^2\dot\theta$, which for $a\simeq10^3$ ($t\simeq11.5$ with $\gamma=0.3$) and $\dot\theta\sim1$ requires $n\sim250$, beyond the Nyquist limit of $N_\theta=256$.
The thresholds in Secs.~\ref{sec:num-sphere} and \ref{sec:num-ellipse} are reached before this time, and $N_\theta=512$ was used in Sec.~\ref{sec:num-sphere} because $t_\varepsilon$ of $\kappa=1$ is close to the limit.
In Sec.~\ref{sec:num-ackley} the density flowing toward the origin for $\kappa\le2$ is no longer resolved by the radial grid after $t\simeq8$, which is why the indicators are evaluated at $t=8$.

\begin{table}[tbp]
  \centering
  {\footnotesize
  \begin{tabularx}{\textwidth}{lYYYY}
    \hline
    & $\Delta t/2$ & $N_\theta\times2$ & $\Delta r/2$ or domain & independent code \\
    \hline
    (1) sphere, quad. & $<10^{-9}$ ($\kappa=10$) & $t_\varepsilon$: 0.5\% & $t_\varepsilon$: $-0.9\%$, $\langle d^2\rangle(t_{\rm ref})$: $-0.6\%$ ($\kappa=1$) & $t_\varepsilon$: 0.5\% \\
    (2) sphere, Ackley & $P_{\rm succ}$, $t(E_{\rm proj}<0.3)$ unchanged ($\kappa=10$) & $P_{\rm succ}(8)$: $\le0.006$ ($\kappa=1,5$) & $P_{\rm succ}(8)$: $\le0.011$, $\langle d^2\rangle(8)$: $\le0.6\%$ ($\kappa=1,2,5$) & $P_{\rm succ}(8)$: 0.01, $\langle d^2\rangle(8)$: 1.5\%, $E_{\rm proj}(8)$: 1.2\% \\
    (3) ellipse, quad. & $<10^{-9}$ ($\kappa=10$) & --- & $t_\varepsilon$: $-0.3\%$, $\langle d^2\rangle(t_{\rm ref})$: $-0.5\%$ ($\kappa=1$) & $t_\varepsilon$: 0.5\% \\
    (4) two constr., quartic & $t_\varepsilon$, $\langle d^2\rangle(t_\varepsilon)$, $P_{\rm succ}$ unchanged ($\kappa=1$) & --- & domain: $t_\varepsilon$: $+1.2$, $+2.4$; $\langle d^2\rangle(t_\varepsilon)$: $+0.04$, $+0.01$ ($\kappa=1,2$) & --- \\
    \hline
  \end{tabularx}}
  \caption{Accuracy checks. Changes of the indicators under halving the time step ($\Delta t/2$), doubling the angular grid ($N_\theta$: $256\to512$), halving the radial grid spacing ($\Delta r/2$), and enlarging the computational domain (Sec.~\ref{sec:num-twocons} only), and the differences from an independent implementation. The values of $\kappa$ for which the check was made are given in parentheses.}
  \label{tab:num-checks}
\end{table}

\paragraph{Accuracy checks.}

Tab.~\ref{tab:num-checks} summarizes the changes of the indicators under the refinements of the time step and the grid and the differences from an independent implementation.
The time step was chosen such that an a posteriori estimate $\varepsilon_{\rm post}(t)$ of the local error of the Strang splitting, obtained from $|\nabla W|^2$ and $a(t)$, satisfies $\varepsilon_{\rm post}(t_{\rm end})\lesssim10^{-4}$ (the exceptions are $1.1$--$1.5\times10^{-4}$ in Sec.~\ref{sec:num-ackley}, where $\varepsilon_{\rm post}(8)\le2\times10^{-6}$ at the evaluation time, and $1.3\times10^{-4}$ for $\kappa=1,2$ in Sec.~\ref{sec:num-twocons}).
The independent implementation, written without access to our code, uses fourth-order finite differences on a nonuniform radial grid and a Pad\'e time stepping.
Both implementations find that the values at $t\gtrsim11$ with $N_\theta=256$ are resolution limited.
The $\vartheta$ resolution ($N_\vartheta=128$) and the half-time-step recomputation for $\kappa=2,5,10$ in Sec.~\ref{sec:num-twocons} have not been checked.

The computational domain requires care because the wave function of QHD+Penalty that reaches the boundary $r_{\max}$ is reflected, which accelerates the recapture and biases the comparison in favor of QHD+Penalty.
In Secs.~\ref{sec:num-sphere}--\ref{sec:num-ellipse} the two-stage domain described above is converged with respect to the domain, whereas a single domain $r\le2.8$ shortens $t_\varepsilon$ of QHD+Penalty by 4--8\%.
In Sec.~\ref{sec:num-twocons} the burst reaches the radial boundary at $t\simeq0.05$--$0.1$ for all $\kappa$, and a domain-converged calculation in three dimensions was beyond our resources.
We therefore estimated the systematic error from the change between the original domain $r\in[0.6,1.4]$ and the enlarged domain $r\in[0.4,1.8]$, computed for $\kappa=1,2$ (Tab.~\ref{tab:num-checks}).
The values quoted in the main text for $\kappa=1,2$ are those of the enlarged domain, for which 13\% of the mass still lies at $r>1.4$ at $t=t_{\rm end}$ for $\kappa=1$.
These changes are shown as the error bars in Fig.~\ref{fig:num-twocons-snap}, where for $\kappa=5,10$ the relative change of $\kappa=2$ ($+6\%$ in $t_\varepsilon$ and $+10\%$ in $\langle d^2\rangle$) is assumed.
The converged values are expected to lie in the direction of the bars, i.e., further in favor of QCHD.

\paragraph{Parameters.}
The parameters of each example are listed in Tab.~\ref{tab:num-params}.

\begin{table}[tbp]
  \centering
  {\footnotesize
  \begin{tabular}{lcccc}
    \hline
    & (1) sphere, quad. & (2) sphere, Ackley & (3) ellipse, quad. & (4) two constr., quartic \\
    \hline
    $N$, $M$ & 2, 1 & 2, 1 & 2, 1 & 3, 2 \\
    $\gamma$, $t_{\rm end}$ & 0.3, 20 & 0.3, 20 & 0.3, 20 & 0.05, 40 \\
    $\kappa$ & \multicolumn{4}{c}{1, 2, 5, 10} \\
    $\sigma_r$ & \multicolumn{4}{c}{0.05} \\
    $t_{\rm ref}$ & 10 & 8 & 10 & 40 \\
    angular grid & $N_\theta=512$ & $N_\theta=256$ & $N_\theta=256$ & $N_\varphi=192$, $N_\vartheta=128$ \\
    $N_r$ & \multicolumn{3}{c}{224 (448 for $t\le3$)} & 112 / 64 \\
    $r$ interval & \multicolumn{3}{c}{$[0,2.8]$ ($[0,5.6]$ for $t\le3$)} & $[0.4,1.8]$ / $[0.6,1.4]$ \\
    $N_t$ (QCHD) & $10^5$ & $2\times10^5$ & $10^5$ & $10^5$ \\
    $N_t$ (Penalty) & $3\times10^5$ & $6\times10^5$ & $3\times10^5$ & $(1.1, 1.5, 2.7, 4)\times10^5$ \\
    \hline
  \end{tabular}}
  \caption{Parameters of the numerical calculations. $N_t$ is the number of time steps ($\Delta t=t_{\rm end}/N_t$).
    For (4), $N_t$ is listed in the order $\kappa=1,2,5,10$, and the radial grid is given for $\kappa=1,2$ / $\kappa=5,10$.}
  \label{tab:num-params}
\end{table}

\newcommand{\arxivfont}{\rmfamily}
\bibliographystyle{yautphalpm}
\bibliography{QCHD-ref}

\end{document}